\documentclass[11pt,american]{article}
\usepackage[T1]{fontenc}
\usepackage[utf8]{inputenc}
\usepackage{color}
\usepackage{babel}
\usepackage{mathtools}
\usepackage{amsmath}
\usepackage{amsthm}
\usepackage{amssymb}
\usepackage{graphicx}
\usepackage{esint}
\usepackage[pdfusetitle,
 bookmarks=true,bookmarksnumbered=false,bookmarksopen=false,
 breaklinks=false,pdfborder={0 0 0},pdfborderstyle={},backref=false,colorlinks=true]
 {hyperref}
\hypersetup{
 linkcolor=magenta, urlcolor=blue, citecolor=blue}

\makeatletter
\theoremstyle{plain}
\newtheorem{thm}{\protect\theoremname}
\theoremstyle{plain}
\newtheorem{lem}[thm]{\protect\lemmaname}
\theoremstyle{remark}
\newtheorem{rem}[thm]{\protect\remarkname}
\theoremstyle{plain}
\newtheorem{lyxalgorithm}[thm]{\protect\algorithmname}
\theoremstyle{plain}
\newtheorem{prop}[thm]{\protect\propositionname}

\DeclareMathOperator{\Tr}{Tr}

\DeclareMathOperator{\sech}{sech}

\numberwithin{equation}{section}

\allowdisplaybreaks

\makeatother

\providecommand{\algorithmname}{Algorithm}
\providecommand{\lemmaname}{Lemma}
\providecommand{\propositionname}{Proposition}
\providecommand{\remarkname}{Remark}
\providecommand{\theoremname}{Theorem}

\begin{document}
\title{Fundamentals of quantum Boltzmann machine learning with visible and
hidden units}
\author{Mark M. Wilde\\
\textit{School of Electrical and Computer Engineering, Cornell University,}\\
\textit{Ithaca, New York 14850, USA}}
\date{\today}
\maketitle
\begin{abstract}
One of the primary applications of classical Boltzmann machines is
generative modeling, wherein the goal is to tune the parameters of
a model distribution so that it closely approximates a target distribution.
Training relies on estimating the gradient of the relative entropy
between the target and model distributions, a task that is well understood
when the classical Boltzmann machine has both visible and hidden units.
For some years now, it has been an obstacle to generalize this finding
to quantum state learning with quantum Boltzmann machines that have
both visible and hidden units. In this paper, I derive an analytical
expression for the gradient of the quantum relative entropy between
a target quantum state and the reduced state of the visible units
of a quantum Boltzmann machine. Crucially, this expression is amenable
to estimation on a quantum computer, as it involves modular-flow-generated
unitary rotations reminiscent of those appearing in my prior work
on rotated Petz recovery maps. This leads to a quantum algorithm for
gradient estimation in this setting. I then specialize the setting
to quantum visible units and classical hidden units, and vice versa;
I also provide analytical expressions for the gradients, along with quantum
algorithms for estimating them. Finally, I replace the quantum relative
entropy objective function with the Petz--Tsallis relative entropy;
here I develop an analytical expression for the gradient and sketch
a quantum algorithm for estimating it, as an application of an independent
derivation of a formula for the derivative of the matrix power function,
which also involves modular-flow-generated unitary rotations. Ultimately,
this paper demarcates progress in training quantum Boltzmann machines
with visible and hidden units for generative modeling and quantum
state learning.
\end{abstract}
\tableofcontents{}

\section{Introduction}

\subsection{Background and motivation}

Boltzmann machines constitute one of the earliest neural-network based
approaches to generative modeling \cite{Ackley1985,Hinton1986}, in
which the goal is to simulate a target probability distribution $q(v)$
by means of a parameterized model distribution $p_{\theta}(v)$, where
$\theta$ is a parameter vector. The traditional approach to Boltzmann
machine learning is to employ a classical Hamiltonian as an energy
function, with some of the variables called visible and the others
called hidden (often referred to as visible and hidden units). The
probability $p_{\theta}(v,h)$ for a particular configuration of visible
and hidden variables to occur is then proportional to a Boltzmann
weight, establishing a strong link between Boltzmann machine learning
and statistical physics. The marginal distribution $p_{\theta}(v)$
on the visible variables should be tuned to closely approximate the
target probability distribution, while the role of the hidden variables
is to mediate complex correlations between the visible variables by
means of simple interactions between the visible and hidden variables.
This is similar to how the Hamiltonian of mean force represents complex
interactions between system variables in classical thermodynamics
\cite{Jarzynski2004,Talkner2020}, which are mediated by the interactions
between a thermodynamic system and an external reservoir.

The conventional approach to generative modeling with classical Boltzmann
machines is to tune $\theta$ so that the relative entropy $D(q\|p_{\theta})$
between the target distribution $q(v)$ and the model distribution
$p_{\theta}(v)$ is minimized. In order to do so, it is helpful to
have a formula for the elements of the gradient $\nabla_{\theta}D(q\|p_{\theta})$,
which is well known to be given by \cite{Ackley1985}
\begin{equation}
\frac{\partial}{\partial\theta_{j}}D(q\|p_{\theta})=\sum_{v,h}q(v)p_{\theta}(h|v)G_{j}(v,h)-\sum_{v,h}p_{\theta}(v,h)G_{j}(v,h),\label{eq:gradient-rel-ent-gen-mod}
\end{equation}
when the classical Hamiltonian has the form $G_{\theta}(v,h)\coloneqq\sum_{j}\theta_{j}G_{j}(v,h)$.
As such, one can estimate the elements of the gradient by means of
Monte Carlo sampling, i.e., sampling from $q(v)p_{\theta}(h|v)$ and
$p_{\theta}(v,h)$. In conjunction, one can employ a standard stochastic
gradient descent algorithm in order to minimize $D(q\|p_{\theta})$.
While this optimization problem is known to be non-convex in general,
in practice, a local minimum of this objective function often works
well. After tuning $\theta$ in this way, one can then sample from
the model distribution $p_{\theta}$ at will.

Roughly a decade ago now, quantum Boltzmann machines were proposed
for generative modeling \cite{Amin2018,Kieferova2017,Benedetti2017},
with the idea being that distributions arising from measuring quantum
states might be able to capture more complex correlations in target
probability distributions than would be possible with classical Boltzmann
machines alone. Beyond this application, quantum Boltzmann machines
can alternatively be used for the non-classical task of quantum state
learning \cite{Kieferova2017}, in which the goal is to model a target
quantum state $\rho$ by means of a parameterized state $\sigma(\theta)$.
It was realized in \cite{Amin2018,Kieferova2017,Benedetti2017} that
there are difficulties in arriving at simple expressions like that
in \eqref{eq:gradient-rel-ent-gen-mod} and procedures for estimating
it when trying to develop quantum extensions of generative modeling,
while using both visible and hidden variables. A notable exception
occurs for quantum state learning when using quantum Boltzmann machines
with visible units only: indeed, the authors of \cite{Kieferova2017}
established a simple expression for the gradient (see \cite[Eq.~(4)]{Kieferova2017}),
while the authors of \cite{Coopmans2024} proved that the corresponding
minimization problem is convex in the parameter vector $\theta$.
Thus, for this particular problem, a hybrid quantum--classical algorithm,
which uses a quantum computer for estimating the gradient and stochastic
gradient descent for optimization, is guaranteed to converge to a
global minimum with efficient sample complexity \cite[Theorem~1]{Coopmans2024}.

What has remained open since these prior contributions is to determine
an analytical expression for the gradient that is amenable to estimation
on a quantum computer when using quantum Boltzmann machines, with
both visible and hidden units, for generative modeling and quantum
state learning. Under the assumption that the visible units of the
quantum Boltzmann machine are classical while the hidden units are
quantum, there has been recent progress on developing algorithms for
optimizing them for generative modeling \cite{Patel2025,Demidik2025,Kimura2025,Vishnu2025,Demidik2025a,Wilde2025a}.
However, these works do not address the aforementioned open question
for the task of quantum state learning with both visible and hidden
units, for which a quantum Boltzmann machine with quantum visible
units is required.

\subsection{Summary of contributions}

In this paper, I address this open question by determining an analytical
expression for the gradient in quantum state learning, which is amenable
to estimation on a quantum computer. In more detail, suppose that
$\rho$ is a target quantum state and $\sigma(\theta)$ is a model
quantum state, based on a quantum Boltzmann machine with visible and
hidden units, and the goal is to minimize the quantum relative entropy
$D(\rho\|\sigma(\theta))$. Here I establish an analytical expression
for the gradient $\nabla_{\theta}D(\rho\|\sigma(\theta))$, which
can be viewed as a quantum generalization of \eqref{eq:gradient-rel-ent-gen-mod},
incorporating aspects to account for the non-commutativity inherent
in quantum mechanics (see Theorem~\ref{thm:q-state-learning-gradient-vh-q}).
Furthermore, I develop a quantum algorithm for estimating the gradient
when given sample access to $\rho$ and $\sigma(\theta)$, implying
that one can find a local minimum of $D(\rho\|\sigma(\theta))$ by
means of a hybrid quantum--classical algorithm that uses a quantum
computer for estimating the gradient and a classical computer for
performing the stochastic gradient descent optimization algorithm.
A subroutine of this quantum algorithm involves simulating what is
called modular flow, and by invoking \cite[Lemma~3]{Qiu2025} and
\cite[Corollary~60 and Lemma~61]{Gilyen2019}, I establish an improvement
of the recent quantum algorithm from \cite{Lim2025} for simulating
modular flow.

The main distinction between the quantum expression reported in Theorem~\ref{thm:q-state-learning-gradient-vh-q}
and the classical one in \eqref{eq:gradient-rel-ent-gen-mod} is that
the conditional probability distribution $p_{\theta}(h|v)$ in \eqref{eq:gradient-rel-ent-gen-mod}
is replaced with a Hermiticity-preserving, trace-preserving map $\Sigma_{v\to vh}^{\theta}$
that is reminiscent of, yet different from, the rotated Petz recovery
map of \cite{Wilde2015}, the twirled Petz recovery map of \cite{Junge2018},
and the state-over-time of \cite{Fullwood2022} (see also \cite{Lie2024}).
In the case that both the visible and hidden units are classical,
the expression reduces to the classical expression in \eqref{eq:gradient-rel-ent-gen-mod}.

I also consider the case of quantum visible units and classical hidden
units, a case of interest and possibly easier to implement in practice
than fully quantum Boltzmann machines. For this scenario, I report
an analytical expression for the gradient (see Theorem~\ref{thm:qc-vh-gradient}),
along with a quantum algorithm for estimating it. The main distinction
between this quantum--classical case and the fully classical case
is that the conditional probability distribution $p_{\theta}(h|v)$
in \eqref{eq:gradient-rel-ent-gen-mod} is replaced with the Hermiticity-preserving,
trace-preserving map $\sum_{x}p_{x}(\theta)\Sigma_{v}^{\theta,x}\otimes|x\rangle\!\langle x|_{h}$,
as defined in Theorem~\ref{thm:qc-vh-gradient}. This map is reminiscent
of, yet different from, the rotated pretty good instrument introduced
in \cite[Remark~5.7]{Wilde2015}.

Finally, I extend all of the results to apply to the case when the
objective function is the Petz--Tsallis relative entropy of order
$q\in\left(0,1\right)\cup\left(1,2\right]$, instead of the quantum
relative entropy. Doing so involves developing an independent derivation
of a formula for the matrix derivative of the power function, as given
in \eqref{eq:deriv-matrix-power-fourier} of Lemma \ref{lem:derivative-matrix-power}
below. This formula could find alternative applications in quantum
information theory, and it is connected to the probability density
functions that appeared in \cite[Lemma~3.2]{Junge2018}, the latter
building on Hirschman's improvement of Hadamard's three-line theorem
\cite{Hirschman1952}. 

\subsection{Paper organization}

The rest of this paper is organized as follows:
\begin{itemize}
\item Section~\ref{sec:Preliminaries} provides some preliminary material
used throughout the paper. This includes a brief subsection on notation
(Section~\ref{subsec:Notation}) and a review of classical Boltzmann
machines and the derivation of the gradient of the relative entropy
(Section~\ref{subsec:Review-of-classical-BMs}). Also, Section~\ref{subsec:Formulas-for-matrix-derivs}
presents various formulas for matrix derivatives, including an independent
derivation of a formula for the matrix derivative of the power function
(see \eqref{eq:deriv-matrix-power-fourier} in Lemma \ref{lem:derivative-matrix-power}).
\item Section~\ref{sec:Quantum-state-learning-QBMs} contains some of the
main aforementioned results, including an analytical expression for
the gradient of the quantum relative entropy between a target state
$\rho$ and the reduced state $\sigma_{v}(\theta)$ of the visible
units of a quantum Boltzmann machine (Theorem~\ref{thm:q-state-learning-gradient-vh-q}),
as well as a quantum algorithm for estimating it (Section~\ref{subsec:Q-algorithm-state-learning-QBM}).
As noted above, a subroutine of this quantum algorithm simulates modular
flow, and I establish an improvement of this subroutine in Appendix~\ref{subsec:Block-encoding-for-modular-flow},
by invoking \cite[Lemma~3]{Qiu2025} and \cite[Corollary~60 and Lemma~61]{Gilyen2019}.
I also evaluate the gradient expression for a model called restricted
quantum Boltzmann machines (Section~\ref{subsec:Application-to-restricted-QBMs}),
and then I show how it reduces to the known expression for the gradient
\cite[Eq.~(4)]{Kieferova2017} when there are no hidden units (Section~\ref{subsec:Consistency-check-no-HUs}).
\item Section~\ref{sec:Quantum-state-learning-qc-QBMs} specializes the
results of Section~\ref{sec:Quantum-state-learning-QBMs} to the
case when the visible units are quantum and the hidden units are classical.
In particular, Section~\ref{subsec:Analytical-formula-for-gradient-qc-QBMs}
provides a formula for the gradient of the quantum relative entropy
in this case, Section~\ref{subsec:Quantum-algorithm-qc-BM} details
a quantum algorithm for estimating the gradient, and Section~\ref{subsec:Application-to-restricted-qc-QBMs}
evaluates the formula for restricted quantum--classical Boltzmann
machines.
\item Section~\ref{sec:Generative-modeling-using-cq-BMs} specializes the
results of Section~\ref{sec:Quantum-state-learning-QBMs} to the
case when the visible units are classical and the hidden units are
quantum. The findings here are related to the recent findings of \cite{Demidik2025},
but it is instructive to see how this reduction follows from the fully
quantum case presented in Section~\ref{sec:Quantum-state-learning-QBMs}.
\item The final technical section of this paper is Section~\ref{sec:Extensions-to-Petz=002013Tsallis},
in which I show how to generalize the findings in the whole paper
to the case when the objective function is the Petz--Tsallis relative
entropy with parameter $q\in\left(0,1\right)\cup\left(1,2\right]$.
To derive the results here, I provide an independent derivation of
a formula for the derivative of the matrix power function (see \eqref{eq:deriv-matrix-power-fourier}
of Lemma~\ref{lem:derivative-matrix-power}, with its proof in Appendix~\ref{app:Proof-of-Equation-matrix-power}).
\item In the conclusion (Section~\ref{sec:Conclusion}), I provide a brief
summary of the main findings of this paper, followed by suggestions
for future research.
\end{itemize}

\section{Preliminaries}

\label{sec:Preliminaries}

\subsection{Notation}

\label{subsec:Notation}Here I briefly establish some notation used
in the rest of the paper. For $d\in\mathbb{N}$, let
\begin{equation}
\left[d\right]\equiv\left\{ 1,\ldots,d\right\} .
\end{equation}
For operators $A$ and $B$, the anticommutator is defined as $\left\{ A,B\right\} \coloneqq AB+BA$.
I also use the following standard physics notation:
\begin{equation}
\left\langle A\right\rangle _{\rho}\equiv\Tr[A\rho].
\end{equation}
The spectral norm of an operator $A$ is denoted by $\left\Vert A\right\Vert \coloneqq\sup_{|\psi\rangle\neq0}\frac{\left\Vert A|\psi\rangle\right\Vert }{\left\Vert |\psi\rangle\right\Vert }$,
and the trace norm is denoted by $\left\Vert A\right\Vert _{1}\coloneqq\Tr[\sqrt{A^{\dag}A}]$.

\subsection{Review of classical Boltzmann machines}

\label{subsec:Review-of-classical-BMs}In this background section,
I review basic proofs in the theory of generative modeling for classical
Boltzmann machines. Although these proofs are well known by now \cite{Ackley1985},
they serve as a starting point for the more involved derivations for
quantum state learning using quantum Boltzmann machines.

Let $d_{v},d_{h}\in\mathbb{N}$, $v\in\left[d_{v}\right]$, $h\in\left[d_{h}\right]$,
and $J\in\mathbb{N}$. The variable $v$ is called visible, and the
variable $h$ is called hidden. For each $j\in\left[J\right]$, let
$G_{j}\colon\left[d_{v}\right]\times\left[d_{h}\right]\to\mathbb{R}$
denote a classical Hamiltonian for a discrete system of dimension
$d_{v}\times d{}_{h}$, also called an energy function. For a parameter
vector $\theta\coloneqq\left(\theta_{1},\ldots,\theta_{J}\right)\in\mathbb{R}^{J}$,
let
\begin{equation}
G_{\theta}(v,h)\coloneqq\sum_{j=1}^{J}\theta_{j}G_{j}(v,h)\label{eq:Hamiltonian-classical-BM}
\end{equation}
denote a parameterized Hamiltonian. A classical Boltzmann machine
corresponds to the following thermal probability distribution over
$G_{\theta}(v,h)$:
\begin{align}
p_{\theta}(v,h) & \coloneqq\frac{e^{-G_{\theta}(v,h)}}{Z_{\theta}},\\
Z_{\theta} & \coloneqq\sum_{v,h}e^{-G_{\theta}(v,h)}.
\end{align}
The marginal probability distribution on the visible units is as follows:
\begin{equation}
p_{\theta}(v)\coloneqq\sum_{h}p_{\theta}(v,h).
\end{equation}

Let $q(v)$ denote a target probability distribution over $v\in\left[d_{v}\right]$.
The goal of generative modeling is to take samples from $q(v)$ and
try to tune the parameter vector $\theta$ such that $p_{\theta}(v)$
closely approximates $q(v)$. The standard measure of closeness to
consider in this context is the classical relative entropy \cite{Kullback1951}:
\begin{equation}
D(q\|p_{\theta})\coloneqq\sum_{v}q(v)\ln\!\left(\frac{q(v)}{p_{\theta}(v)}\right),\label{eq:classical-BM-objective}
\end{equation}
and the goal in generative modeling is to minimize $D(q\|p_{\theta})$
with respect to the parameter vector $\theta$.

Training relies on calculating the gradient of \eqref{eq:classical-BM-objective}.
Each element of the gradient, for $j\in\left[J\right]$, is well known
to be as follows \cite{Ackley1985}:
\begin{equation}
\frac{\partial}{\partial\theta_{j}}D(q(v)\|p_{\theta}(v))=\sum_{v,h}q(v)p_{\theta}(h|v)G_{j}(v,h)-\sum_{v,h}p_{\theta}(v,h)G_{j}(v,h).\label{eq:classical-gradient}
\end{equation}
The utility of this formula is that one can estimate each element
of the gradient in an unbiased way by means of Monte Carlo sampling.
The first term in \eqref{eq:classical-gradient} can be estimated
by repeatedly taking a sample $\left(v,h\right)$ from $q(v)p_{\theta}(h|v)$,
calculating $G_{j}(v,h)$, and then, after all of the samples have
been collected, calculating the sample mean. The second term in \eqref{eq:classical-gradient}
can be estimated in a similar way by instead sampling from $p_{\theta}(v,h)$.

It is worthwhile to revisit the derivation of \eqref{eq:classical-gradient}
in detail. Given that one of the main contributions of this paper
is to derive nontrivial quantum generalizations of it, comparing the
derivations can be insightful. With this in mind, consider that
\begin{align}
 & \frac{\partial}{\partial\theta_{j}}D(q(v)\|p_{\theta}(v))\nonumber \\
 & =\frac{\partial}{\partial\theta_{j}}\left(\sum_{v}q(v)\ln q(v)-\sum_{v}q(v)\ln p_{\theta}(v)\right)\\
 & =-\sum_{v}q(v)\frac{\partial}{\partial\theta_{j}}\ln p_{\theta}(v)\\
 & =-\sum_{v}\frac{q(v)}{p_{\theta}(v)}\frac{\partial}{\partial\theta_{j}}p_{\theta}(v)\\
 & =-\sum_{v}\frac{q(v)}{p_{\theta}(v)}\frac{\partial}{\partial\theta_{j}}\left(\sum_{h}p_{\theta}(v,h)\right)\\
 & =-\sum_{v}\frac{q(v)}{p_{\theta}(v)}\sum_{h}\frac{\partial}{\partial\theta_{j}}\left(\frac{e^{-G_{\theta}(v,h)}}{Z_{\theta}}\right)\\
 & =-\sum_{v}\frac{q(v)}{p_{\theta}(v)}\sum_{h}\left[-\frac{e^{-G_{\theta}(v,h)}}{Z_{\theta}}\left(\frac{\partial}{\partial\theta_{j}}G_{\theta}(v,h)\right)-\frac{e^{-G_{\theta}(v,h)}}{Z_{\theta}^{2}}\frac{\partial}{\partial\theta_{j}}Z_{\theta}\right]\\
 & =\sum_{v}\frac{q(v)}{p_{\theta}(v)}\sum_{h}\left[p_{\theta}(v,h)G_{j}(v,h)-p_{\theta}(v,h)\frac{1}{Z_{\theta}}\sum_{v',h'}e^{-G_{\theta}(v',h')}G_{j}(v',h')\right]\\
 & =\sum_{v}\frac{q(v)}{p_{\theta}(v)}\sum_{h}\left[p_{\theta}(v,h)G_{j}(v,h)-p_{\theta}(v,h)\sum_{v',h'}p_{\theta}(v',h')G_{j}(v',h')\right]\\
 & =\sum_{v,h}q(v)\left[\frac{p_{\theta}(v,h)}{p_{\theta}(v)}G_{j}(v,h)-\frac{p_{\theta}(v,h)}{p_{\theta}(v)}\sum_{v',h'}p_{\theta}(v',h')G_{j}(v',h')\right]\\
 & =\sum_{v,h}q(v)p_{\theta}(h|v)G_{j}(v,h)-\sum_{v,h}q(v)\frac{p_{\theta}(v,h)}{p_{\theta}(v)}\sum_{v',h'}p_{\theta}(v',h')G_{j}(v',h')\\
 & =\sum_{v,h}q(v)p_{\theta}(h|v)G_{j}(v,h)-\sum_{v',h'}p_{\theta}(v',h')G_{j}(v',h'),
\end{align}
thus completing the proof of \eqref{eq:classical-gradient} after
substituting $\left(v',h'\right)\to\left(v,h\right)$.

A particular form for a classical Boltzmann machine, which has found
extensive use in applications, is known as a restricted Boltzmann
machine \cite{Hinton2002}, in which the form of the Hamiltonian in
\eqref{eq:Hamiltonian-classical-BM} is more restricted. Let $m,n\in\mathbb{N}$.
Let $\theta\equiv\left(a,b,w\right)$, where $a\in\mathbb{R}^{m}$,
$b\in\mathbb{R}^{n}$, and $w\in\mathbb{R}^{m\times n}$. Now suppose
that $v\in\left[d_{v}\right]^{m}$ and $h\in\left[d_{h}\right]^{n}$.
Then the Hamiltonian $G(\theta)$ for a restricted Boltzmann machine
is defined as
\begin{equation}
G_{\theta}(v,h)\coloneqq a^{T}v+b^{T}h+v^{T}wh.\label{eq:rQBM-def-1}
\end{equation}
Furthermore, the elements of the gradient of $D(q(v)\|p_{\theta}(v))$
are as follows:
\begin{align}
\nabla_{a}D(q(v)\|p_{\theta}(v)) & =\sum_{v,h}q(v)p_{\theta}(h|v)v-\sum_{v,h}p_{\theta}(v,h)v\label{eq:rBMs-gradient-1}\\
 & =\sum_{v}q(v)v-\sum_{v}p_{\theta}(v)v,\\
\nabla_{b}D(q(v)\|p_{\theta}(v)) & =\sum_{v,h}q(v)p_{\theta}(h|v)h-\sum_{v,h}p_{\theta}(v,h)h\\
 & =\sum_{h}\left[\sum_{v}q(v)p_{\theta}(h|v)\right]h-\sum_{h}p_{\theta}(h)h,\\
\nabla_{w}D(q(v)\|p_{\theta}(v)) & =\sum_{v,h}q(v)p_{\theta}(h|v)vh^{T}-\sum_{v,h}p_{\theta}(v,h)vh^{T}.\label{eq:rBMs-gradient-last}
\end{align}
In Sections \ref{subsec:Application-to-restricted-QBMs}, \ref{subsec:Application-to-restricted-qc-QBMs},
and \ref{subsec:Application-to-restricted-cq-BMs}, I generalize these
formulas to various quantum settings, including the fully quantum
(Section~\ref{subsec:Application-to-restricted-QBMs}), quantum--classical
(Section~\ref{subsec:Application-to-restricted-qc-QBMs}), and classical--quantum
(Section~\ref{subsec:Application-to-restricted-cq-BMs}) settings.

\subsection{Formulas for matrix derivatives}

\label{subsec:Formulas-for-matrix-derivs}

In order to derive analytical expressions for gradients in various
settings considered in this paper, I make use of formulas for matrix
derivatives. To begin with, let us recall the following formulas for
the derivative of the matrix exponential:
\begin{lem}[Derivative of matrix exponential]
\label{lem:derivative-matrix-exp}For $x\mapsto B(x)$ a Hermitian
operator-valued function,
\begin{align}
\frac{\partial}{\partial x}e^{B(x)} & =\int_{0}^{1}dt\,e^{tB(x)}\left(\frac{\partial}{\partial x}B(x)\right)e^{\left(1-t\right)B(x)},\label{eq:duhamel-form}\\
 & =\frac{1}{2}\left\{ \Phi_{B(x)}\!\left(\frac{\partial}{\partial x}B(x)\right),e^{B(x)}\right\} ,\label{eq:duhamel-fourier}
\end{align}
where the quantum channel $\Phi_{B(x)}$ and the high-peak tent probability
density $\gamma(t)$ are defined as
\begin{align}
\Phi_{B(x)}(Y) & \coloneqq\int_{-\infty}^{\infty}dt\,\gamma(t)\:e^{-iB(x)t}Ye^{iB(x)t},\\
\gamma(t) & \coloneqq\frac{2}{\pi}\ln\left|\coth\!\left(\frac{\pi t}{2}\right)\right|.\label{eq:high-peak-tent-def}
\end{align}
\end{lem}

\begin{proof}
Both equalities in \eqref{eq:duhamel-form} and \eqref{eq:duhamel-fourier}
are known, the first being known as Duhamel's formula. See, e.g.,
\cite[Proposition~47]{Wilde2025} for a proof of \eqref{eq:duhamel-form},
and see, e.g., \cite[Lemmas~10 and 12]{Patel2025} for a proof of
\eqref{eq:duhamel-fourier}, which includes a proof of the fact that
$\gamma(t)$ is a probability density function. See also \cite{Hastings2007,Kim2012,Ejima2019,Kato2019,Anshu2021}
for various developments related to \eqref{eq:duhamel-fourier}.
\end{proof}
\begin{rem}[Derivative of a thermal state]
Lemma \ref{lem:derivative-matrix-exp} implies the following formula
for the derivative of a thermal state $\sigma(x)$:
\begin{align}
\frac{\partial}{\partial x}\sigma(x) & =-\frac{1}{2}\left\{ \Phi_{B(x)}\!\left(\frac{\partial}{\partial x}B(x)\right),\sigma(x)\right\} +\sigma(x)\left\langle \frac{\partial}{\partial x}B(x)\right\rangle ,\label{eq:deriv-thermal-state}
\end{align}
where
\begin{align}
\sigma(x) & \coloneqq\frac{e^{-B(x)}}{Z(x)},\\
Z(x) & \coloneqq\Tr\!\left[e^{-B(x)}\right].
\end{align}
This formula is essential to several recent theoretical developments
regarding the training of quantum Boltzmann machines \cite{Patel2025,Patel2025a,Minervini2025,Wilde2025}.
\end{rem}

Let us now recall the following formulas for the derivative of the
matrix logarithm:
\begin{lem}[Derivative of matrix logarithm]
\label{lem:derivative-matrix-log}For $x\mapsto A(x)$ a positive
definite operator-valued function,
\begin{align}
\frac{\partial}{\partial x}\ln A(x) & =\int_{0}^{\infty}ds\,\left(A(x)+sI\right)^{-1}\left(\frac{\partial}{\partial x}A(x)\right)\left(A(x)+sI\right)^{-1}\label{eq:deriv-matrix-log-standard}\\
 & =A(x)^{-\frac{1}{2}}\Upsilon_{A(x)}\!\left(\frac{\partial}{\partial x}A(x)\right)A(x)^{-\frac{1}{2}},\label{eq:deriv-matrix-log-fourier}
\end{align}
where the quantum channel $\Upsilon_{A(x)}$ and the logistic probability
density function $\beta(t)$ with scale parameter $\frac{1}{\pi}$
are defined as
\begin{align}
\Upsilon_{A(x)}(Y) & \coloneqq\int_{-\infty}^{\infty}dt\,\beta(t)\,A(x){}^{-\frac{it}{2}}YA(x){}^{\frac{it}{2}},\\
\beta(t) & \coloneqq\frac{\pi}{2\left(\cosh\!\left(\pi t\right)+1\right)}=\frac{\pi}{4}\sech^{2}\!\left(\frac{\pi t}{2}\right).\label{eq:logistic-prob-dens}
\end{align}
\end{lem}

\begin{proof}
The equality in \eqref{eq:deriv-matrix-log-standard} is well known.
See, e.g., \cite[Proposition~48]{Wilde2025} for a proof of the equality
in \eqref{eq:deriv-matrix-log-standard}, and see \cite[Lemma~3.4]{Sutter2017}
and Remark~\ref{rem:div-diff-log-proof} for a proof of \eqref{eq:deriv-matrix-log-fourier}.
\end{proof}
Finally, let us state the following formulas for the derivative of
a matrix power, the second of which was independently derived in \cite{Beigi2025}: 
\begin{lem}[Derivative of matrix power]
\label{lem:derivative-matrix-power}For $x\mapsto A(x)$ a positive
definite operator-valued function and $r\in\left(-1,0\right)\cup\left(0,1\right)$,
\begin{align}
\frac{\partial}{\partial x}A(x)^{r} & =\frac{\sin(\pi r)}{\pi}\int_{0}^{\infty}ds\,s^{r}\left(A(x)+sI\right)^{-1}\left(\frac{\partial}{\partial x}A(x)\right)\left(A(x)+sI\right)^{-1},\label{eq:deriv-matrix-power-standard}\\
 & =rA(x)^{\frac{r-1}{2}}\Upsilon_{A(x)}^{r}\!\left(\frac{\partial}{\partial x}A(x)\right)A(x)^{\frac{r-1}{2}},\label{eq:deriv-matrix-power-fourier}
\end{align}
where the quantum channel $\Upsilon_{A(x)}^{r}$ and the probability
density function $\beta_{r}(t)$ are defined as
\begin{align}
\Upsilon_{A(x)}^{r}(Y) & \coloneqq\int_{-\infty}^{\infty}dt\,\beta_{r}(t)\,A(x){}^{-\frac{it}{2}}YA(x){}^{\frac{it}{2}},\\
\beta_{r}(t) & \coloneqq\frac{\sin(\pi r)}{2r\left(\cosh\!\left(\pi t\right)+\cos(\pi r)\right)}.\label{eq:logistic-prob-dens-1}
\end{align}
\end{lem}

\begin{proof}
See \cite[Proposition~50]{Wilde2025} for a proof of \eqref{eq:deriv-matrix-power-standard},
and see Appendix~\ref{app:Proof-of-Equation-matrix-power} for a
proof of \eqref{eq:deriv-matrix-power-fourier}.
\end{proof}
\begin{rem}
Note that the equality in \eqref{eq:deriv-matrix-log-fourier} is
a limiting case of the equality in \eqref{eq:deriv-matrix-power-fourier}
because
\begin{equation}
\lim_{r\to0}\frac{1}{r}\frac{\partial}{\partial x}A(x)^{r}=\frac{\partial}{\partial x}\ln A(x).\label{eq:r-0-limit-log}
\end{equation}
Indeed, the equality in \eqref{eq:r-0-limit-log} follows because
\begin{equation}
\frac{1}{r}\frac{\partial}{\partial x}A(x)^{r}=A(x)^{\frac{r-1}{2}}\Upsilon_{A(x)}^{r}\!\left(\frac{\partial}{\partial x}A(x)\right)A(x)^{\frac{r-1}{2}},
\end{equation}
so that
\begin{align}
\lim_{r\to0}\frac{1}{r}\frac{\partial}{\partial x}A(x)^{r} & =\lim_{r\to0}A(x)^{\frac{r-1}{2}}\Upsilon_{A(x)}^{r}\!\left(\frac{\partial}{\partial x}A(x)\right)A(x)^{\frac{r-1}{2}}\\
 & =A(x)^{-\frac{1}{2}}\Upsilon_{A(x)}\!\left(\frac{\partial}{\partial x}A(x)\right)A(x)^{-\frac{1}{2}},
\end{align}
where I used that
\begin{equation}
\lim_{r\to0}\beta_{r}(t)=\beta(t).
\end{equation}
\begin{rem}[Fr\'echet derivatives]
The various expressions in \eqref{eq:duhamel-form}, \eqref{eq:duhamel-fourier},
\eqref{eq:deriv-matrix-log-standard}, \eqref{eq:deriv-matrix-log-fourier},
\eqref{eq:deriv-matrix-power-standard}, and \eqref{eq:deriv-matrix-power-fourier}
can be expressed as the following Fr\'echet derivatives \cite{Coleman2012},
for all positive definite $A$, Hermitian $B$ and $H$, and $r\in\left(-1,0\right)\cup\left(0,1\right)$:
\begin{align}
D\exp(B)[H] & =\int_{0}^{1}dt\,e^{tB}He^{\left(1-t\right)B},\\
 & =\frac{1}{2}\left\{ \Phi_{B}\!\left(H\right),e^{B}\right\} ,\label{eq:frechet-exp}\\
D\ln(A)[H] & =\int_{0}^{\infty}ds\,\left(A+sI\right)^{-1}H\left(A+sI\right)^{-1}\\
 & =A^{-\frac{1}{2}}\Upsilon_{A}\!\left(H\right)A{}^{-\frac{1}{2}},\label{eq:frechet-log}\\
DA^{r}[H] & =\frac{\sin(\pi r)}{\pi}\int_{0}^{\infty}ds\,s^{r}\left(A+sI\right)^{-1}H\left(A+sI\right)^{-1},\\
 & =rA^{\frac{r-1}{2}}\Upsilon_{A}^{r}\!\left(H\right)A{}^{\frac{r-1}{2}}.\label{eq:frechet-power}
\end{align}
In the case that $H$ commutes with both $A$ and $B$, the above
expressions reduce to the conventional scalar derivatives. This is
especially transparent for the expressions in \eqref{eq:frechet-exp},
\eqref{eq:frechet-log}, and \eqref{eq:frechet-power}, given that
the channels $\Phi_{B}$, $\Upsilon_{A}$, and $\Upsilon_{A}^{r}$
each correspond to random-time Hamiltonian simulations that leave
$H$ invariant under the commuting assumption (that is, $\Upsilon_{A}(H)=H$
and $\Upsilon_{A}^{r}(H)=H$ if $\left[A,H\right]=0$ and $\Phi_{B}(H)=H$
if $\left[B,H\right]=0$).
\begin{rem}
\label{rem:constant-time-dists}Note that the probability densities
$\gamma(t)$, $\beta(t)$, and $\beta_{r}(t)$ have essentially all
of their probability mass concentrated on a constant-sized interval
containing $t=0$. To see this, consider that, for $T>\frac{\ln2}{\pi}\approx0.22$,
the probability mass of $\gamma(t)$ outside of the interval $\left[-T,T\right]$
is bounded from above as follows:
\begin{align}
\int_{\left|t\right|>T}dt\,\gamma(t) & =\frac{2}{\pi}\int_{\left|t\right|>T}dt\,\ln\!\left|\coth\!\left(\frac{\pi t}{2}\right)\right|.\\
 & =\frac{4}{\pi}\int_{T}^{\infty}dt\,\ln\!\left(\coth\!\left(\frac{\pi t}{2}\right)\right)\\
 & \leq\frac{4}{\pi}\int_{T}^{\infty}dt\,4e^{-\pi t}\\
 & =\frac{16}{\pi^{2}}e^{-\pi T},\label{eq:high-peak-tent-up-bnd}
\end{align}
where I used the facts that $\gamma(t)$ is an even function, $\ln(1+x)\leq x$
for $x>0$, $\coth\!\left(\frac{\pi t}{2}\right)=1+\frac{2}{e^{\pi t}-1}$,
and $e^{\pi t}-1\geq\frac{1}{2}e^{\pi t}$ for $t\geq\frac{\ln2}{\pi}$.
For example, at $T=10$, the upper bound in \eqref{eq:high-peak-tent-up-bnd}
evaluates to $\frac{16}{\pi^{2}}e^{-10\pi}\approx3.9\times10^{-14}$.
Similarly, the probability mass of $\beta(t)$ outside of the interval
$\left[-T,T\right]$ is bounded from above as follows:
\begin{align}
\int_{\left|t\right|>T}dt\,\beta(t) & =\frac{\pi}{2}\int_{\left|t\right|>T}dt\,\frac{1}{\cosh\!\left(\pi t\right)+1}.\\
 & =\pi\int_{T}^{\infty}dt\,\frac{1}{\cosh\!\left(\pi t\right)+1}\\
 & \leq2\pi\int_{T}^{\infty}dt\,e^{-\pi t}\\
 & =2e^{-\pi T},\label{eq:beta-t-up-bnd}
\end{align}
where I used the fact that $\beta(t)$ is an even function and that
$\cosh(x)+1\geq\frac{1}{2}e^{x}$ for $x\geq0$. At $T=10$, the upper
bound in \eqref{eq:beta-t-up-bnd} evaluates to $2e^{-10\pi}\approx4.5\times10^{-14}$.
A similar analysis and conclusion can be reached for $\beta_{r}(t)$. 
\end{rem}

\end{rem}

\end{rem}

\section{Quantum state learning using quantum Boltzmann machines}

\label{sec:Quantum-state-learning-QBMs}This section contains some
of the main claims of this paper, including an analytical formula
for the gradient of the quantum relative entropy in quantum state
learning, when using quantum Boltzmann machines with visible and hidden
units (Theorem~\ref{thm:q-state-learning-gradient-vh-q}). Additionally,
I present a quantum algorithm for estimating the gradient (Section~\ref{subsec:Q-algorithm-state-learning-QBM}).
I also show how the gradient formula simplifies when considering a
restricted quantum Boltzmann machine (Section~\ref{subsec:Application-to-restricted-QBMs})
and how it reduces to the known expression from \cite[Eq.~(4)]{Kieferova2017}
when there are only visible units (Section~\ref{subsec:Consistency-check-no-HUs}).
The results of this section also represent the starting point from
which the results in Sections \ref{sec:Quantum-state-learning-qc-QBMs}
and \ref{sec:Generative-modeling-using-cq-BMs} are derived.

To begin with, let us define a quantum Boltzmann machine (parameterized
thermal state) with visible and hidden units as follows:
\begin{align}
\sigma_{vh}(\theta) & \coloneqq\frac{e^{-G(\theta)}}{Z(\theta)},\label{eq:def-QBM-state-vh}\\
Z(\theta) & \coloneqq\Tr[e^{-G(\theta)}],\\
G(\theta) & \coloneqq\sum_{j=1}^{J}\theta_{j}G_{j},\label{eq:gen-param-ham}
\end{align}
where $J\in\mathbb{N}$, $\theta\coloneqq\left(\theta_{1},\ldots,\theta_{J}\right)\in\mathbb{R}^{J}$
is a parameter vector and each $G_{j}$ is a Hamiltonian acting on
both the visible and hidden systems. In the above, the letter $v$
denotes the visible system, and the letter $h$ denotes the hidden
system. The reduced state on the visible system $v$ is thus given
by
\begin{equation}
\sigma_{v}(\theta)\coloneqq\Tr_{h}[\sigma_{vh}(\theta)].\label{eq:def-reduced-state-vis}
\end{equation}

The goal of quantum state learning is to minimize the quantum relative
entropy between a target state $\rho$ and the reduced state $\sigma_{v}(\theta)$
of the visible system \cite{Kieferova2017}. Recall that the quantum
(Umegaki) relative entropy is defined as \cite{Umegaki1962}
\begin{equation}
D(\rho\|\sigma_{v}(\theta))\coloneqq\Tr[\rho\ln\rho]-\Tr[\rho\ln\sigma_{v}(\theta)].\label{eq:q-rel-ent-def}
\end{equation}

\subsection{Analytical formula for the gradient of quantum Boltzmann machines}

Theorem~\ref{thm:q-state-learning-gradient-vh-q} below provides
an analytical expression for the gradient of \eqref{eq:q-rel-ent-def},
which is a quantum generalization of the expression in \eqref{eq:classical-gradient}.
One key difference with the classical formula is that the first term
of \eqref{eq:gradient-fully-QBM} is expressed in terms of the Hermiticity-preserving,
trace-preserving (HPTP) map $\Sigma_{v\to vh}^{\theta}$. In general,
$\Sigma_{v\to vh}^{\theta}$ is not a quantum channel: if its input
is a quantum state, then its output is a quasi-state (Hermitian with
unit trace), as studied considerably in other contexts \cite{Fitzsimons2015,Fullwood2025,Fullwood2025a,Ji2025}.
However, if the Hamiltonian terms acting on the visible units commute,
then $\Sigma_{v\to vh}^{\theta}$ reduces to a classical--quantum
channel, as observed later on in Theorem~\ref{thm:gradient-cq-BMs}.
\begin{thm}
\label{thm:q-state-learning-gradient-vh-q}Let $\rho$ be a target
quantum state, and let $\sigma_{v}(\theta)$ be the state in \eqref{eq:def-reduced-state-vis}.
The partial derivatives of $D(\rho\|\sigma_{v}(\theta))$ are as follows:
\begin{equation}
\frac{\partial}{\partial\theta_{j}}D(\rho\|\sigma_{v}(\theta))=\left\langle G_{j}\right\rangle _{\Sigma_{v\to vh}^{\theta}(\rho)}-\left\langle G_{j}\right\rangle _{\sigma_{vh}(\theta)},\label{eq:gradient-fully-QBM}
\end{equation}
where 
\begin{align}
\Sigma_{v\to vh}^{\theta} & \coloneqq\Phi_{vh}^{\theta}\circ\Xi_{v\to vh}^{\theta}\circ\Upsilon_{v}^{\theta},\label{eq:herm-pres-map-Sigma}\\
\Xi_{v\to vh}^{\theta}(R_{v}) & \coloneqq\frac{1}{2}\left\{ \sigma_{vh}(\theta),\sigma_{v}(\theta)^{-\frac{1}{2}}R_{v}\sigma_{v}(\theta)^{-\frac{1}{2}}\otimes I_{h}\right\} ,\\
\Phi_{vh}^{\theta}(Y_{vh}) & \coloneqq\int_{-\infty}^{\infty}dt\,\gamma(t)\:e^{-iG(\theta)t}Y_{vh}e^{iG(\theta)t},\\
\Upsilon_{v}^{\theta}(X_{v}) & \coloneqq\int_{-\infty}^{\infty}dt\,\beta(t)\,\sigma_{v}(\theta)^{-\frac{it}{2}}X_{v}\sigma_{v}(\theta)^{\frac{it}{2}}. \label{eq:upsilon-channel}
\end{align}
In the above, $\Xi_{v\to vh}^{\theta}$ is a Hermiticity-preserving,
trace-preserving superoperator, $\Phi_{vh}^{\theta}$ is a quantum
channel, $\Upsilon_{v}^{\theta}$ is a quantum channel, $\gamma(t)$
is the high-peak tent probability density function defined in \eqref{eq:high-peak-tent-def},
and $\beta(t)$ is the logistic probability density function defined
in \eqref{eq:logistic-prob-dens}.
\end{thm}

\begin{proof}
Consider that
\begin{align}
 & \frac{\partial}{\partial\theta_{j}}D(\rho\|\sigma_{v}(\theta))\nonumber \\
 & =-\Tr\!\left[\rho\frac{\partial}{\partial\theta_{j}}\ln\sigma_{v}(\theta)\right]\\
 & \overset{(a)}{=}-\Tr\!\left[\rho\sigma_{v}(\theta)^{-\frac{1}{2}}\Upsilon_{v}^{\theta}\!\left(\frac{\partial}{\partial\theta_{j}}\sigma_{v}(\theta)\right)\sigma_{v}(\theta)^{-\frac{1}{2}}\right]\\
 & \overset{(b)}{=}-\Tr\!\left[\sigma_{v}(\theta)^{-\frac{1}{2}}\Upsilon_{v}^{\theta}(\rho)\sigma_{v}(\theta)^{-\frac{1}{2}}\left(\frac{\partial}{\partial\theta_{j}}\Tr_{h}\!\left[\sigma_{vh}(\theta)\right]\right)\right]\\
 & \overset{(c)}{=}-\Tr\!\left[\sigma_{v}(\theta)^{-\frac{1}{2}}\Upsilon_{v}^{\theta}(\rho)\sigma_{v}(\theta)^{-\frac{1}{2}}\Tr_{h}\!\left[\frac{\partial}{\partial\theta_{j}}\sigma_{vh}(\theta)\right]\right]\\
 & \overset{(d)}{=}-\Tr\!\left[\begin{array}{c}
\sigma_{v}(\theta)^{-\frac{1}{2}}\Upsilon_{v}^{\theta}(\rho)\sigma_{v}(\theta)^{-\frac{1}{2}}\times\\
\Tr_{h}\!\left[-\frac{1}{2}\left\{ \Phi_{vh}^{\theta}(G_{j}),\sigma_{vh}(\theta)\right\} +\sigma_{vh}(\theta)\left\langle G_{j}\right\rangle _{\sigma_{vh}(\theta)}\right]
\end{array}\right]\\
 & =\frac{1}{2}\Tr\!\left[\sigma_{v}(\theta)^{-\frac{1}{2}}\Upsilon_{v}^{\theta}(\rho)\sigma_{v}(\theta)^{-\frac{1}{2}}\Tr_{h}\!\left[\left\{ \Phi_{vh}^{\theta}(G_{j}),\sigma_{vh}(\theta)\right\} \right]\right]\nonumber \\
 & \qquad-\left\langle G_{j}\right\rangle _{\sigma_{vh}(\theta)}\Tr\!\left[\sigma_{v}(\theta)^{-\frac{1}{2}}\Upsilon_{v}^{\theta}(\rho)\sigma_{v}(\theta)^{-\frac{1}{2}}\Tr_{h}\!\left[\sigma_{vh}(\theta)\right]\right]\\
 & \overset{(e)}{=}\frac{1}{2}\Tr\!\left[\left(\sigma_{v}(\theta)^{-\frac{1}{2}}\Upsilon_{v}^{\theta}(\rho)\sigma_{v}(\theta)^{-\frac{1}{2}}\otimes I_{h}\right)\left\{ \Phi_{vh}^{\theta}(G_{j}),\sigma_{vh}(\theta)\right\} \right]\nonumber \\
 & \qquad-\left\langle G_{j}\right\rangle _{\sigma_{vh}(\theta)}\Tr\!\left[\sigma_{v}(\theta)^{-\frac{1}{2}}\Upsilon_{v}^{\theta}(\rho)\sigma_{v}(\theta)^{-\frac{1}{2}}\sigma_{v}(\theta)\right]\\
 & \overset{(f)}{=}\frac{1}{2}\Tr\!\left[\left(\sigma_{v}(\theta)^{-\frac{1}{2}}\Upsilon_{v}^{\theta}(\rho)\sigma_{v}(\theta)^{-\frac{1}{2}}\otimes I_{h}\right)\left\{ \Phi_{vh}^{\theta}(G_{j}),\sigma_{vh}(\theta)\right\} \right]-\left\langle G_{j}\right\rangle _{\sigma_{vh}(\theta)}.
\end{align}
The equality $(a)$ follows from Lemma \ref{lem:derivative-matrix-log}.
The equality $(b)$ follows from \eqref{eq:def-reduced-state-vis}
and the facts that $\Tr[A\Upsilon_{v}^{\theta}(B)]=\Tr[\Upsilon_{v}^{\theta}(A)B]$
for all linear operators $A$ and $B$. It also follows because
\begin{equation}
\Upsilon_{v}^{\theta}(\sigma_{v}(\theta)^{-\frac{1}{2}}(\cdot)\sigma_{v}(\theta)^{-\frac{1}{2}})=\sigma_{v}(\theta)^{-\frac{1}{2}}\Upsilon_{v}^{\theta}(\cdot)\sigma_{v}(\theta)^{-\frac{1}{2}}.
\end{equation}
The equality $(c)$ follows because the partial derivative and the
partial trace are linear operations. The equality $(d)$ follows from
\eqref{eq:deriv-thermal-state}. The equality $(e)$ follows because
\begin{equation}
\Tr\!\left[\left(A_{v}\otimes I_{h}\right)B_{vh}\right]=\Tr\!\left[A_{v}\Tr_{h}\!\left[B_{vh}\right]\right].
\end{equation}
The equality $(f)$ follows because 
\begin{equation}
\Tr\!\left[\sigma_{v}(\theta)^{-\frac{1}{2}}\Upsilon_{v}^{\theta}(\rho)\sigma_{v}(\theta)^{-\frac{1}{2}}\sigma_{v}(\theta)\right]=1,
\end{equation}
 given that $\rho$ is a state and $\Upsilon_{v}^{\theta}$ is a channel.
Now consider that
\begin{align}
 & \frac{1}{2}\Tr\!\left[\left(\sigma_{v}(\theta)^{-\frac{1}{2}}\Upsilon_{v}^{\theta}(\rho)\sigma_{v}(\theta)^{-\frac{1}{2}}\otimes I_{h}\right)\left\{ \Phi_{vh}^{\theta}(G_{j}),\sigma_{vh}(\theta)\right\} \right]\nonumber \\
 & =\frac{1}{2}\Tr\!\left[\left(\sigma_{v}(\theta)^{-\frac{1}{2}}\Upsilon_{v}^{\theta}(\rho)\sigma_{v}(\theta)^{-\frac{1}{2}}\otimes I_{h}\right)\Phi_{vh}^{\theta}(G_{j})\sigma_{vh}(\theta)\right]\nonumber \\
 & \qquad+\frac{1}{2}\Tr\!\left[\left(\sigma_{v}(\theta)^{-\frac{1}{2}}\Upsilon_{v}^{\theta}(\rho)\sigma_{v}(\theta)^{-\frac{1}{2}}\otimes I_{h}\right)\sigma_{vh}(\theta)\Phi_{vh}^{\theta}(G_{j})\right]\\
 & =\Tr\!\left[\Phi_{vh}^{\theta}(G_{j})\frac{1}{2}\left\{ \sigma_{vh}(\theta),\sigma_{v}(\theta)^{-\frac{1}{2}}\Upsilon_{v}^{\theta}(\rho)\sigma_{v}(\theta)^{-\frac{1}{2}}\otimes I_{h}\right\} \right]\\
 & \overset{(g)}{=}\Tr\!\left[G_{j}\Phi_{vh}^{\theta}\!\left(\frac{1}{2}\left\{ \sigma_{vh}(\theta),\sigma_{v}(\theta)^{-\frac{1}{2}}\Upsilon_{v}^{\theta}(\rho)\sigma_{v}(\theta)^{-\frac{1}{2}}\otimes I_{h}\right\} \right)\right]\\
 & =\Tr\!\left[G_{j}(\Phi_{vh}^{\theta}\circ\Xi_{v\to vh}^{\theta}\circ\Upsilon_{v}^{\theta})(\rho)\right]\\
 & \overset{(h)}{=}\Tr\!\left[G_{j}\Sigma_{v\to vh}^{\theta}(\rho)\right]\\
 & =\left\langle G_{j}\right\rangle _{\Sigma_{v\to vh}^{\theta}(\rho)},
\end{align}
where the equality $(g)$ follows because $\Tr[A\Phi_{vh}^{\theta}(B)]=\Tr[\Phi_{vh}^{\theta}(A)B]$
for all linear operators $A$ and $B$ and the equality $(h)$ follows
from the definition of the Hermiticity-preserving map $\Sigma_{v\to vh}^{\theta}$
in \eqref{eq:herm-pres-map-Sigma}.
\end{proof}

\subsection{Quantum algorithm for estimating the gradient}

\label{subsec:Q-algorithm-state-learning-QBM}

\begin{figure}

\centering{}\includegraphics[width=1\textwidth]{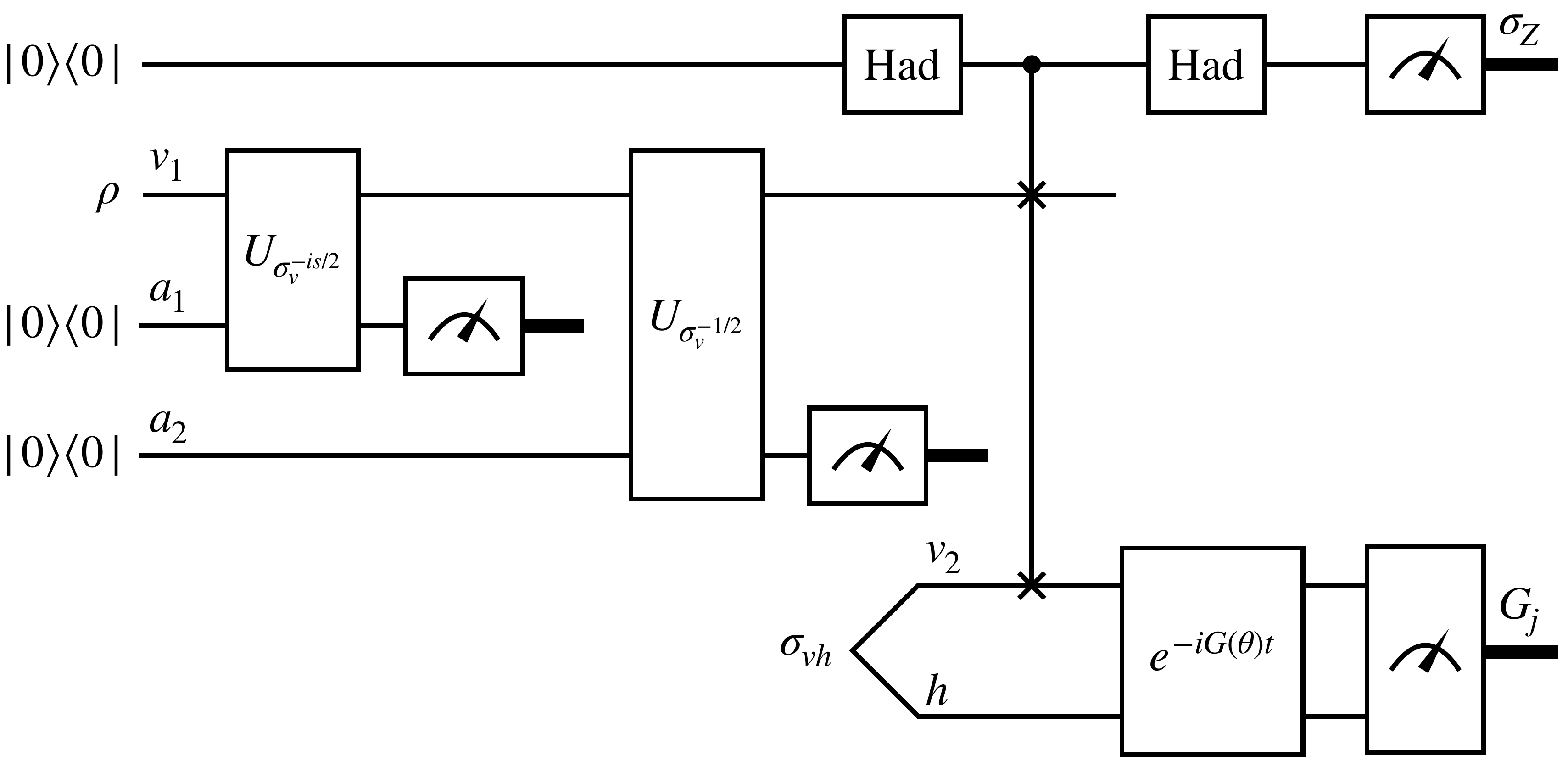}\caption{Depiction of a quantum circuit that estimates $\left\langle G_{j}\right\rangle _{\Sigma_{v\to vh}^{\theta}\!\left(\rho\right)}$,
the first term in \eqref{eq:gradient-fully-QBM}. The first part of
the circuit prepares a block-encoding of $\sigma_{v}^{-1/2}\sigma_{v}^{-is/2}$,
which acts on the target state $\rho$. The last part of the circuit
performs a swap test and measures the observable $e^{iG(\theta)t}G_{j}e^{-iG(\theta)t}$.
In each execution of the circuit, the value $s$ is sampled from the
logistic probability density $\beta(t)$ in \eqref{eq:logistic-prob-dens},
and the value $t$ is sampled from the high-peak tent probability
density $\gamma(t)$ in \eqref{eq:high-peak-tent-def}.}\label{fig:q-circuit-gradient-est}
\end{figure}

In this section, I present a quantum algorithm for estimating the
term $\left\langle G_{j}\right\rangle _{\Sigma_{v\to vh}^{\theta}(\rho)}$
in \eqref{eq:gradient-fully-QBM}. The key quantum circuit used for
this purpose is depicted in Figure~\ref{fig:q-circuit-gradient-est}.
It relies on quantum singular value transformation (QSVT) \cite{Gilyen2019},
which is a general-purpose framework for applying functions to matrices
block-encoded into unitary quantum circuits. It also relies on the
Hadamard test \cite{Cleve1998} and Hamiltonian simulation \cite{Lloyd1996,Childs2018},
which are two basic primitives used in quantum algorithms.

Before delving into the analysis of the quantum circuit depicted in
Figure~\ref{fig:q-circuit-gradient-est}, let us recall basics of
the block-encoding formalism \cite{Low2019hamiltonian,Gilyen2019}.
We say that a unitary $U$ is an $\left(\alpha,\delta\right)$-approximate
block-encoding of a matrix $A$ if
\begin{equation}
\left\Vert A-\alpha\left(\langle0|\otimes I\right)U\left(|0\rangle\otimes I\right)\right\Vert \leq\delta,
\end{equation}
where $\alpha$ is a normalization factor, $\delta$ is the approximation
error, and the norm $\left\Vert \cdot\right\Vert $ is the spectral
norm. It is also common to use the notation $\left(\alpha,a,\delta\right)$
for approximate block-encodings, where $a$ is the number of qubits
needed for the state vector $|0\rangle$. If $\delta=0$, then the
block-encoding is exact and the equality $\frac{A}{\alpha}=\left(\langle0|\otimes I\right)U\left(|0\rangle\otimes I\right)$
holds. In this case, the unitary $U$ is called an exact block-encoding
of $A$ because it encodes the matrix $A$ in the top-left block up
to the normalization factor $\alpha$:
\begin{equation}
U=\begin{bmatrix}\frac{A}{\alpha} & \cdot\\
\cdot & \cdot
\end{bmatrix}.
\end{equation}
Given a unitary that prepares a purification of a state $\sigma$,
it is possible to block-encode the state $\sigma$ \cite[Lemma~45]{Gilyen2019}.
If one only has sample access to a state $\sigma$, it is still possible
to block-encode it approximately by density matrix exponentiation
\cite[Corollary~21]{Gilyen2022a} (see also \cite[Lemma~2.21]{Wang2025}).
Essential to the quantum algorithm presented here, these block-encodings
of $\sigma$ can be transformed by QSVT to block-encodings of $\sigma^{-\frac{1}{2}}$
and $\sigma^{-\frac{is}{2}}$, as done in \cite{Gilyen2022} and Appendix~\ref{subsec:Block-encoding-for-modular-flow},
respectively, with a number of queries that depends inversely on the
minimum eigenvalue of $\sigma$ and logarithmically on the inverse
of the desired approximation error. Note that the modular-flow-simulation
algorithm in Appendix~\ref{subsec:Block-encoding-for-modular-flow}
improves upon the quantum algorithm from \cite{Lim2025}.

With this background in place, I now present an algorithm for estimating
the term $\left\langle G_{j}\right\rangle _{\Sigma_{v\to vh}^{\theta}(\rho)}$
in \eqref{eq:gradient-fully-QBM}. Note that the second term $\left\langle G_{j}\right\rangle _{\sigma_{vh}(\theta)}$
in \eqref{eq:gradient-fully-QBM} can be estimated by preparing the
thermal state $\sigma_{vh}(\theta)$ and measuring the observable
$G_{j}$.
\begin{lyxalgorithm}
\label{alg:estimate-grad-QBM}Suppose that we have access to a unitary
$U_{\psi}$ that prepares a purification $\psi^{\sigma}$ of the state
$\sigma_{v}(\theta)$. Set $\kappa>0$ to satisfy $\kappa^{-1}I\leq\sigma_{v}(\theta)$.
Set $\varepsilon>0$ to be the desired error, $\delta\in\left(0,1\right)$
to be the desired failure probability, and $m=1$. The algorithm for
estimating $\left\langle G_{j}\right\rangle _{\Sigma_{v\to vh}^{\theta}(\rho)}$
proceeds according to the following steps:
\begin{enumerate}
\item Pick $s\in\mathbb{R}$ at random according to the logistic probability
density $\beta(t)$ in \eqref{eq:logistic-prob-dens}, and pick $t\in\mathbb{R}$
at random according to the high-peak-tent probability density $\gamma(t)$
in \eqref{eq:high-peak-tent-def}.
\item Transform $U_{\psi}$ to an approximate block-encoding of $\sigma_{v}(\theta)^{-\frac{1}{2}}\sigma_{v}(\theta)^{-\frac{is}{2}}$,
by querying it multiple times and using QSVT, as shown in Appendix~\ref{sec:Detailed-analysis-QSVT}.
\item Execute the quantum circuit depicted in Figure~\ref{fig:q-circuit-gradient-est},
realizing the unitary evolution $e^{-iG(\theta)t}$ by Hamiltonian
simulation \cite{Lloyd1996,Childs2018}, and record the measurement
outcomes $z_{m}$ and $g_{m}$, which are the outcomes observed when
measuring $\sigma_{Z}$ on the control qubit and $G_{j}$ on the data
qubits in systems $v_{2}$ and $h$, respectively.
\item Set $Y_{m}\leftarrow\left(-1\right)^{z_{m}}g_{m}$, and $m\leftarrow m+1$.
\item Repeat Steps 1-4 $M\in\mathbb{N}$ times, where
\begin{align}
M & \geq O\!\left(\left(\frac{\kappa\left\Vert G_{j}\right\Vert }{\varepsilon}\right)^{2}\ln\!\left(\frac{1}{\delta}\right)\right),\label{eq:number-trials-main-alg}
\end{align}
and set $\overline{Y_{M}}\coloneqq\frac{\kappa}{M}\sum_{m=1}^{M}Y_{m}$
as an estimate of $\left\langle G_{j}\right\rangle _{\Sigma_{v\to vh}^{\theta}(\rho)}$.
\end{enumerate}
\end{lyxalgorithm}

\begin{thm}
\label{thm:alg-complexity-claim}For $\varepsilon>0$ and $\delta\in\left(0,1\right)$,
the estimate $\overline{Y_{M}}$ produced by Algorithm~\ref{alg:estimate-grad-QBM}
satisfies the following:
\begin{equation}
\Pr\!\left[\left|\overline{Y_{M}}-\left\langle G_{j}\right\rangle _{\Sigma_{v\to vh}^{\theta}(\rho)}\right|\leq\varepsilon\right]\geq1-\delta.
\end{equation}
Each execution of the quantum circuit in Figure~\ref{fig:q-circuit-gradient-est}
queries the unitary $U_{\psi}$ the following number of times:
\begin{equation}
\tilde{O}\!\left(\kappa\ln\!\left(\frac{\kappa\left\Vert G_{j}\right\Vert }{\varepsilon}\right)\right),
\end{equation}
where the notation $\tilde{O}$ suppresses various logarithmic factors.
Thus, the total number of times that Algorithm~\ref{alg:estimate-grad-QBM}
queries the unitary $U_{\psi}$ is given by
\begin{equation}
\tilde{O}\!\left(\frac{\kappa^{3}\left\Vert G_{j}\right\Vert ^{2}}{\varepsilon^{2}}\ln\!\left(\frac{\kappa\left\Vert G_{j}\right\Vert }{\varepsilon}\right)\ln\!\left(\frac{1}{\delta}\right)\right).\label{eq:total-num-queries-alg}
\end{equation}
\end{thm}

\begin{proof}
This claim follows as a consequence of the Hoeffding bound \cite{Hoeffding1963}
and the error analysis presented in Appendix~\ref{sec:Detailed-analysis-QSVT}.
\end{proof}
In the remainder of this section, I show how the state evolves through
the circuit depicted in Figure~\ref{fig:q-circuit-gradient-est},
in order to give a basic sense of the key quantum subroutine of Algorithm~\ref{alg:estimate-grad-QBM},
while leaving a detailed error analysis to Appendix~\ref{sec:Detailed-analysis-QSVT}.
In various steps below, ‘$\approx$’ denotes equality up to the approximation
errors of the QSVT block encodings, and I also make the abbreviations
$\sigma_{v}\equiv\sigma_{v}(\theta)$ and $\sigma_{vh}\equiv\sigma_{vh}(\theta)$.

Prior to the controlled swap in Figure~\ref{fig:q-circuit-gradient-est}
(which acts on the control qubit $c$ and systems $v_{1}$ and $v_{2}$),
the state of the registers $v_{1}a_{1}a_{2}v_{2}h$ is
\begin{equation}
\omega_{v_{1}a_{1}a_{2}v_{2}h}\equiv U_{\sigma_{v}^{-1/2}}(U_{\sigma_{v}^{-is/2}}(\rho_{v_{1}}\otimes|0\rangle\!\langle0|_{a_{1}})U_{\sigma_{v}^{-is/2}}^{\dag}\otimes|0\rangle\!\langle0|_{a_{2}})U_{\sigma_{v}^{-1/2}}^{\dag}\otimes\sigma_{v_{2}h}.\label{eq:omega-state-alg}
\end{equation}
The unitary $U_{\sigma_{v}^{-1/2}}$ is a block-encoding of $\sigma_{v}^{-1/2}$,
and the unitary $U_{\sigma_{v}^{-is/2}}$ is a block-encoding of $\sigma_{v}^{-is/2}$,
so that
\begin{align}
\alpha_{1}\left(I_{v_{1}}\otimes\langle0|_{a_{2}}\right)U_{\sigma_{v}^{-is/2}}\left(I_{v_{1}}\otimes|0\rangle_{a_{2}}\right) & \approx\sigma_{v}^{-is/2},\label{eq:QSVT-approx-1}\\
\alpha_{2}\left(I_{v_{1}}\otimes\langle0|_{a_{1}}\right)U_{\sigma_{v}^{-1/2}}\left(I_{v_{1}}\otimes|0\rangle_{a_{1}}\right) & \approx\sigma_{v}^{-1/2},\label{eq:QSVT-approx-2}
\end{align}
for constants $\alpha_{1}=1$ and $\alpha_{2}=\sqrt{\kappa}$, where
$\kappa>0$ satisfies $\kappa^{-1}I\leq\sigma_{v}(\theta)$. The constants
$\alpha_{1}$ and $\alpha_{2}$ arise from the normalization of the
QSVT block encodings. After performing the controlled swap, the following
observable is measured:
\begin{equation}
X_{c}\otimes O_{v_{1}a_{1}a_{2}v_{2}h},
\end{equation}
where $X_{c}$ denotes the Pauli-$X$ matrix acting on the control
qubit $c$, and
\begin{align}
O_{v_{1}a_{1}a_{2}v_{2}h} & \equiv I_{v_{1}}\otimes|0\rangle\!\langle0|_{a_{1}}\otimes|0\rangle\!\langle0|_{a_{2}}\otimes O_{v_{2}h}(t),\label{eq:observable-all-sys-alg}\\
O_{v_{2}h}(t) & \equiv\left(e^{iG(\theta)t}G_{j}e^{-iG(\theta)t}\right)_{v_{2}h}.
\end{align}
Thus, the expectation of the measurement outcome is given by
\begin{equation}
\Tr\!\left[\left(X_{c}\otimes O_{v_{1}a_{1}a_{2}v_{2}h}\right)\left(\text{c-}F_{v_{1}v_{2}}\right)\left(|+\rangle\!\langle+|_{c}\otimes\omega_{v_{1}a_{1}a_{2}v_{2}h}\right)\left(\text{c-}F_{v_{1}v_{2}}\right)^{\dag}\right],
\end{equation}
where
\begin{equation}
\text{c-}F_{v_{1}v_{2}}\coloneqq|0\rangle\!\langle0|_{c}\otimes I_{v_{1}v_{2}}+|1\rangle\!\langle1|_{c}\otimes F_{v_{1}v_{2}}
\end{equation}
denotes the controlled swap unitary.

Let us now explicitly evaluate the measurement expectation value and
show that it approximately yields the desired term term $\left\langle G_{j}\right\rangle _{\Sigma_{v\to vh}^{\theta}(\rho)}$
in \eqref{eq:gradient-fully-QBM}. Consider that
\begin{align}
 & \left(\alpha_{1}\alpha_{2}\right)^{2}\Tr\!\left[\left(X_{c}\otimes O_{v_{1}a_{1}a_{2}v_{2}h}\right)\left(\text{c-}F_{v_{1}v_{2}}\right)\left(|+\rangle\!\langle+|_{c}\otimes\omega_{v_{1}a_{1}a_{2}v_{2}h}\right)\left(\text{c-}F_{v_{1}v_{2}}\right)^{\dag}\right]\nonumber \\
 & \overset{(a)}{=}\frac{\kappa}{2}\Tr\!\left[\left\{ |0\rangle\!\langle0|_{a_{1}}\otimes|0\rangle\!\langle0|_{a_{2}}\otimes O_{v_{2}h}(t),F_{v_{1}v_{2}}\right\} \omega_{v_{1}a_{1}a_{2}v_{2}h}\right]\\
 & =\frac{\kappa}{2}\Tr\!\left[\left\{ O_{v_{2}h}(t),F_{v_{1}v_{2}}\right\} \left(\langle0|_{a_{1}}\otimes\langle0|_{a_{2}}\right)\omega_{v_{1}a_{1}a_{2}v_{2}h}\left(|0\rangle_{a_{1}}\otimes|0\rangle_{a_{2}}\right)\right]\\
 & \overset{(b)}{\approx}\frac{1}{2}\Tr\!\left[\left\{ O_{v_{2}h}(t),F_{v_{1}v_{2}}\right\} \left(\sigma_{v}^{-\frac{1}{2}}\sigma_{v}^{-\frac{is}{2}}\rho_{v_{1}}\sigma_{v}^{\frac{is}{2}}\sigma_{v}^{-\frac{1}{2}}\otimes\sigma_{v_{2}h}\right)\right]\\
 & =\frac{1}{2}\Tr\!\left[F_{v_{1}v_{2}}O_{v_{2}h}(t)\left(\sigma_{v}^{-\frac{1}{2}}\sigma_{v}^{-\frac{is}{2}}\rho_{v_{1}}\sigma_{v}^{\frac{is}{2}}\sigma_{v}^{-\frac{1}{2}}\otimes\sigma_{v_{2}h}\right)\right]\nonumber \\
 & \qquad+\frac{1}{2}\Tr\!\left[O_{v_{2}h}(t)F_{v_{1}v_{2}}\left(\sigma_{v}^{-\frac{1}{2}}\sigma_{v}^{-\frac{is}{2}}\rho_{v_{1}}\sigma_{v}^{\frac{is}{2}}\sigma_{v}^{-\frac{1}{2}}\otimes\sigma_{v_{2}h}\right)\right]\\
 & =\frac{1}{2}\Tr\!\left[F_{v_{1}v_{2}}\left(\sigma_{v}^{-\frac{1}{2}}\sigma_{v}^{-\frac{is}{2}}\rho_{v_{1}}\sigma_{v}^{\frac{is}{2}}\sigma_{v}^{-\frac{1}{2}}\otimes O_{v_{2}h}(t)\sigma_{v_{2}h}\right)\right]\nonumber \\
 & \qquad+\frac{1}{2}\Tr\!\left[F_{v_{1}v_{2}}\left(\sigma_{v}^{-\frac{1}{2}}\sigma_{v}^{-\frac{is}{2}}\rho_{v_{1}}\sigma_{v}^{\frac{is}{2}}\sigma_{v}^{-\frac{1}{2}}\otimes\sigma_{v_{2}h}O_{v_{2}h}(t)\right)\right]\\
 & \overset{(c)}{=}\frac{1}{2}\Tr\!\left[\sigma_{v}^{-\frac{1}{2}}\sigma_{v}^{-\frac{is}{2}}\rho_{v}\sigma_{v}^{\frac{is}{2}}\sigma_{v}^{-\frac{1}{2}}O_{v_{2}h}(t)\sigma_{vh}\right]\nonumber \\
 & \qquad+\frac{1}{2}\Tr\!\left[\sigma_{v}^{-\frac{1}{2}}\sigma_{v}^{-\frac{is}{2}}\rho_{v}\sigma_{v}^{\frac{is}{2}}\sigma_{v}^{-\frac{1}{2}}\sigma_{vh}O_{v_{2}h}(t)\right]\\
 & =\frac{1}{2}\Tr\!\left[e^{iG(\theta)t}G_{j}e^{-iG(\theta)t}\left\{ \sigma_{vh},\sigma_{v}^{-\frac{1}{2}}\sigma_{v}^{-\frac{is}{2}}\rho_{v}\sigma_{v}^{\frac{is}{2}}\sigma_{v}^{-\frac{1}{2}}\otimes I_{h}\right\} \right].
\end{align}
The equality $(a)$ follows from an analysis similar to that in \cite[Eqs.~(B1)--(B10)]{Minervini2025}.
The approximation $(b)$ follows from \eqref{eq:QSVT-approx-1}--\eqref{eq:QSVT-approx-2};
i.e., 
\begin{align}
 & \left(\alpha_{1}\alpha_{2}\right)^{2}\Tr_{a_{1}a_{2}}\!\left[\left(|0\rangle\!\langle0|_{a_{1}}\otimes|0\rangle\!\langle0|_{a_{2}}\right)\omega_{v_{1}a_{1}a_{2}v_{2}h}\right]\nonumber \\
 & =\left(\alpha_{1}\alpha_{2}\right)^{2}\left(\langle0|_{a_{1}}\otimes\langle0|_{a_{2}}\right)\omega_{v_{1}a_{1}a_{2}v_{2}h}\left(|0\rangle_{a_{1}}\otimes|0\rangle_{a_{2}}\right)\\
 & \approx\sigma_{v}^{-\frac{1}{2}}\sigma_{v}^{-\frac{is}{2}}\rho_{v_{1}}\sigma_{v}^{-\frac{is}{2}}\sigma_{v}^{-\frac{1}{2}}\otimes\sigma_{v_{2}h}.
\end{align}
The equality $(c)$ follows from the swap-trick identity $\Tr[F(X\otimes Y)]=\Tr[XY]$.

Thus, the quantum circuit in Figure~\ref{fig:q-circuit-gradient-est}
produces an estimate of
\begin{equation}
\frac{1}{2}\Tr\!\left[e^{iG(\theta)t}G_{j}e^{-iG(\theta)t}\left\{ \sigma_{vh},\sigma_{v}^{-\frac{1}{2}}\sigma_{v}^{-\frac{is}{2}}\rho_{v}\sigma_{v}^{\frac{is}{2}}\sigma_{v}^{-\frac{1}{2}}\otimes I_{h}\right\} \right],
\end{equation}
which is the desired quantity $\left\langle G_{j}\right\rangle _{\Sigma_{v\to vh}^{\theta}(\rho)}$
after taking the expectation over $s$ and $t$. This completes the
derivation that the circuit in Figure~\ref{fig:q-circuit-gradient-est}
estimates the term $\left\langle G_{j}\right\rangle _{\Phi_{vh}^{\theta}\!\left(\eta(\theta)\right)}$
in \eqref{eq:gradient-fully-QBM}.
\begin{rem}
Algorithm~\ref{alg:estimate-grad-QBM} assumes access to a unitary
$U_{\psi}$ that prepares a purification of the state $\sigma_{v}(\theta)$,
in order to realize a block-encoding of $\sigma_{v}(\theta)$. Realizing
this unitary relies on quantum algorithms for thermal state preparation
(see, e.g., \cite{Chen2025}). As indicated above, one can still arrive
at a block-encoding of $\sigma_{v}(\theta)$ even if one just has
sample access to $\sigma_{v}(\theta)$, by employing \cite[Corollary~21]{Gilyen2022a}
(see also \cite[Lemma~2.21]{Wang2025}). In doing so, extra error
terms will be introduced that can be accounted for.
\end{rem}

\subsection{Application to restricted quantum Boltzmann machines}

\label{subsec:Application-to-restricted-QBMs}In this brief section,
I introduce a model of restricted quantum Boltzmann machines (rQBMs)
that generalizes that put forward in \cite[Section~2.1]{Wiebe2019},
in the sense that there are no commuting restrictions imposed. Furthermore,
I show what the gradient terms in Theorem~\ref{thm:q-state-learning-gradient-vh-q}
evaluate to for rQBMs. Indeed, there is not much of a distinction
with the general case presented in Theorem~\ref{thm:q-state-learning-gradient-vh-q},
but I still list the terms for completeness. 

Let us begin by defining a general model of rQBMs. Let $m,n\in\mathbb{N}$.
Let $\theta\equiv\left(a,b,w\right)$, where $a\in\mathbb{R}^{m}$,
$b\in\mathbb{R}^{n}$, and $w\in\mathbb{R}^{m\times n}$. Let $\left(V_{i}\right)_{i=1}^{m}$
be a tuple of Hermitian operators, and let $\left(H_{j}\right)_{j=1}^{n}$
be a tuple of Hermitian operators. Then the rQBM Hamiltonian $G(\theta)$
is defined as
\begin{equation}
G(\theta)\coloneqq\sum_{i=1}^{m}a_{i}V_{i}\otimes I+I\otimes\sum_{j=1}^{n}b_{j}H_{j}+\sum_{i=1}^{m}\sum_{j=1}^{n}w_{i,j}V_{i}\otimes H_{j}.\label{eq:rQBM-def}
\end{equation}
The thermal state is defined as
\begin{align}
\sigma_{vh}(\theta) & \coloneqq\frac{e^{-G(\theta)}}{Z(\theta)},\\
Z(\theta) & \coloneqq\Tr[e^{-G(\theta)}].
\end{align}

\begin{thm}
For the rQBM model defined in \eqref{eq:rQBM-def}, the partial derivatives
of the quantum relative entropy $D(\rho\|\sigma_{v}(\theta))$ are
as follows:
\begin{align}
\frac{\partial}{\partial a_{i}}D(\rho\|\sigma_{v}(\theta)) & =\left\langle V_{i}\right\rangle _{(\Tr_{h}\circ\Sigma_{v\to vh}^{\theta})\left(\rho\right)}-\left\langle V_{i}\right\rangle _{\sigma_{v}(\theta)},\\
\frac{\partial}{\partial b_{j}}D(\rho\|\sigma_{v}(\theta)) & =\left\langle H_{j}\right\rangle _{(\Tr_{v}\circ\Sigma_{v\to vh}^{\theta})\left(\rho\right)}-\left\langle H_{j}\right\rangle _{\sigma_{h}(\theta)},\\
\frac{\partial}{\partial w_{i,j}}D(\rho\|\sigma_{v}(\theta)) & =\left\langle V_{i}\otimes H_{j}\right\rangle _{\Sigma_{v\to vh}^{\theta}\!\left(\rho\right)}-\left\langle V_{i}\otimes H_{j}\right\rangle _{\sigma_{vh}(\theta)},
\end{align}
where the HPTP map $\Sigma_{v\to vh}^{\theta}$ is defined as in \eqref{eq:herm-pres-map-Sigma},
but with respect to the Hamiltonian $G(\theta)$ in \eqref{eq:rQBM-def}.
\end{thm}

The expressions above generalize those from \eqref{eq:rBMs-gradient-1}--\eqref{eq:rBMs-gradient-last},
for classical restricted Boltzmann machines.

\subsection{Consistency check: no hidden units}

\label{subsec:Consistency-check-no-HUs}In this brief section, I show
that the gradient formula in Theorem~\ref{thm:q-state-learning-gradient-vh-q}
is consistent with the formula reported in \cite[Eq.~(4)]{Kieferova2017}
whenever there are no hidden units (i.e., the system $h$ in \eqref{eq:def-QBM-state-vh}
is trivial). 
\begin{prop}
If there are no hidden units, then
\begin{align}
\frac{\partial}{\partial\theta_{j}}D(\rho\|\sigma_{v}(\theta)) & =\left\langle G_{j}\right\rangle _{\Sigma_{v}^{\theta}(\rho)}-\left\langle G_{j}\right\rangle _{\sigma_{v}(\theta)},\label{eq:no-hidden-unit-consistency-1}\\
 & =\left\langle G_{j}\right\rangle _{\rho}-\left\langle G_{j}\right\rangle _{\sigma_{v}(\theta)},\label{eq:no-hidden-unit-consistency-2}
\end{align}
where
\begin{align}
\Sigma_{v}^{\theta} & \coloneqq\Phi_{v}^{\theta}\circ\Xi_{v}\circ\Upsilon_{v}^{\theta},\label{eq:Sigma-map-no-hidden-units}\\
\Xi_{v}^{\theta}(R_{v}) & \coloneqq\frac{1}{2}\left\{ \sigma_{v}(\theta),\sigma_{v}(\theta)^{-\frac{1}{2}}R_{v}\sigma_{v}(\theta)^{-\frac{1}{2}}\otimes I_{h}\right\} ,\\
\sigma_{v}(\theta) & =\frac{e^{-G(\theta)}}{Z(\theta)}.
\end{align}
\end{prop}

\begin{proof}
If there are no hidden units, then observe that the HPTP map $\Sigma_{v\to vh}^{\theta}$
in \eqref{eq:herm-pres-map-Sigma} reduces to
\begin{equation}
\Sigma_{v}^{\theta}(R_{v})=\Phi_{v}^{\theta}\!\left(\left\{ \sigma_{v}(\theta),\sigma_{v}(\theta)^{-\frac{1}{2}}\Upsilon_{v}^{\theta}(R_{v})\sigma_{v}(\theta)^{-\frac{1}{2}}\right\} \right),
\end{equation}
as stated in \eqref{eq:Sigma-map-no-hidden-units}. To establish the
equality between \eqref{eq:no-hidden-unit-consistency-1} and \eqref{eq:no-hidden-unit-consistency-2},
it suffices to prove that $\left\langle G_{j}\right\rangle _{\Sigma_{v}^{\theta}(\rho)}=\left\langle G_{j}\right\rangle _{\rho}$.
To this end, consider that
\begin{align}
 & \left\langle G_{j}\right\rangle _{\Sigma_{v}^{\theta}(\rho)}\nonumber \\
 & =\Tr\!\left[G_{j}\frac{1}{2}\Phi_{v}^{\theta}\!\left(\left\{ \sigma_{v}(\theta),\sigma_{v}(\theta)^{-\frac{1}{2}}\Upsilon_{v}^{\theta}(\rho)\sigma_{v}(\theta)^{-\frac{1}{2}}\right\} \right)\right]\\
 & =\Tr\!\left[\frac{1}{2}\left\{ \Phi_{v}^{\theta}(G_{j}),\sigma_{v}(\theta)\right\} \sigma_{v}(\theta)^{-\frac{1}{2}}\Upsilon_{v}^{\theta}(\rho)\sigma_{v}(\theta)^{-\frac{1}{2}}\right]\\
 & =\frac{1}{Z(\theta)}\Tr\!\left[\frac{1}{2}\left\{ \Phi_{v}^{\theta}(G_{j}),e^{-G(\theta)}\right\} \sigma_{v}(\theta)^{-\frac{1}{2}}\Upsilon_{v}^{\theta}(\rho)\sigma_{v}(\theta)^{-\frac{1}{2}}\right]\\
 & \overset{(a)}{=}\frac{1}{Z(\theta)}\Tr\!\left[\begin{array}{c}
\int_{0}^{1}dt\,e^{-tG(\theta)}G_{j}e^{-\left(1-t\right)G(\theta)}\times\\
\int_{0}^{\infty}ds\,\left(\sigma_{v}(\theta)+sI\right)^{-1}\rho\left(\sigma_{v}(\theta)+sI\right)^{-1}
\end{array}\right]\\
 & =\int_{0}^{1}dt\,\int_{0}^{\infty}ds\,\Tr\!\left[\sigma_{v}(\theta)^{t}G_{j}\sigma_{v}(\theta)^{1-t}\left(\sigma_{v}(\theta)+sI\right)^{-1}\rho\left(\sigma_{v}(\theta)+sI\right)^{-1}\right]\\
 & =\int_{0}^{1}dt\,\int_{0}^{\infty}ds\,\Tr\!\left[G_{j}\sigma_{v}(\theta)^{1-t}\left(\sigma_{v}(\theta)+sI\right)^{-1}\rho\left(\sigma_{v}(\theta)+sI\right)^{-1}\sigma_{v}(\theta)^{t}\right]\\
 & \overset{(b)}{=}\Tr\!\left[G_{j}\rho\right].
\end{align}
The equality $(a)$ follows from \eqref{eq:duhamel-form}--\eqref{eq:duhamel-fourier}
and \eqref{eq:deriv-matrix-log-standard}--\eqref{eq:deriv-matrix-log-fourier}.
The equality $(b)$ can be understood as follows: the superoperator
formed from the derivative of the matrix exponential and the superoperator
formed from the derivative of the matrix logarithm, as given above,
are inverses of each other. Indeed, to see this, let $\sum_{k}\lambda_{k}\Pi_{k}$
be a spectral decomposition of $\sigma_{v}(\theta)$ and consider
that
\begin{align}
 & \Tr\!\left[G_{j}\sigma_{v}(\theta)^{1-t}\left(\sigma_{v}(\theta)+sI\right)^{-1}\rho\left(\sigma_{v}(\theta)+sI\right)^{-1}\sigma_{v}(\theta)^{t}\right]\nonumber \\
 & =\Tr\!\left[\begin{array}{c}
G_{j}\left(\sum_{j}\lambda_{j}\Pi_{j}\right)^{1-t}\left(\sum_{m}\lambda_{m}\Pi_{m}+sI\right)^{-1}\times\\
\rho\left(\sum_{\ell}\lambda_{\ell}\Pi_{\ell}+sI\right)^{-1}\left(\sum_{k}\lambda_{k}\Pi_{k}\right)^{t}
\end{array}\right]\\
 & =\Tr\!\left[\begin{array}{c}
G_{j}\left(\sum_{j}\lambda_{j}^{1-t}\Pi_{j}\right)\left(\sum_{m}\left(\lambda_{m}+s\right)^{-1}\Pi_{m}\right)\times\\
\rho\left(\sum_{\ell}\left(\lambda_{\ell}+s\right)^{-1}\Pi_{\ell}\right)\left(\sum_{k}\lambda_{k}^{t}\Pi_{k}\right)
\end{array}\right]\\
 & =\sum_{j,m,\ell,k}\lambda_{j}^{1-t}\left(\lambda_{m}+s\right)^{-1}\left(\lambda_{\ell}+s\right)^{-1}\lambda_{k}^{t}\Tr\!\left[G_{j}\Pi_{j}\Pi_{m}\rho\Pi_{\ell}\Pi_{k}\right]\\
 & =\sum_{j,\ell}\lambda_{j}^{1-t}\lambda_{\ell}^{t}\left(\lambda_{j}+s\right)^{-1}\left(\lambda_{\ell}+s\right)^{-1}\Tr\!\left[G_{j}\Pi_{j}\rho\Pi_{\ell}\right],
\end{align}
which implies that
\begin{align}
 & \int_{0}^{1}dt\,\int_{0}^{\infty}ds\,\Tr\!\left[G_{j}\sigma_{v}(\theta)^{1-t}\left(\sigma_{v}(\theta)+sI\right)^{-1}\rho\left(\sigma_{v}(\theta)+sI\right)^{-1}\sigma_{v}(\theta)^{t}\right]\nonumber \\
 & =\int_{0}^{1}dt\,\int_{0}^{\infty}ds\,\sum_{j,\ell}\lambda_{j}^{1-t}\lambda_{\ell}^{t}\left(\lambda_{j}+s\right)^{-1}\left(\lambda_{\ell}+s\right)^{-1}\Tr\!\left[G_{j}\Pi_{j}\rho\Pi_{\ell}\right]\\
 & =\sum_{j,\ell}\left(\int_{0}^{1}dt\,\lambda_{j}^{1-t}\lambda_{\ell}^{t}\right)\int_{0}^{\infty}ds\,\left(\lambda_{j}+s\right)^{-1}\left(\lambda_{\ell}+s\right)^{-1}\Tr\!\left[G_{j}\Pi_{j}\rho\Pi_{\ell}\right]\\
 & =\sum_{j,\ell:\lambda_{j}=\lambda_{\ell}}\Tr\!\left[G_{j}\Pi_{j}\rho\Pi_{\ell}\right]\nonumber \\
 & \qquad+\sum_{j,\ell:\lambda_{j}\neq\lambda_{\ell}}\left(\frac{\lambda_{j}-\lambda_{\ell}}{\ln\lambda_{j}-\ln\lambda_{\ell}}\right)\left(\frac{\ln\lambda_{j}-\ln\lambda_{\ell}}{\lambda_{j}-\lambda_{\ell}}\right)\Tr\!\left[G_{j}\Pi_{j}\rho\Pi_{\ell}\right]\\
 & =\sum_{j,\ell}\Tr\!\left[G_{j}\Pi_{j}\rho\Pi_{\ell}\right]\\
 & =\Tr\!\left[G_{j}\rho\right],
\end{align}
thus concluding the proof.
\end{proof}

\section{Quantum state learning using quantum--classical Boltzmann machines}

\label{sec:Quantum-state-learning-qc-QBMs}In this section, I consider
the quantum state learning task when using a quantum--classical Boltzmann
machine, i.e., one for which the visible system is quantum and the
hidden system is classical. Here I provide an analytical formula for
the gradient of the quantum relative entropy in quantum state learning,
when using quantum--classical quantum Boltzmann machines (Theorem~\ref{thm:qc-vh-gradient}).
Additionally, I briefly sketch a quantum algorithm for estimating
the gradient (Section~\ref{subsec:Quantum-algorithm-qc-BM}), which
is a special case of the quantum algorithm from Section~\ref{subsec:Q-algorithm-state-learning-QBM}.
I also show how the gradient formula simplifies when considering a
restricted quantum--classical Boltzmann machine (Section~\ref{subsec:Application-to-restricted-qc-QBMs}).

The most general form for a quantum--classical Boltzmann machine
involves a parameterized Hamiltonian $G(\theta)$ of the form in \eqref{eq:gen-param-ham},
but with each term $G_{j}$ having a classical hidden system, so that
it can be written as 
\begin{equation}
G_{j}\coloneqq\sum_{x}G_{v}^{j,x}\otimes|x\rangle\!\langle x|_{h},\label{eq:term-commuting-hidden}
\end{equation}
where $\left\{ |x\rangle\right\} _{x}$ is an orthonormal basis and
each $G_{v}^{j,x}$ is a Hermitian operator acting on the visible
system. This assumption implies that
\begin{align}
G(\theta) & =\sum_{j=1}^{J}\theta_{j}G_{j}\\
 & =\sum_{j=1}^{J}\theta_{j}\left(\sum_{x}G_{v}^{j,x}\otimes|x\rangle\!\langle x|_{h}\right),\\
 & =\sum_{x}G_{v}^{x}(\theta)\otimes|x\rangle\!\langle x|_{h},
\end{align}
where
\begin{align*}
G_{v}^{x}(\theta) & \coloneqq\sum_{j=1}^{J}\theta_{j}G_{v}^{j,x}.
\end{align*}
The thermal state in \eqref{eq:def-QBM-state-vh} then has the following
quantum--classical form:
\begin{equation}
\sigma_{vh}(\theta)=\sum_{x}p_{x}(\theta)\sigma_{v}^{x}(\theta)\otimes|x\rangle\!\langle x|_{h},\label{eq:qc-state}
\end{equation}
where
\begin{align}
\sigma_{v}^{x}(\theta) & \coloneqq\frac{e^{-G_{v}^{x}(\theta)}}{\Tr[e^{-G_{v}^{x}(\theta)}]},\label{eq:sigma-x-v-def}\\
p_{x}(\theta) & \coloneqq\frac{\Tr[e^{-G_{v}^{x}(\theta)}]}{\Tr[e^{-G(\theta)}]}=\frac{\Tr[e^{-G_{v}^{x}(\theta)}]}{\sum_{x}\Tr[e^{-G_{v}^{x}(\theta)}]}.\label{eq:p_x_theta_def}
\end{align}
Observe that the reduced state of the visible system is
\begin{equation}
\sigma_{v}(\theta)=\Tr_{h}\!\left[\sigma_{vh}(\theta)\right]=\sum_{x}p_{x}(\theta)\sigma_{v}^{x}(\theta).\label{eq:visible-state-qc}
\end{equation}

\subsection{Analytical formula for the gradient of quantum--classical Boltzmann
machines}

\label{subsec:Analytical-formula-for-gradient-qc-QBMs}Theorem~\ref{thm:qc-vh-gradient}
below provides an analytical expression for the gradient of \eqref{eq:q-rel-ent-def},
under the model assumptions presented in \eqref{eq:term-commuting-hidden}--\eqref{eq:visible-state-qc}.
This expression is also a quantum generalization of the expression
in \eqref{eq:classical-gradient}, reducing to it in the fully classical
case. The first term of \eqref{eq:gradient-qc-QBM} is expressed in
terms of the HPTP map $\rho\to\sum_{x}p_{x}(\theta)\Sigma_{v}^{\theta,x}(\rho)\otimes|x\rangle\!\langle x|_{h}$,
which is reminiscent of, yet distinct from, the rotated pretty good
instrument of \cite[Remark~5.7]{Wilde2015}. Each term in \eqref{eq:gradient-qc-QBM}
can be estimated by first sampling $x$ from $p_{x}(\theta)$, and
then estimating $\langle G_{v}^{j,x}\rangle_{\Sigma_{v}^{\theta,x}(\rho)}$
and $\langle G_{v}^{j,x}\rangle_{\sigma_{v}^{x}(\theta)}$.
\begin{thm}
\label{thm:qc-vh-gradient}Let $\rho$ be a target quantum state,
and let $\sigma_{v}(\theta)$ be a model state of the form in \eqref{eq:visible-state-qc}.
The partial derivatives of $D(\rho\|\sigma_{v}(\theta))$ are as follows:
\begin{equation}
\frac{\partial}{\partial\theta_{j}}D(\rho\|\sigma_{v}(\theta))=\sum_{x}p_{x}(\theta)\left\langle G_{v}^{j,x}\right\rangle _{\Sigma_{v}^{\theta,x}(\rho)}-\sum_{x}p_{x}(\theta)\left\langle G_{v}^{j,x}\right\rangle _{\sigma_{v}^{x}(\theta)},\label{eq:gradient-qc-QBM}
\end{equation}
where
\begin{align}
\Sigma_{v}^{\theta,x} & \coloneqq\Phi_{v}^{\theta,x}\circ\Xi_{v}^{\theta,x}\circ\Upsilon_{v}^{\theta},\label{eq:Sigma-map-x-dependent}\\
\Xi_{v}^{\theta,x}(R_{v}) & \coloneqq\frac{1}{2}\left\{ \sigma_{v}^{x}(\theta),\sigma_{v}(\theta)^{-\frac{1}{2}}R_{v}\sigma_{v}(\theta)^{-\frac{1}{2}}\right\} ,\label{eq:Xi-map-x-dependent}\\
\Phi_{v}^{\theta,x}(Y_{v}) & \coloneqq\int_{-\infty}^{\infty}dt\,\gamma(t)\,e^{-iG_{v}^{x}(\theta)t}Y_{v}e^{iG_{v}^{x}(\theta)t},\\
\Upsilon_{v}^{\theta}(X_{v}) & \coloneqq\int_{-\infty}^{\infty}dt\,\beta(t)\,\sigma_{v}(\theta)^{-\frac{it}{2}}X_{v}\sigma_{v}(\theta)^{\frac{it}{2}}.
\end{align}
\end{thm}

\begin{proof}
Let us first prove that the HPTP map $\Sigma_{v\to vh}^{\theta}$
in \eqref{eq:herm-pres-map-Sigma} reduces as follows under the model
assumptions in \eqref{eq:term-commuting-hidden}--\eqref{eq:visible-state-qc}:
\begin{align}
\Sigma_{v\to vh}^{\theta}(\rho) & =\sum_{x}p_{x}(\theta)\Sigma_{v}^{\theta,x}(\rho)\otimes|x\rangle\!\langle x|_{h}\\
 & =\sum_{x}p_{x}(\theta)\Phi_{v}^{\theta,x}\!\left(\frac{1}{2}\left\{ \sigma_{v}^{x}(\theta),\sigma_{v}(\theta)^{-\frac{1}{2}}\Upsilon_{v}^{\theta}(\rho)\sigma_{v}(\theta)^{-\frac{1}{2}}\right\} \right)\otimes|x\rangle\!\langle x|_{h},
\end{align}
so that it is reminiscent of, yet distinct from, the rotated pretty
good instrument from \cite[Remark~5.7]{Wilde2015}. To this end, consider
that
\begin{align}
 & \Sigma_{v\to vh}^{\theta}(\rho)\nonumber \\
 & =\Phi_{vh}^{\theta}\!\left(\frac{1}{2}\left\{ \sigma_{vh}(\theta),\sigma_{v}(\theta)^{-\frac{1}{2}}\Upsilon_{v}^{\theta}(\rho)\sigma_{v}(\theta)^{-\frac{1}{2}}\otimes I_{h}\right\} \right),\\
 & =\Phi_{vh}^{\theta}\!\left(\frac{1}{2}\left\{ \sum_{x}p_{x}(\theta)\sigma_{v}^{x}(\theta)\otimes|x\rangle\!\langle x|_{h},\sigma_{v}(\theta)^{-\frac{1}{2}}\Upsilon_{v}^{\theta}(\rho)\sigma_{v}(\theta)^{-\frac{1}{2}}\otimes I_{h}\right\} \right)\\
 & =\Phi_{vh}^{\theta}\!\left(\sum_{x}p_{x}(\theta)\frac{1}{2}\left\{ \sigma_{v}^{x}(\theta),\sigma_{v}(\theta)^{-\frac{1}{2}}\Upsilon_{v}^{\theta}(\rho)\sigma_{v}(\theta)^{-\frac{1}{2}}\right\} \otimes|x\rangle\!\langle x|_{h}\right)\\
 & =\int_{-\infty}^{\infty}dt\,\gamma(t)\:e^{-iG(\theta)t}\times\nonumber \\
 & \quad\left(\sum_{x}p_{x}(\theta)\frac{1}{2}\left\{ \sigma_{v}^{x}(\theta),\sigma_{v}(\theta)^{-\frac{1}{2}}\Upsilon_{v}^{\theta}(\rho)\sigma_{v}(\theta)^{-\frac{1}{2}}\right\} \otimes|x\rangle\!\langle x|_{h}\right)e^{iG(\theta)t}\\
 & =\int_{-\infty}^{\infty}dt\,\gamma(t)\:\left(\sum_{x'}e^{-iG_{v}^{x'}(\theta)t}\otimes|x'\rangle\!\langle x'|_{h}\right)\times\nonumber \\
 & \qquad\left(\sum_{x}p_{x}(\theta)\frac{1}{2}\left\{ \sigma_{v}^{x}(\theta),\sigma_{v}(\theta)^{-\frac{1}{2}}\Upsilon_{v}^{\theta}(\rho)\sigma_{v}(\theta)^{-\frac{1}{2}}\right\} \otimes|x\rangle\!\langle x|_{h}\right)\times\nonumber \\
 & \qquad\left(\sum_{x''}e^{iG_{v}^{x''}(\theta)t}\otimes|x''\rangle\!\langle x''|_{h}\right)\\
 & =\int_{-\infty}^{\infty}dt\,\gamma(t)\:\sum_{x}p_{x}(\theta)e^{-iG_{v}^{x}(\theta)t}\times\nonumber \\
 & \qquad\left(\frac{1}{2}\left\{ \sigma_{v}^{x}(\theta),\sigma_{v}(\theta)^{-\frac{1}{2}}\Upsilon_{v}^{\theta}(\rho)\sigma_{v}(\theta)^{-\frac{1}{2}}\right\} \right)e^{iG_{v}^{x}(\theta)t}\otimes|x\rangle\!\langle x|_{h}\\
 & =\sum_{x}p_{x}(\theta)\Phi_{v}^{\theta,x}\!\left(\frac{1}{2}\left\{ \sigma_{v}^{x}(\theta),\sigma_{v}(\theta)^{-\frac{1}{2}}\Upsilon_{v}^{\theta}(\rho)\sigma_{v}(\theta)^{-\frac{1}{2}}\right\} \right)\otimes|x\rangle\!\langle x|_{h}.
\end{align}

The remainder of the proof involves showing how the two terms in \eqref{eq:gradient-fully-QBM}
simplify under the model assumptions in \eqref{eq:term-commuting-hidden}--\eqref{eq:visible-state-qc}.
For the first term, consider that
\begin{align}
 & \left\langle G_{j}\right\rangle _{\Sigma_{v\to vh}^{\theta}\!\left(\rho\right)}\nonumber \\
 & =\Tr\!\left[\begin{array}{c}
\left(\sum_{x'}G_{v}^{j,x'}\otimes|x'\rangle\!\langle x'|_{h}\right)\times\\
\left(\sum_{x}p_{x}(\theta)\Phi_{v}^{\theta,x}\!\left(\frac{1}{2}\left\{ \sigma_{v}^{x}(\theta),\sigma_{v}(\theta)^{-\frac{1}{2}}\Upsilon_{v}^{\theta}(\rho)\sigma_{v}(\theta)^{-\frac{1}{2}}\right\} \right)\otimes|x\rangle\!\langle x|_{h}\right)
\end{array}\right]\\
 & =\sum_{x}p_{x}(\theta)\Tr\!\left[G_{v}^{j,x}\Phi_{v}^{\theta,x}\!\left(\frac{1}{2}\left\{ \sigma_{v}^{x}(\theta),\sigma_{v}(\theta)^{-\frac{1}{2}}\Upsilon_{v}^{\theta}(\rho)\sigma_{v}(\theta)^{-\frac{1}{2}}\right\} \right)\right]\\
 & =\sum_{x}p_{x}(\theta)\Tr\!\left[G_{v}^{j,x}\Sigma_{v}^{\theta,x}(\rho)\right]\\
 & =\sum_{x}p_{x}(\theta)\left\langle G_{v}^{j,x}\right\rangle _{\Sigma_{v}^{\theta,x}(\rho)},
\end{align}
where the map $\Sigma_{v}^{\theta,x}$ is defined in \eqref{eq:Sigma-map-x-dependent}.
For the second term, consider that
\begin{align}
\left\langle G_{j}\right\rangle _{\sigma_{vh}(\theta)} & =\Tr\!\left[\left(\sum_{x}G_{v}^{j,x}\otimes|x\rangle\!\langle x|_{h}\right)\left(\sum_{x'}p_{x'}(\theta)\sigma_{v}^{x'}(\theta)\otimes|x'\rangle\!\langle x'|_{h}\right)\right]\label{eq:reduction-2nd-term-qc-1}\\
 & =\sum_{x}p_{x}(\theta)\Tr\!\left[G_{v}^{j,x}\sigma_{v}^{x}(\theta)\right]\\
 & =\sum_{x}p_{x}(\theta)\left\langle G_{v}^{j,x}\right\rangle _{\sigma_{v}^{x}(\theta)},\label{eq:reduction-2nd-term-qc-last}
\end{align}
thus concluding the proof.
\end{proof}

\subsection{Quantum algorithm for estimating the gradient}

\label{subsec:Quantum-algorithm-qc-BM}In this brief subsection, I
sketch a quantum algorithm for estimating the gradient in \eqref{eq:gradient-qc-QBM}.
Here again I assume the ability to prepare thermal states, as is common
in theoretical work on quantum Boltzmann machines. In this case, I
assume the ability to prepare the quantum--classical state in \eqref{eq:qc-state}.
Estimating the second term $\sum_{x}p_{x}(\theta)\langle G_{v}^{j,x}\rangle_{\sigma_{v}^{x}(\theta)}$
in \eqref{eq:gradient-qc-QBM} then consists of preparing the state
$\sigma_{vh}(\theta)$ and measuring the observable $G_{v}^{j,x}$
if the value $x$ is observed in the classical hidden system. Estimating
the first term $\sum_{x}p_{x}(\theta)\langle G_{v}^{j,x}\rangle_{\Sigma_{v}^{\theta,x}(\rho)}$
in \eqref{eq:gradient-qc-QBM} is possible by the same procedure described
in Algorithm \ref{alg:estimate-grad-QBM}: however, in this case,
the hidden system is classical, which simplifies the implementation,
so that the Hamiltonian evolution $e^{-iG_{v}^{x}(\theta)t}$ can
be realized conditioned on the value $x$ in the classical hidden
system, as well as the observable $G_{v}^{j,x}$ to be measured. The
query complexity to a unitary that prepares a purification of $\sigma_{v}(\theta)$
is the same as that given in \eqref{eq:number-trials-main-alg}.

\subsection{Application to restricted quantum--classical Boltzmann machines}

\label{subsec:Application-to-restricted-qc-QBMs}In this section,
I apply the results of Section~\ref{subsec:Analytical-formula-for-gradient-qc-QBMs}
to restricted quantum--classical Boltzmann machines, which is a model
first established in \cite[Section~2.1]{Wiebe2019}. Restricted quantum--classical
Boltzmann machines are of the form introduced in \eqref{eq:rQBM-def},
with the exception that the tuple $\left(H_{j}\right)_{j=1}^{n}$
forms a commuting tuple. Under this assumption, the Hamiltonian $G(\theta)$
in \eqref{eq:rQBM-def} takes on a particular form, as presented in
Lemma \ref{lem:rqcBM-Hamiltonian} below.
\begin{lem}
\label{lem:rqcBM-Hamiltonian}Suppose that $\left(H_{j}\right)_{j=1}^{n}$
forms a commuting tuple. Then it follows that
\begin{equation}
H_{j}=\sum_{x}h_{j,x}|x\rangle\!\langle x|,
\end{equation}
where $\left\{ |x\rangle\right\} _{x}$ is an orthonormal basis and
$h_{j,x}\in\mathbb{R}$ for all $j$ and $x$, and the Hamiltonian
$G(\theta)$ can be written as
\begin{equation}
G(\theta)=\sum_{x}G_{v}^{x}(\theta)\otimes|x\rangle\!\langle x|,
\end{equation}
where
\begin{equation}
G_{v}^{x}(\theta)\coloneqq\sum_{j=1}^{n}b_{j}h_{j,x}I+\sum_{i=1}^{m}\left(a_{i}+\sum_{j=1}^{n}w_{i,j}h_{j,x}\right)V_{i}.\label{eq:G_x-r-qc-BMs}
\end{equation}
\end{lem}

\begin{proof}
Consider that
\begin{align}
G(\theta) & =\sum_{i=1}^{m}a_{i}V_{i}\otimes I+I\otimes\sum_{j=1}^{n}b_{j}H_{j}+\sum_{i=1}^{m}\sum_{j=1}^{n}w_{i,j}V_{i}\otimes H_{j}\\
 & =\sum_{i=1}^{m}a_{i}V_{i}\otimes\sum_{x}|x\rangle\!\langle x|+I\otimes\sum_{j=1}^{n}b_{j}\left(\sum_{x}h_{j,x}|x\rangle\!\langle x|\right)\nonumber \\
 & \qquad+\sum_{i=1}^{m}\sum_{j=1}^{n}w_{i,j}V_{i}\otimes\left(\sum_{x}h_{j,x}|x\rangle\!\langle x|\right)\\
 & =\sum_{x}\sum_{i=1}^{m}a_{i}V_{i}\otimes|x\rangle\!\langle x|+\sum_{x}\sum_{j=1}^{n}b_{j}h_{j,x}I\otimes|x\rangle\!\langle x|\nonumber \\
 & \qquad+\sum_{x}\sum_{i=1}^{m}\sum_{j=1}^{n}w_{i,j}h_{j,x}V_{i}\otimes|x\rangle\!\langle x|\\
 & =\sum_{x}\left(\sum_{i=1}^{m}a_{i}V_{i}+\sum_{j=1}^{n}b_{j}h_{j,x}I+\sum_{i=1}^{m}\sum_{j=1}^{n}w_{i,j}h_{j,x}V_{i}\right)\otimes|x\rangle\!\langle x|\\
 & =\sum_{x}\left(\sum_{j=1}^{n}b_{j}h_{j,x}I+\sum_{i=1}^{m}\left(a_{i}+\sum_{j=1}^{n}w_{i,j}h_{j,x}\right)V_{i}\right)\otimes|x\rangle\!\langle x|\\
 & =\sum_{x}G_{v}^{x}(\theta)\otimes|x\rangle\!\langle x|,
\end{align}
thus concluding the proof.
\end{proof}
For each $x$, define the following thermal state:
\begin{align}
\sigma_{v}^{x}(\theta) & \coloneqq\frac{e^{-G_{v}^{x}(\theta)}}{Z_{x}(\theta)},\\
Z_{x}(\theta) & \coloneqq\Tr\!\left[e^{-G_{v}^{x}(\theta)}\right],
\end{align}
and observe that
\begin{align}
\sigma_{v}^{x}(\theta) & =\frac{e^{-\tilde{G}_{v}^{x}(\theta)}}{\tilde{Z}_{x}(\theta)},\label{eq:modified-sig-x-state-1}
\end{align}
where
\begin{align}
\tilde{G}_{v}^{x}(\theta) & \coloneqq\sum_{i=1}^{m}\left(a_{i}+\sum_{j=1}^{n}w_{i,j}h_{j,x}\right)V_{i},\\
\tilde{Z}_{x}(\theta) & \coloneqq\Tr\!\left[e^{-\tilde{G}_{v}^{x}(\theta)}\right],
\end{align}
so that $\sigma_{v}^{x}(\theta)$ has no dependence on the parameter
vector $b$. For $t\in\mathbb{R}$, observe that
\begin{align}
e^{-iG_{v}^{x}(\theta)t}Ye^{iG_{v}^{x}(\theta)t} & =e^{-i\tilde{G}_{v}^{x}(\theta)t}Ye^{i\tilde{G}_{v}^{x}(\theta)t},\label{eq:rqcBM-unitary-evol}
\end{align}
so that this unitary channel also has no dependence on $b$. As such,
both the thermal state $\sigma_{v}^{x}(\theta)$ and the related unitary
evolution in \eqref{eq:rqcBM-unitary-evol} simplify for restricted
quantum--classical Boltzmann machines, having no dependence on the
parameter vector $b$.

In the following theorem, I show what the gradient formula in \eqref{eq:gradient-qc-QBM}
evaluates to for restricted quantum--classical Boltzmann machines.
These expressions generalize those from \eqref{eq:rBMs-gradient-1}--\eqref{eq:rBMs-gradient-last}
for classical restricted Boltzmann machines.
\begin{thm}
The partial derivatives of $D(\rho\|\sigma_{v}(\theta))$ are as follows:
\begin{align}
\frac{\partial}{\partial a_{i}}D(\rho\|\sigma_{v}(\theta)) & =\sum_{x}p_{x}(\theta)\left\langle V_{i}\right\rangle _{\Sigma_{v}^{\theta,x}(\rho)}-\sum_{x}p_{x}(\theta)\left\langle V_{i}\right\rangle _{\sigma_{v}^{x}(\theta)}\label{eq:gradient-restricted-qc-BM-1}\\
\frac{\partial}{\partial b_{j}}D(\rho\|\sigma_{v}(\theta)) & =\sum_{x}q_{x}(\theta)h_{j,x}-\sum_{x}p_{x}(\theta)h_{j,x},\label{eq:gradient-restricted-qc-BM-2}\\
\frac{\partial}{\partial w_{i,j}}D(\rho\|\sigma_{v}(\theta)) & =\sum_{x}r_{x}h_{j,x}\left\langle V_{i}\right\rangle _{\Sigma_{v}^{\theta,x}(\rho)}-\sum_{x}p_{x}(\theta)h_{j,x}\left\langle V_{i}\right\rangle _{\sigma_{x}(\theta)},
\end{align}
where $\sigma_{x}(\theta)$ is given in \eqref{eq:modified-sig-x-state-1},
$p_{x}(\theta)$ is defined from \eqref{eq:p_x_theta_def} and \eqref{eq:G_x-r-qc-BMs},
$\Sigma_{v}^{\theta,x}$and $\Xi_{v}^{\theta,x}$ are as given in
\eqref{eq:Sigma-map-x-dependent} and \eqref{eq:Xi-map-x-dependent},
respectively, and the quantum channel $\Phi_{v}^{\theta,x}$ can be
simplified as follows:
\begin{equation}
\Phi_{v}^{\theta,x}(Y_{v})\coloneqq\int_{-\infty}^{\infty}dt\,\gamma(t)\,e^{-i\tilde{G}_{v}^{x}(\theta)t}Y_{v}e^{i\tilde{G}_{v}^{x}(\theta)t}.
\end{equation}
Furthermore, the probability distribution $q_{x}(\theta)$ is given
by
\begin{equation}
q_{x}(\theta)\coloneqq\Tr\!\left[\Lambda_{x}(\theta)\rho\right],
\end{equation}
where $\left(\Lambda_{x}(\theta)\right)_{x}$ is a positive operator-valued
measure (POVM), and each measurement operator $\Lambda_{x}$ is defined
as
\begin{equation}
\Lambda_{x}\coloneqq\Upsilon_{v}^{\theta}\!\left(\sigma_{v}(\theta)^{-\frac{1}{2}}p_{x}(\theta)\sigma_{v}^{x}(\theta)\sigma_{v}(\theta)^{-\frac{1}{2}}\right),
\label{eq:rotated-pretty-good}
\end{equation}
and the channel $\Upsilon_{v}^{\theta}$ is defined in~\eqref{eq:upsilon-channel}.
\end{thm}

\begin{proof}
These equalities follow in large part directly from plugging \eqref{eq:G_x-r-qc-BMs}
into \eqref{eq:gradient-qc-QBM} and simplifying. To understand where
the first term in \eqref{eq:gradient-restricted-qc-BM-2} comes from,
consider that
\begin{align}
\left\langle I\right\rangle _{\Sigma_{v}^{\theta,x}(\rho)} & =\Tr\!\left[\Sigma_{v}^{\theta,x}(\rho)\right]\\
 & =\Tr\!\left[\left(\Phi_{v}^{\theta,x}\circ\Xi_{v}^{\theta,x}\circ\Upsilon_{v}^{\theta}\right)(\rho)\right]\\
 & =\Tr\!\left[\left(\Xi_{v}^{\theta,x}\circ\Upsilon_{v}^{\theta}\right)(\rho)\right]\\
 & =\Tr\!\left[\frac{1}{2}\left\{ \sigma_{v}^{x}(\theta),\sigma_{v}(\theta)^{-\frac{1}{2}}\Upsilon_{v}^{\theta}(\rho)\sigma_{v}(\theta)^{-\frac{1}{2}}\right\} \right]\\
 & =\Tr\!\left[\sigma_{v}^{x}(\theta)\sigma_{v}(\theta)^{-\frac{1}{2}}\Upsilon_{v}^{\theta}(\rho)\sigma_{v}(\theta)^{-\frac{1}{2}}\right]\\
 & =\Tr\!\left[\Upsilon_{v}^{\theta}\!\left(\sigma_{v}(\theta)^{-\frac{1}{2}}\sigma_{v}^{x}(\theta)\sigma_{v}(\theta)^{-\frac{1}{2}}\right)\rho\right],
\end{align}
so that
\begin{align}
p_{x}(\theta)\left\langle I\right\rangle _{\Sigma_{v}^{\theta,x}(\rho)} & =p_{x}(\theta)\Tr\!\left[\Upsilon_{v}^{\theta}\!\left(\sigma_{v}(\theta)^{-\frac{1}{2}}\sigma_{v}^{x}(\theta)\sigma_{v}(\theta)^{-\frac{1}{2}}\right)\rho\right],\\
 & =\Tr\!\left[\Lambda_{x}(\theta)\rho\right].
\end{align}
The set $\left(\Lambda_{x}(\theta)\right)_{x}$ forms a legitimate
POVM, essentially a twirled version of the well known pretty good
measurement \cite{Belavkin75,Belavkin75a,HolPGM78,Hausladen93bach,Hausladen1994},
because
\begin{align}
\sum_{x}\Lambda_{x}(\theta) & =\sum_{x}\Upsilon_{v}^{\theta}\!\left(\sigma_{v}(\theta)^{-\frac{1}{2}}p_{x}(\theta)\sigma_{v}^{x}(\theta)\sigma_{v}(\theta)^{-\frac{1}{2}}\right)\\
 & =\Upsilon_{v}^{\theta}\!\left(\sigma_{v}(\theta)^{-\frac{1}{2}}\left(\sum_{x}p_{x}(\theta)\sigma_{v}^{x}(\theta)\right)\sigma_{v}(\theta)^{-\frac{1}{2}}\right)\\
 & =\Upsilon_{v}^{\theta}\!\left(\sigma_{v}(\theta)^{-\frac{1}{2}}\sigma_{v}(\theta)\sigma_{v}(\theta)^{-\frac{1}{2}}\right)\\
 & =\Upsilon_{v}^{\theta}\!\left(I\right)\\
 & =I,
\end{align}
thus concluding the proof.
\end{proof}

\begin{rem}
    I note here that the POVM  in \eqref{eq:rotated-pretty-good} was originally put forward in \cite[Remark~5.7]{Wilde2015} and \cite[Corollary~4.3]{Junge2018}. This decoder was rediscovered more recently in \cite{Beigi2024}.
\end{rem}

\section{Generative modeling using classical--quantum Boltzmann machines}

\label{sec:Generative-modeling-using-cq-BMs}In this section, I consider
the generative modeling task when using a classical--quantum Boltzmann
machine, i.e., one for which the visible system is classical and the
hidden system is quantum. This kind of model has been considered in
recent work on the theory of quantum Boltzmann machines \cite{Demidik2025,Kimura2025,Vishnu2025,Demidik2025a}.
Here I provide an analytical formula for the gradient of the quantum
relative entropy in quantum state learning, when using classical--quantum
Boltzmann machines (Theorem~\ref{thm:gradient-cq-BMs}), which can
be understood as a conditional version of the gradient formula in
\eqref{eq:no-hidden-unit-consistency-2}, the latter applying when
a quantum Boltzmann machine has visible units only. See also \cite{Demidik2025}
in this context. Additionally, I present a quantum algorithm for estimating
the gradient (Section~\ref{subsec:Quantum-algorithm-cq-BMs}), which
is simpler than the quantum algorithm from Section~\ref{subsec:Q-algorithm-state-learning-QBM}.
I also show how the formula simplifies when considering a restricted
classical--quantum Boltzmann machine (Section~\ref{subsec:Application-to-restricted-cq-BMs}).

In generative modeling, the target state is classical, i.e., diagonal
in an orthonormal basis $\left\{ |x\rangle\right\} _{x}$:
\begin{equation}
\rho=\sum_{x}r_{x}|x\rangle\!\langle x|_{v},\label{eq:diagonal-target-state}
\end{equation}
where $\left(r_{x}\right)_{x}$ is a probability distribution. 

The most general form for a classical--quantum Boltzmann machine
involves a parameterized Hamiltonian $G(\theta)$ of the form in \eqref{eq:gen-param-ham},
but with each term $G_{j}$ having a classical visible system, so
that it can be written as
\begin{equation}
G_{j}\coloneqq\sum_{x}|x\rangle\!\langle x|_{v}\otimes G_{h}^{j,x},\label{eq:term-commuting-visible}
\end{equation}
where each $G_{h}^{j,x}$ is a Hermitian operator acting on the hidden
system. This assumption implies that
\begin{align}
G(\theta) & =\sum_{j=1}^{J}\theta_{j}\left(\sum_{x}|x\rangle\!\langle x|_{v}\otimes G_{h}^{j,x}\right),\\
 & =\sum_{x}|x\rangle\!\langle x|_{v}\otimes G_{h}^{x}(\theta),
\end{align}
where
\begin{align*}
G_{h}^{x}(\theta) & \coloneqq\sum_{j=1}^{J}\theta_{j}G_{h}^{j,x}.
\end{align*}
The thermal state in \eqref{eq:def-QBM-state-vh} then has the following
classical--quantum form:
\begin{equation}
\sigma_{vh}(\theta)=\sum_{x}p_{x}(\theta)|x\rangle\!\langle x|_{v}\otimes\sigma_{h}^{x}(\theta),
\end{equation}
where
\begin{align}
\sigma_{h}^{x}(\theta) & \coloneqq\frac{e^{-G_{h}^{x}(\theta)}}{\Tr[e^{-G_{h}^{x}(\theta)}]},\\
p_{x}(\theta) & \coloneqq\frac{\Tr[e^{-G_{h}^{x}(\theta)}]}{\Tr[e^{-G(\theta)}]}=\frac{\Tr[e^{-G_{h}^{x}(\theta)}]}{\sum_{x}\Tr[e^{-G_{h}^{x}(\theta)}]}.\label{eq:prob-dist-visible-cq-BM}
\end{align}
Observe that the reduced state of the visible system is given by
\begin{equation}
\sigma_{v}(\theta)=\Tr_{h}\!\left[\sigma_{vh}(\theta)\right]=\sum_{x}p_{x}(\theta)|x\rangle\!\langle x|_{v}.\label{eq:diagonal-visible-state}
\end{equation}

\subsection{Analytical formula for the gradient of classical--quantum Boltzmann
machines}

\label{subsec:Analytical-formula-for-gradient-cq-QBMs}Theorem~\ref{thm:gradient-cq-BMs}
below provides an analytical expression for the gradient of \eqref{eq:q-rel-ent-def},
under the model assumptions presented in \eqref{eq:term-commuting-visible}--\eqref{eq:diagonal-visible-state}.
Given that both the target state and the reduced state of the visible
system is classical, the quantum relative entropy reduces to the classical
relative entropy:
\begin{equation}
D(\rho\|\sigma_{v}(\theta))=\sum_{x}r_{x}\ln\!\left(\frac{r_{x}}{p_{x}(\theta)}\right),
\end{equation}
where $r_{x}$ is the probability distribution in \eqref{eq:diagonal-target-state}
and $p_{x}(\theta)$ is given in \eqref{eq:prob-dist-visible-cq-BM}.
The gradient formula in \eqref{eq:gradient-cq-QBM} is a quantum generalization
of the expression in \eqref{eq:classical-gradient}, reducing to it
in the fully classical case. It can also be understood as a generalization
of the formula in \eqref{eq:no-hidden-unit-consistency-2} for the
gradient when there are visibly units only: in \eqref{eq:gradient-cq-QBM},
there is extra conditioning on the classical value $x$. Each term
in \eqref{eq:gradient-cq-QBM} can be estimated by first sampling
$x$ from either $r_{x}$ or $p_{x}(\theta)$, and then estimating
$\langle G_{v}^{j,x}\rangle_{\sigma_{v}^{x}(\theta)}$.
\begin{thm}
\label{thm:gradient-cq-BMs}Let $\rho$ be a diagonal target state
of the form in \eqref{eq:diagonal-target-state}, and let $\sigma_{v}(\theta)$
be a diagonal model state of the form in \eqref{eq:diagonal-visible-state}.
The partial derivatives of $D(\rho\|\sigma_{v}(\theta))$ are as follows:
\begin{equation}
\frac{\partial}{\partial\theta_{j}}D(\rho\|\sigma_{v}(\theta))=\sum_{x}r_{x}\left\langle G_{h}^{j,x}\right\rangle _{\sigma_{h}^{x}(\theta)}-\sum_{x}p_{x}(\theta)\left\langle G_{h}^{j,x}\right\rangle _{\sigma_{h}^{x}(\theta)}.\label{eq:gradient-cq-QBM}
\end{equation}
\end{thm}

\begin{proof}
Consider that
\begin{align}
\left\langle G_{j}\right\rangle _{\sigma_{vh}(\theta)} & =\Tr\!\left[G_{j}\sigma_{vh}(\theta)\right]\label{eq:reduction-2nd-term-cq-1}\\
 & =\Tr\!\left[\left(\sum_{x}|x\rangle\!\langle x|_{v}\otimes G_{h}^{j,x}\right)\left(\sum_{x'}p_{x'}(\theta)|x'\rangle\!\langle x'|_{v}\otimes\sigma_{h}^{x'}(\theta)\right)\right]\\
 & =\sum_{x}p_{x}(\theta)\Tr\!\left[G_{h}^{j,x}\sigma_{h}^{x}(\theta)\right]\\
 & =\sum_{x}p_{x}(\theta)\left\langle G_{h}^{j,x}\right\rangle _{\sigma_{h}^{x}(\theta)}.\label{eq:reduction-2nd-term-cq-last}
\end{align}
Furthermore, observe that
\begin{equation}
\sigma_{v}(\theta)^{-\frac{1}{2}}\Upsilon_{v}^{\theta}(\rho)\sigma_{v}(\theta)^{-\frac{1}{2}}=\sum_{x}\frac{r_{x}}{p_{x}(\theta)}|x\rangle\!\langle x|_{v},
\end{equation}
implying that
\begin{align}
 & \Tr\!\left[\Phi_{vh}^{\theta}(G_{j})\frac{1}{2}\left\{ \sigma_{vh}(\theta),\sigma_{v}(\theta)^{-\frac{1}{2}}\Upsilon_{v}^{\theta}(\rho)\sigma_{v}(\theta)^{-\frac{1}{2}}\otimes I_{h}\right\} \right]\nonumber \\
 & =\Tr\!\left[\Phi_{vh}^{\theta}(G_{j})\frac{1}{2}\left\{ \sum_{x}p_{x}(\theta)|x\rangle\!\langle x|_{v}\otimes\sigma_{h}^{x}(\theta),\sum_{x}\frac{r_{x}}{p_{x}(\theta)}|x\rangle\!\langle x|_{v}\otimes I_{h}\right\} \right]\\
 & =\Tr\!\left[\Phi_{vh}^{\theta}\!\left(\sum_{x}|x\rangle\!\langle x|_{v}\otimes G_{h}^{j,x}\right)\left(\sum_{x}r_{x}|x\rangle\!\langle x|_{v}\otimes\sigma_{h}^{x}(\theta)\right)\right]\\
 & =\int_{-\infty}^{\infty}dt\,\gamma(t)\:\Tr\!\left[\begin{array}{c}
\left(\sum_{x}|x\rangle\!\langle x|_{v}\otimes e^{-iG_{h}^{x}(\theta)t}G_{h}^{j,x}e^{iG_{h}^{x}(\theta)t}\right)\times\\
\left(\sum_{x}r_{x}|x\rangle\!\langle x|_{v}\otimes\sigma_{h}^{x}(\theta)\right)
\end{array}\right]\\
 & =\int_{-\infty}^{\infty}dt\,\gamma(t)\:\sum_{x}r_{x}\Tr\!\left[e^{-iG_{h}^{x}(\theta)t}G_{h}^{j,x}e^{iG_{h}^{x}(\theta)t}\sigma_{h}^{x}(\theta)\right]\\
 & =\int_{-\infty}^{\infty}dt\,\gamma(t)\:\sum_{x}r_{x}\Tr\!\left[G_{h}^{j,x}e^{iG_{h}^{x}(\theta)t}\sigma_{h}^{x}(\theta)e^{-iG_{h}^{x}(\theta)t}\right]\\
 & =\sum_{x}r_{x}\Tr\!\left[G_{h}^{j,x}\sigma_{h}^{x}(\theta)\right]\\
 & =\sum_{x}r_{x}\left\langle G_{h}^{j,x}\right\rangle _{\sigma_{h}^{x}(\theta)},
\end{align}
thus concluding the proof. The penultimate equality follows because
$G_{h}^{x}(\theta)$ and $\sigma_{h}^{x}(\theta)$ commute.
\end{proof}

\subsection{Quantum algorithm for estimating the gradient}

\label{subsec:Quantum-algorithm-cq-BMs}In this case, an algorithm
for estimating the gradient in \eqref{eq:gradient-cq-QBM} is simpler
than that presented in Sections \ref{subsec:Q-algorithm-state-learning-QBM}
and \ref{subsec:Quantum-algorithm-qc-BM}. This follows because the
visible system is classical and due to the form of the gradient in
\eqref{eq:gradient-cq-QBM}. Indeed, one can estimate the first term
$\sum_{x}r_{x}\langle G_{h}^{j,x}\rangle_{\sigma_{h}^{x}(\theta)}$
in \eqref{eq:gradient-cq-QBM} by sampling $x$ from the target distribution
$r_{x}$, preparing the thermal state $\sigma_{h}^{x}(\theta)$, and
measuring the observable $G_{h}^{j,x}$. Similarly, the second term
$\sum_{x}p_{x}(\theta)\langle G_{h}^{j,x}\rangle_{\sigma_{h}^{x}(\theta)}$
in \eqref{eq:gradient-cq-QBM} can be estimated by sampling $x$ from
$p_{x}(\theta)$, preparing the thermal state $\sigma_{h}^{x}(\theta)$,
and measuring the observable $G_{h}^{j,x}$.

\subsection{Application to restricted classical--quantum Boltzmann machines}

\label{subsec:Application-to-restricted-cq-BMs}In this section, I
apply the results of Section~\ref{subsec:Analytical-formula-for-gradient-cq-QBMs}
to restricted classical--quantum Boltzmann machines, which is a model
introduced in \cite[Section~2.1]{Demidik2025} and called semi-quantum
restricted Boltzmann machines therein. Restricted classical--quantum
Boltzmann machines are of the form introduced in \eqref{eq:rQBM-def},
with the exception that the tuple $\left(V_{i}\right)_{i=1}^{m}$
forms a commuting tuple. Under this assumption, the Hamiltonian $G(\theta)$
in \eqref{eq:rQBM-def} takes on a particular form, as presented in
Lemma \ref{lem:rcqBM-Hamiltonian} below.
\begin{lem}
\label{lem:rcqBM-Hamiltonian}Suppose that $\left(V_{i}\right)_{i=1}^{m}$
forms a commuting tuple. Then it follows that
\begin{equation}
V_{i}=\sum_{x}v_{i,x}|x\rangle\!\langle x|,\label{eq:visible-rcqBM-form}
\end{equation}
where $\left\{ |x\rangle\right\} _{x}$ is an orthonormal basis and
$v_{i,x}\in\mathbb{R}$ for all $i$ and $x$, and the Hamiltonian
$G(\theta)$ can be written as
\begin{equation}
G(\theta)=\sum_{x}|x\rangle\!\langle x|\otimes G_{x}(\theta),
\end{equation}
where
\begin{equation}
G_{x}(\theta)\coloneqq\sum_{i=1}^{m}a_{i}v_{i,x}I+\sum_{j=1}^{n}\left(b_{j}+\sum_{i=1}^{m}w_{i,j}v_{i,x}\right)H_{j}.\label{eq:conditioned-Ham-visible-commuting}
\end{equation}
\end{lem}

\begin{proof}
By plugging \eqref{eq:visible-rcqBM-form} into \eqref{eq:rQBM-def},
consider that
\begin{align}
G(\theta) & =\sum_{i=1}^{m}a_{i}V_{i}\otimes I+I\otimes\sum_{j=1}^{n}b_{j}H_{j}+\sum_{i=1}^{m}\sum_{j=1}^{n}w_{i,j}V_{i}\otimes H_{j}\\
 & =\sum_{i=1}^{m}a_{i}\left(\sum_{x}v_{i,x}|x\rangle\!\langle x|\right)\otimes I+\sum_{x}|x\rangle\!\langle x|\otimes\sum_{j=1}^{n}b_{j}H_{j}\nonumber \\
 & \qquad+\sum_{i=1}^{m}\sum_{j=1}^{n}w_{i,j}\left(\sum_{x}v_{i,x}|x\rangle\!\langle x|\right)\otimes H_{j}\\
 & =\sum_{x}|x\rangle\!\langle x|\otimes\sum_{i=1}^{m}a_{i}v_{i,x}I+\sum_{x}|x\rangle\!\langle x|\otimes\sum_{j=1}^{n}b_{j}H_{j}\nonumber \\
 & \qquad+\sum_{x}|x\rangle\!\langle x|\otimes\sum_{i=1}^{m}\sum_{j=1}^{n}w_{i,j}v_{i,x}H_{j}\\
 & =\sum_{x}|x\rangle\!\langle x|\otimes\left(\sum_{i=1}^{m}a_{i}v_{i,x}I+\sum_{j=1}^{n}b_{j}H_{j}+\sum_{i=1}^{m}\sum_{j=1}^{n}w_{i,j}v_{i,x}H_{j}\right)\\
 & =\sum_{x}|x\rangle\!\langle x|\otimes\left(\sum_{i=1}^{m}a_{i}v_{i,x}I+\sum_{j=1}^{n}\left(b_{j}+\sum_{i=1}^{m}w_{i,j}v_{i,x}\right)H_{j}\right)\\
 & =\sum_{x}|x\rangle\!\langle x|\otimes G_{x}(\theta),
\end{align}
thus concluding the proof.
\end{proof}
For each $x$, define the following thermal state:
\begin{align}
\sigma_{h}^{x}(\theta) & \coloneqq\frac{e^{-G_{x}(\theta)}}{Z_{x}(\theta)},\\
Z_{x}(\theta) & \coloneqq\Tr\!\left[e^{-G_{x}(\theta)}\right],
\end{align}
and observe that
\begin{align}
\sigma_{h}^{x}(\theta) & =\frac{e^{-\tilde{G}_{h}^{x}(\theta)}}{\tilde{Z}_{x}(\theta)},\label{eq:modified-sig-x-state}
\end{align}
where
\begin{align}
\tilde{G}_{h}^{x}(\theta) & \coloneqq\sum_{j=1}^{n}\left(b_{j}+\sum_{i=1}^{m}w_{i,j}v_{i,x}\right)H_{j},\\
\tilde{Z}_{x}(\theta) & \coloneqq\Tr\!\left[e^{-\tilde{G}_{h}^{x}(\theta)}\right],
\end{align}
so that $\sigma_{h}^{x}(\theta)$ has no dependence on the parameter
vector $a$.

In the following theorem, I show what the gradient formula in \eqref{eq:gradient-qc-QBM}
evaluates to for restricted classical--quantum Boltzmann machines.
These expressions generalize those from \eqref{eq:rBMs-gradient-1}--\eqref{eq:rBMs-gradient-last},
for classical restricted Boltzmann machines.
\begin{thm}
The partial derivatives of $D(\rho\|\sigma_{v}(\theta))$ are as follows:
\begin{align}
\frac{\partial}{\partial a_{i}}D(\rho\|\sigma_{v}(\theta)) & =\sum_{x}r_{x}v_{i,x}-\sum_{x}p_{x}(\theta)v_{i,x},\label{eq:gradient-restricted-cq-QBM}\\
\frac{\partial}{\partial b_{j}}D(\rho\|\sigma_{v}(\theta)) & =\sum_{x}r_{x}\left\langle H_{j}\right\rangle _{\sigma_{h}^{x}(\theta)}-\sum_{x}p_{x}(\theta)\left\langle H_{j}\right\rangle _{\sigma_{h}^{x}(\theta)},\\
\frac{\partial}{\partial w_{i,j}}D(\rho\|\sigma_{v}(\theta)) & =\sum_{x}r_{x}v_{i,x}\left\langle H_{j}\right\rangle _{\sigma_{h}^{x}(\theta)}-\sum_{x}p_{x}(\theta)v_{i,x}\left\langle H_{j}\right\rangle _{\sigma_{h}^{x}(\theta)},
\end{align}
where $\sigma_{h}^{x}(\theta)$ is given in \eqref{eq:modified-sig-x-state}.
\end{thm}

\begin{proof}
These formulas follow by plugging \eqref{eq:conditioned-Ham-visible-commuting}
into \eqref{eq:gradient-cq-QBM} and simplifying.
\end{proof}

\section{Extensions to Petz--Tsallis relative entropies}

\label{sec:Extensions-to-Petz=002013Tsallis}In this section, I show
how to extend all of the results of this paper to the case when the
objective function to minimize is the Petz--Tsallis relative quasi-entropy.
For $q\in\left(0,1\right)\cup(1,2]$ and for states $\rho$ and $\sigma$,
the Petz--Tsallis relative quasi-entropy is defined as
\begin{align}
D_{q}(\rho\|\sigma) & \coloneqq\frac{Q_{q}(\rho\|\sigma)-1}{q-1},\\
Q_{q}(\rho\|\sigma) & \coloneqq\Tr\!\left[\rho^{q}\sigma^{1-q}\right],
\end{align}
and it obeys the data-processing inequality for these values of $q$
\cite{Petz1985,Petz1986}. Note that the Petz--Tsallis relative entropy
converges to the quantum relative entropy in the $q\to1$ limit:
\begin{equation}
D(\rho\|\sigma)=\lim_{q\to1}D_{q}(\rho\|\sigma).
\end{equation}
Consistent with this limit, all of the formulas that follow converge
to prior formulas presented in this paper in the $q\to1$ limit.

\subsection{Quantum Boltzmann machines}

The following theorem generalizes the gradient formula in \eqref{eq:gradient-fully-QBM}
for quantum Boltzmann machines to the case when the objective function
is the Petz--Tsallis relative entropy. The main distinction with
the formula in \eqref{eq:gradient-fully-QBM} is that various powers
involving the parameter $q$ appear in the expression in \eqref{eq:gradient-fully-QBM-petz-tsallis}
and the map $\Sigma_{v\to vh}^{\theta,q}$ in \eqref{eq:herm-pres-map-Sigma-petz-tsalis}
is Hermiticity preserving, but no longer trace preserving.
\begin{thm}
For all $q\in\left(0,1\right)\cup(1,2]$, the partial derivatives
of $D_{q}(\rho\|\sigma)$ are as follows:
\begin{equation}
\frac{\partial}{\partial\theta_{j}}D_{q}(\rho\|\sigma_{v}(\theta))=\left\langle G_{j}\right\rangle _{\Sigma_{v\to vh}^{\theta,q}(\rho^{q})}-Q_{q}(\rho\|\sigma_{v}(\theta))\left\langle G_{j}\right\rangle _{\sigma_{vh}(\theta)},\label{eq:gradient-fully-QBM-petz-tsallis}
\end{equation}
where 
\begin{align}
\Sigma_{v\to vh}^{\theta,q} & \coloneqq\Phi_{vh}^{\theta}\circ\Xi_{v\to vh}^{\theta,q}\circ\Upsilon_{v}^{\theta,1-q},\label{eq:herm-pres-map-Sigma-petz-tsalis}\\
\Xi_{v\to vh}^{\theta,q}(R_{v}) & \coloneqq\frac{1}{2}\left\{ \sigma_{vh}(\theta),\sigma_{v}(\theta)^{-\frac{q}{2}}R_{v}\sigma_{v}(\theta)^{-\frac{q}{2}}\otimes I_{h}\right\} ,\\
\Phi_{vh}^{\theta}(Y_{vh}) & \coloneqq\int_{-\infty}^{\infty}dt\,\gamma(t)\:e^{-iG(\theta)t}Y_{vh}e^{iG(\theta)t},\\
\Upsilon_{v}^{\theta,q}(X_{v}) & \coloneqq\int_{-\infty}^{\infty}dt\,\beta_{1-q}(t)\,\sigma_{v}(\theta)^{-\frac{it}{2}}X_{v}\sigma_{v}(\theta)^{\frac{it}{2}},
\end{align}
and the probability density function $\beta_{1-q}(t)$ is defined
in \eqref{eq:logistic-prob-dens-1}.
\end{thm}

\begin{proof}
Consider that
\begin{align}
 & \frac{\partial}{\partial\theta_{j}}D_{q}(\rho\|\sigma_{v}(\theta))\nonumber \\
 & =\frac{1}{q-1}\frac{\partial}{\partial\theta_{j}}\Tr\!\left[\rho^{q}\sigma_{v}(\theta)^{1-q}\right]\\
 & =\frac{1}{q-1}\Tr\!\left[\rho^{q}\frac{\partial}{\partial\theta_{j}}\sigma_{v}(\theta)^{1-q}\right]\\
 & \overset{(a)}{=}-\Tr\!\left[\rho^{q}\sigma_{v}(\theta)^{-\frac{q}{2}}\Upsilon_{v}^{\theta,q}\left(\frac{\partial}{\partial\theta_{j}}\sigma_{v}(\theta)\right)\sigma_{v}(\theta)^{-\frac{q}{2}}\right]\\
 & =-\Tr\!\left[\sigma_{v}(\theta)^{-\frac{q}{2}}\Upsilon_{v}^{\theta,q}(\rho^{q})\sigma_{v}(\theta)^{-\frac{q}{2}}\frac{\partial}{\partial\theta_{j}}\Tr_{h}\!\left[\sigma_{vh}(\theta)\right]\right]\\
 & =-\Tr\!\left[\sigma_{v}(\theta)^{-\frac{q}{2}}\Upsilon_{v}^{\theta,q}(\rho^{q})\sigma_{v}(\theta)^{-\frac{q}{2}}\Tr_{h}\!\left[\frac{\partial}{\partial\theta_{j}}\sigma_{vh}(\theta)\right]\right]\\
 & =-\Tr\!\left[\begin{array}{c}
\sigma_{v}(\theta)^{-\frac{q}{2}}\Upsilon_{v}^{\theta,q}(\rho^{q})\sigma_{v}(\theta)^{-\frac{q}{2}}\times\\
\Tr_{h}\!\left[-\frac{1}{2}\left\{ \Phi_{vh}^{\theta}(G_{j}),\sigma_{vh}(\theta)\right\} +\sigma_{vh}(\theta)\left\langle G_{j}\right\rangle _{\sigma_{vh}(\theta)}\right]
\end{array}\right]\\
 & =\frac{1}{2}\Tr\!\left[\sigma_{v}(\theta)^{-\frac{q}{2}}\Upsilon_{v}^{\theta,q}(\rho^{q})\sigma_{v}(\theta)^{-\frac{q}{2}}\Tr_{h}\!\left[\left\{ \Phi_{vh}^{\theta}(G_{j}),\sigma_{vh}(\theta)\right\} \right]\right]\nonumber \\
 & \qquad-\Tr\!\left[\sigma_{v}(\theta)^{-\frac{q}{2}}\Upsilon_{v}^{\theta,q}(\rho^{q})\sigma_{v}(\theta)^{-\frac{q}{2}}\Tr_{h}\!\left[\sigma_{vh}(\theta)\right]\right]\left\langle G_{j}\right\rangle _{\sigma_{vh}(\theta)}\\
 & =\frac{1}{2}\Tr\!\left[\left(\sigma_{v}(\theta)^{-\frac{q}{2}}\Upsilon_{v}^{\theta,q}(\rho^{q})\sigma_{v}(\theta)^{-\frac{q}{2}}\otimes I_{h}\right)\left\{ \Phi_{vh}^{\theta}(G_{j}),\sigma_{vh}(\theta)\right\} \right]\nonumber \\
 & \qquad-\Tr\!\left[\rho^{q}\Upsilon_{v}^{\theta,1-q}(\sigma_{v}(\theta)^{1-q})\right]\left\langle G_{j}\right\rangle _{\sigma_{vh}(\theta)}\\
 & =\frac{1}{2}\Tr\!\left[G_{j}\Phi_{vh}^{\theta}\!\left(\left\{ \sigma_{vh}(\theta),\sigma_{v}(\theta)^{-\frac{q}{2}}\Upsilon_{v}^{\theta,q}(\rho^{q})\sigma_{v}(\theta)^{-\frac{q}{2}}\otimes I_{h}\right\} \right)\right]\nonumber \\
 & \qquad-\Tr\!\left[\rho^{q}\sigma_{v}(\theta)^{1-q}\right]\left\langle G_{j}\right\rangle _{\sigma_{vh}(\theta)}\\
 & =\left\langle G_{j}\right\rangle _{\Sigma_{v\to vh}^{\theta,q}(\rho^{q})}-Q_{q}(\rho\|\sigma_{v}(\theta))\left\langle G_{j}\right\rangle _{\sigma_{vh}(\theta)},
\end{align}
thus concluding the proof. The reasoning behind the steps is similar
to that given in the proof of Theorem~\ref{thm:q-state-learning-gradient-vh-q}.
The equality $(a)$ follows from \eqref{eq:deriv-matrix-power-fourier}.
\end{proof}

\subsection{Quantum--classical Boltzmann machines}

The following theorem generalizes the gradient formula in \eqref{eq:gradient-qc-QBM}
for quantum--classical Boltzmann machines to the case when the objective
function is the Petz--Tsallis relative entropy:
\begin{thm}
For all $q\in\left(0,1\right)\cup(1,2]$, the partial derivatives
of $D_{q}(\rho\|\sigma_{v}(\theta))$ are as follows:
\begin{multline}
\frac{\partial}{\partial\theta_{j}}D_{q}(\rho\|\sigma_{v}(\theta))=\sum_{x}p_{x}(\theta)\left\langle G_{v}^{j,x}\right\rangle _{\Sigma_{v}^{\theta,q,x}(\rho^{q})}\\
-Q_{q}(\rho\|\sigma_{v}(\theta))\sum_{x}p_{x}(\theta)\left\langle G_{v}^{j,x}\right\rangle _{\sigma_{v}^{x}(\theta)},\label{eq:gradient-qc-BM-petz-tsallis}
\end{multline}
where $\sigma_{v}^{x}(\theta)$ is defined in \eqref{eq:sigma-x-v-def}
and 
\begin{align}
\Sigma_{v}^{\theta,q,x} & \coloneqq\Phi_{v}^{\theta,x}\circ\Xi_{v}^{\theta,q,x}\circ\Upsilon_{v}^{\theta,q},\\
\Xi_{v\to v}^{\theta,q,x}(R_{v}) & \coloneqq\frac{1}{2}\left\{ \sigma_{v}^{x}(\theta),\sigma_{v}(\theta)^{-\frac{q}{2}}R_{v}\sigma_{v}(\theta)^{-\frac{q}{2}}\right\} ,\\
\Phi_{v}^{\theta,x}(Y_{v}) & \coloneqq\int_{-\infty}^{\infty}dt\,\gamma(t)\,e^{-iG_{v}^{x}(\theta)t}Y_{v}e^{iG_{v}^{x}(\theta)t}.
\end{align}
\end{thm}

\begin{proof}
Consider that
\begin{align}
 & \Sigma_{v\to vh}^{\theta,q}(\rho^{q})\nonumber \\
 & =\Phi_{vh}^{\theta}\!\left(\frac{1}{2}\left\{ \sigma_{vh}(\theta),\sigma_{v}(\theta)^{-\frac{q}{2}}\Upsilon_{v}^{\theta,q}(\rho^{q})\sigma_{v}(\theta)^{-\frac{q}{2}}\otimes I_{h}\right\} \right),\\
 & =\Phi_{vh}^{\theta}\!\left(\frac{1}{2}\left\{ \sum_{x}p_{x}(\theta)\sigma_{v}^{x}(\theta)\otimes|x\rangle\!\langle x|_{h},\sigma_{v}(\theta)^{-\frac{q}{2}}\Upsilon_{v}^{\theta,q}(\rho)\sigma_{v}(\theta)^{-\frac{q}{2}}\otimes I_{h}\right\} \right)\\
 & =\Phi_{vh}^{\theta}\!\left(\sum_{x}p_{x}(\theta)\frac{1}{2}\left\{ \sigma_{v}^{x}(\theta),\sigma_{v}(\theta)^{-\frac{q}{2}}\Upsilon_{v}^{\theta,q}(\rho)\sigma_{v}(\theta)^{-\frac{q}{2}}\right\} \otimes|x\rangle\!\langle x|_{h}\right)\\
 & =\int_{-\infty}^{\infty}dt\,\gamma(t)\:e^{-iG(\theta)t}\times\nonumber \\
 & \quad\left(\sum_{x}p_{x}(\theta)\frac{1}{2}\left\{ \sigma_{v}^{x}(\theta),\sigma_{v}(\theta)^{-\frac{q}{2}}\Upsilon_{v}^{\theta,q}(\rho^{q})\sigma_{v}(\theta)^{-\frac{q}{2}}\right\} \otimes|x\rangle\!\langle x|_{h}\right)e^{iG(\theta)t}\\
 & =\int_{-\infty}^{\infty}dt\,\gamma(t)\:\left(\sum_{x'}e^{-iG_{v}^{x'}(\theta)t}\otimes|x'\rangle\!\langle x'|_{h}\right)\times\nonumber \\
 & \qquad\left(\sum_{x}p_{x}(\theta)\frac{1}{2}\left\{ \sigma_{v}^{x}(\theta),\sigma_{v}(\theta)^{-\frac{q}{2}}\Upsilon_{v}^{\theta,q}(\rho^{q})\sigma_{v}(\theta)^{-\frac{q}{2}}\right\} \otimes|x\rangle\!\langle x|_{h}\right)\times\nonumber \\
 & \qquad\left(\sum_{x''}e^{iG_{v}^{x''}(\theta)t}\otimes|x''\rangle\!\langle x''|_{h}\right)\\
 & =\int_{-\infty}^{\infty}dt\,\gamma(t)\:\sum_{x}p_{x}(\theta)e^{-iG_{v}^{x}(\theta)t}\times\nonumber \\
 & \qquad\left(\frac{1}{2}\left\{ \sigma_{v}^{x}(\theta),\sigma_{v}(\theta)^{-\frac{q}{2}}\Upsilon_{v}^{\theta,q}(\rho^{q})\sigma_{v}(\theta)^{-\frac{q}{2}}\right\} \right)e^{iG_{v}^{x}(\theta)t}\otimes|x\rangle\!\langle x|_{h}\\
 & =\sum_{x}p_{x}(\theta)\Phi_{v}^{\theta,x}\!\left(\frac{1}{2}\left\{ \sigma_{v}^{x}(\theta),\sigma_{v}(\theta)^{-\frac{q}{2}}\Upsilon_{v}^{\theta,q}(\rho^{q})\sigma_{v}(\theta)^{-\frac{q}{2}}\right\} \right)\otimes|x\rangle\!\langle x|_{h}.
\end{align}
Then it follows that
\begin{align}
 & \left\langle G_{j}\right\rangle _{\Sigma_{v\to vh}^{\theta,q}\!\left(\rho^{q}\right)}\nonumber \\
 & =\Tr\!\left[\begin{array}{c}
\left(\sum_{x'}G_{v}^{j,x'}\otimes|x'\rangle\!\langle x'|_{h}\right)\times\\
\left(\sum_{x}p_{x}(\theta)\Phi_{v}^{\theta,x}\!\left(\frac{1}{2}\left\{ \sigma_{v}^{x}(\theta),\sigma_{v}(\theta)^{-\frac{q}{2}}\Upsilon_{v}^{\theta,q}(\rho^{q})\sigma_{v}(\theta)^{-\frac{q}{2}}\right\} \right)\otimes|x\rangle\!\langle x|_{h}\right)
\end{array}\right]\\
 & =\sum_{x}p_{x}(\theta)\Tr\!\left[G_{v}^{j,x}\Phi_{v}^{\theta,x}\!\left(\frac{1}{2}\left\{ \sigma_{v}^{x}(\theta),\sigma_{v}(\theta)^{-\frac{q}{2}}\Upsilon_{v}^{\theta,q}(\rho^{q})\sigma_{v}(\theta)^{-\frac{q}{2}}\right\} \right)\right]\\
 & =\sum_{x}p_{x}(\theta)\Tr\!\left[G_{v}^{j,x}\Sigma_{v}^{\theta,q,x}(\rho^{q})\right]\\
 & =\sum_{x}p_{x}(\theta)\left\langle G_{v}^{j,x}\right\rangle _{\Sigma_{v}^{\theta,q,x}(\rho^{q})},
\end{align}
where
\begin{equation}
\Sigma_{v}^{\theta,q,x}(R_{v})\coloneqq\Phi_{G_{v}^{x}(\theta)}\!\left(\frac{1}{2}\left\{ \sigma_{v}^{x}(\theta),\sigma_{v}(\theta)^{-\frac{q}{2}}\Upsilon_{v}^{\theta,q}(R_{v})\sigma_{v}(\theta)^{-\frac{q}{2}}\right\} \right).
\end{equation}
Additionally, the following equality is a consequence of \eqref{eq:reduction-2nd-term-qc-1}--\eqref{eq:reduction-2nd-term-qc-last}:
\begin{align}
\left\langle G_{j}\right\rangle _{\sigma_{vh}(\theta)} & =\sum_{x}p_{x}(\theta)\left\langle G_{v}^{j,x}\right\rangle _{\sigma_{v}^{x}(\theta)},
\end{align}
thus concluding the proof.
\end{proof}

\subsection{Classical--quantum Boltzmann machines}

The following theorem generalizes the gradient formula in \eqref{eq:gradient-cq-QBM}
for classical--quantum Boltzmann machines to the case when the objective
function is the Petz--Tsallis relative entropy. However, since the
states being compared are classical, the Petz--Tsallis relative entropy
reduces to the Tsallis relative entropy:
\begin{equation}
D_{q}(\rho\|\sigma_{v}(\theta))=\frac{\sum_{x}r_{x}^{q}p_{x}(\theta)^{1-q}-1}{q-1}.
\end{equation}

\begin{thm}
For all $q\in\left(0,1\right)\cup(1,2]$, the partial derivatives
of $D_{q}(\rho\|\sigma_{v}(\theta))$ are as follows:
\begin{multline}
\frac{\partial}{\partial\theta_{j}}D_{q}(\rho\|\sigma_{v}(\theta))=\sum_{x}r_{x}^{q}p_{x}(\theta)^{1-q}\left\langle G_{h}^{j,x}\right\rangle _{\sigma_{h}^{x}(\theta)}\\
-\left(\sum_{x}r_{x}^{q}p_{x}(\theta)^{1-q}\right)\left(\sum_{x}p_{x}(\theta)\left\langle G_{h}^{j,x}\right\rangle _{\sigma_{h}^{x}(\theta)}\right).\label{eq:gradient-cq-BM-petz-tsallis}
\end{multline}
\end{thm}

\begin{proof}
Consider from \eqref{eq:reduction-2nd-term-cq-1}--\eqref{eq:reduction-2nd-term-cq-last}
that
\begin{align}
\left\langle G_{j}\right\rangle _{\sigma_{vh}(\theta)} & =\sum_{x}p_{x}(\theta)\left\langle G_{h}^{j,x}\right\rangle _{\sigma_{h}^{x}(\theta)}.
\end{align}
Furthermore, observe that
\begin{equation}
\sigma_{v}(\theta)^{-\frac{q}{2}}\Upsilon_{v}^{\theta,1-q}(\rho^{q})\sigma_{v}(\theta)^{-\frac{q}{2}}=\sum_{x}\frac{r_{x}^{q}}{p_{x}(\theta)^{q}}|x\rangle\!\langle x|_{v},
\end{equation}
implying that
\begin{align}
 & \Tr\!\left[\Phi_{vh}^{\theta}(G_{j})\frac{1}{2}\left\{ \sigma_{vh}(\theta),\sigma_{v}(\theta)^{-\frac{q}{2}}\Upsilon_{v}^{\theta,1-q}(\rho^{q})\sigma_{v}(\theta)^{-\frac{q}{2}}\otimes I_{h}\right\} \right]\nonumber \\
 & =\Tr\!\left[\Phi_{vh}^{\theta}(G_{j})\frac{1}{2}\left\{ \sum_{x}p_{x}(\theta)|x\rangle\!\langle x|_{v}\otimes\sigma_{h}^{x}(\theta),\sum_{x}\frac{r_{x}^{q}}{p_{x}(\theta)^{q}}|x\rangle\!\langle x|_{v}\otimes I_{h}\right\} \right]\\
 & =\Tr\!\left[\Phi_{vh}^{\theta}\!\left(\sum_{x}|x\rangle\!\langle x|_{v}\otimes G_{h}^{j,x}\right)\sum_{x}r_{x}^{q}p_{x}(\theta)^{1-q}|x\rangle\!\langle x|_{v}\otimes\sigma_{h}^{x}(\theta)\right]\\
 & =\int_{-\infty}^{\infty}dt\,\gamma(t)\:\Tr\!\left[\begin{array}{c}
\left(\sum_{x}|x\rangle\!\langle x|_{v}\otimes e^{-iG_{h}^{x}(\theta)t}G_{h}^{j,x}e^{iG_{h}^{x}(\theta)t}\right)\times\\
\sum_{x}r_{x}^{q}p_{x}(\theta)^{1-q}|x\rangle\!\langle x|_{v}\otimes\sigma_{h}^{x}(\theta)
\end{array}\right]\\
 & =\int_{-\infty}^{\infty}dt\,\gamma(t)\:\sum_{x}r_{x}^{q}p_{x}(\theta)^{1-q}\Tr\!\left[e^{-iG_{h}^{x}(\theta)t}G_{h}^{j,x}e^{iG_{h}^{x}(\theta)t}\sigma_{h}^{x}(\theta)\right]\\
 & =\int_{-\infty}^{\infty}dt\,\gamma(t)\:\sum_{x}r_{x}^{q}p_{x}(\theta)^{1-q}\Tr\!\left[G_{h}^{j,x}e^{iG_{h}^{x}(\theta)t}\sigma_{h}^{x}(\theta)e^{-iG_{h}^{x}(\theta)t}\right]\\
 & =\sum_{x}r_{x}^{q}p_{x}(\theta)^{1-q}\Tr\!\left[G_{h}^{j,x}\sigma_{h}^{x}(\theta)\right]\\
 & =\sum_{x}r_{x}^{q}p_{x}(\theta)^{1-q}\left\langle G_{h}^{j,x}\right\rangle _{\sigma_{h}^{x}(\theta)},
\end{align}
thus concluding the proof.
\end{proof}

\subsection{Quantum algorithms for estimating the gradients}

One can also design quantum algorithms for estimating the gradients
in \eqref{eq:gradient-fully-QBM-petz-tsallis} and \eqref{eq:gradient-qc-BM-petz-tsallis},
following the approach given in Sections \ref{subsec:Q-algorithm-state-learning-QBM}
and \ref{subsec:Quantum-algorithm-qc-BM}, respectively. In the Petz--Tsallis
case, however, $\rho^{q}$ appears in the expressions for the gradients.
In order to handle the term $\rho^{q}$, one can employ a block-encoding
of it, by means of sample access to the target state $\rho$. As indicated
previously, one can realize a block-encoding of $\rho^{q}$ from sample
access to $\rho$ by means of density matrix exponentiation and QSVT
\cite[Corollary~21]{Gilyen2022a} (see also \cite[Lemma~2.21]{Wang2025}).
Another modification needed is to realize the transformation $(\cdot)\to\sigma_{v}(\theta)^{-\frac{q}{2}}(\cdot)\sigma_{v}(\theta)^{-\frac{q}{2}}$.
This is also possible by QSVT, with performance comparable to the
transformation when $q=1$ (see, e.g., \cite[Lemmas~27 and 36]{Liu2025}).
Thus, the performance of a QSVT-based algorithm for estimating the
gradients in \eqref{eq:gradient-fully-QBM-petz-tsallis} and \eqref{eq:gradient-qc-BM-petz-tsallis}
is comparable to that for Algorithm \ref{alg:estimate-grad-QBM},
but the main drawback is the extra overhead in realizing block-encodings
of $\rho^{q}$, so that it seems preferable to focus on the $q=1$
case, as done in earlier sections of this paper.

\section{Conclusion}

\label{sec:Conclusion}In conclusion, this paper lays a foundation
for quantum state learning with quantum Boltzmann machines that have
both visible and units, by providing an analytical formula for the
gradient (Theorem~\ref{thm:q-state-learning-gradient-vh-q}) and
a quantum algorithm for estimating it (Section~\ref{subsec:Q-algorithm-state-learning-QBM}).
I also considered two special cases of a fully quantum Boltzmann machine,
which include quantum--classical and classical--quantum Boltzmann
machines (those with quantum visible units and classical hidden units,
as well as classical visible units and quantum hidden units, respectively).
The gradient formula in Theorem~\ref{thm:q-state-learning-gradient-vh-q}
generalizes the classical formula in \eqref{eq:classical-gradient},
and it incorporates non-commutativity by means of modular-flow-generated
unitary rotations, similar to my prior work on the Petz recovery map
in quantum relative entropy inequalities \cite{Wilde2015} (see also
\cite{Junge2018} in this context). I also showed how to generalize
all of the findings to the case when the objective function to be
minimized is the Petz--Tsallis relative entropy, which relies on
an independent derivation of a formula for the derivative of the matrix
power function (see \eqref{eq:deriv-matrix-power-fourier}).

Going forward from here, the most pressing open question is to develop
a fully quantum contrastive divergence algorithm for training quantum
Boltzmann machines, as a generalization of the well known method \cite{Hinton2002}
used for training restricted classical Boltzmann machines. The algorithm
put forward in Section~\ref{subsec:Q-algorithm-state-learning-QBM}
can be used for training quantum Boltzmann machines, but it seems
plausible that a quantum contrastive divergence approach, should it
be found, would ultimately be more efficient and thus more useful
in practice. At the least, it would be ideal to eliminate the dependence
of the runtime of a training algorithm on $\kappa$, as is the case
with Algorithm~\ref{alg:estimate-grad-QBM}. There has been recent
progress on developing contrastive-divergence and related training
algorithms for classical--quantum Boltzmann machines \cite{Vishnu2025,Kimura2025,Demidik2025a},
but the approaches used there are not obviously amenable to the fully
quantum case, as any such algorithm would likely need to make use
of the gradient formulas put forward in the present paper (i.e., Theorems
\ref{thm:q-state-learning-gradient-vh-q} and \ref{thm:qc-vh-gradient}).
I have also not addressed the more general model of an evolved quantum
Boltzmann machine \cite{Minervini2025} and instead leave this for
future analysis.\medskip{}

\textit{Note added}---After the first version of this paper appeared
on the arXiv as \cite{Wilde2025b}, the authors of \cite{Beigi2025}
notified me of their paper, which independently derived \eqref{eq:deriv-matrix-power-fourier}
using a different method.\medskip{}

\textit{Acknowledgements}---I acknowledge insightful discussions
with Marcello Benedetti, Greeshma Oruganti, Yixian Qiu, and Zhicheng
Zhang. I also acknowledge helpful discussions with all participants
of the International Workshop on Quantum Boltzmann Machines, which
helped motivate the current paper. I am grateful for support from
the National Science Foundation under grant nos.~2329662 and 2611810 and from
the Cornell School of Electrical and Computer Engineering.

\footnotesize

\bibliographystyle{alphaurl}
\bibliography{Ref}

\normalsize

\appendix

\section{Proof of Equation \eqref{eq:deriv-matrix-power-fourier}}

\label{app:Proof-of-Equation-matrix-power}By applying a well known
identity (see, e.g., \cite[Theorem~42]{Wilde2025}), the following
equality holds:
\begin{equation}
\frac{\partial}{\partial x}A(x)^{r}=\sum_{k,\ell}f_{x^{r}}^{\left[1\right]}(\lambda_{k},\lambda_{\ell})\Pi_{k}\left(\frac{\partial}{\partial x}A(x)\right)\Pi_{\ell},
\end{equation}
where the spectral decomposition of $A(x)$ is given by $A(x)=\sum_{k}\lambda_{k}\Pi_{k}$,
with $\lambda_{k}$ an eigenvalue and $\Pi_{k}$ an eigenprojection
and I have suppressed the dependence of $\lambda_{k}$ and $\Pi_{k}$
on the parameter $x$. Furthermore, $f_{x^{r}}^{\left[1\right]}(x,y)$
is the first divided difference of the function $x\mapsto x^{r}$,
defined for $x,y>0$ as
\begin{equation}
f_{x^{r}}^{\left[1\right]}(x,y)\coloneqq\begin{cases}
rx^{r-1} & :x=y\\
\frac{x^{r}-y^{r}}{x-y} & :x\neq y
\end{cases}.
\end{equation}
Now consider that
\begin{align}
\frac{x^{r}-y^{r}}{x-y} & =\frac{\left(xy\right)^{\frac{r}{2}}}{\left(xy\right)^{\frac{1}{2}}}\frac{\left(\frac{x}{y}\right)^{\frac{r}{2}}-\left(\frac{x}{y}\right)^{-\frac{r}{2}}}{\left(\frac{x}{y}\right)^{\frac{1}{2}}-\left(\frac{x}{y}\right)^{-\frac{1}{2}}}\label{eq:div-diff-r-power}\\
 & =\left(xy\right)^{\frac{r-1}{2}}\frac{\left(\frac{x}{y}\right)^{\frac{r}{2}}-\left(\frac{x}{y}\right)^{-\frac{r}{2}}}{\left(\frac{x}{y}\right)^{\frac{1}{2}}-\left(\frac{x}{y}\right)^{-\frac{1}{2}}}\\
 & =\left(xy\right)^{\frac{r-1}{2}}\frac{e^{\frac{ru}{2}}-e^{-\frac{ru}{2}}}{e^{\frac{u}{2}}-e^{-\frac{u}{2}}}\\
 & =\left(xy\right)^{\frac{r-1}{2}}\frac{\sinh\!\left(\frac{ru}{2}\right)}{\sinh\!\left(\frac{u}{2}\right)},
\end{align}
where $u\coloneqq\ln\!\left(\frac{x}{y}\right)$. Observe that $\lim_{u\to0}\frac{\sinh\left(\frac{ru}{2}\right)}{\sinh\left(\frac{u}{2}\right)}=r$,
consistent with the fact that $f_{x^{r}}^{\left[1\right]}(y,y)=ry^{r-1}$.
By applying Lemma \ref{lem:contour-int-fourier-trans}, the inverse
Fourier transform of $\frac{\sinh\left(\frac{ru}{2}\right)}{\sinh\left(\frac{u}{2}\right)}$
is given by
\begin{align}
\frac{\sinh\!\left(\frac{ru}{2}\right)}{\sinh\!\left(\frac{u}{2}\right)} & =\int_{-\infty}^{\infty}dt\,g_{r}(t)e^{-\frac{iut}{2}}\label{eq:key-integral-power-deriv}\\
 & =\int_{-\infty}^{\infty}dt\,g_{r}(t)\left(\frac{x}{y}\right)^{-\frac{it}{2}},
\end{align}
where
\begin{equation}
g_{r}(t)\coloneqq\frac{\sin(\pi r)}{2\left(\cosh\!\left(\pi t\right)+\cos(\pi r)\right)}.\label{eq:g-r-func}
\end{equation}
This implies that
\begin{align}
 & \sum_{k,\ell}f_{x^{r}}^{\left[1\right]}(\lambda_{k},\lambda_{\ell})\Pi_{k}\left(\frac{\partial}{\partial x}A(x)\right)\Pi_{\ell}\nonumber \\
 & =\sum_{k,\ell}\left(\lambda_{k}\lambda_{\ell}\right)^{\frac{r-1}{2}}\int_{-\infty}^{\infty}dt\,g_{r}(t)\left(\frac{\lambda_{k}}{\lambda_{\ell}}\right)^{-\frac{it}{2}}\Pi_{k}\left(\frac{\partial}{\partial x}A(x)\right)\Pi_{\ell}\\
 & =\int_{-\infty}^{\infty}dt\,g_{r}(t)\left(\sum_{k}\lambda_{k}^{\frac{r-1}{2}}\lambda_{k}^{-\frac{it}{2}}\Pi_{k}\right)\left(\frac{\partial}{\partial x}A(x)\right)\left(\sum_{\ell}\lambda_{\ell}^{\frac{r-1}{2}}\lambda_{\ell}^{\frac{it}{2}}\Pi_{\ell}\right)\\
 & =\int_{-\infty}^{\infty}dt\,g_{r}(t)A(x)^{\frac{r-1}{2}}A(x)^{-\frac{it}{2}}\left(\frac{\partial}{\partial x}A(x)\right)A(x)^{\frac{it}{2}}A(x)^{\frac{r-1}{2}}\\
 & =A(x)^{\frac{r-1}{2}}\left(\int_{-\infty}^{\infty}dt\,g_{r}(t)A(x)^{-\frac{it}{2}}\left(\frac{\partial}{\partial x}A(x)\right)A(x)^{\frac{it}{2}}\right)A(x)^{\frac{r-1}{2}}\\
 & =rA(x)^{\frac{r-1}{2}}\left(\int_{-\infty}^{\infty}dt\,\beta_{r}(t)A(x)^{-\frac{it}{2}}\left(\frac{\partial}{\partial x}A(x)\right)A(x)^{\frac{it}{2}}\right)A(x)^{\frac{r-1}{2}},
\end{align}
where $\beta_{r}(t)$ is the probability density function defined
in \eqref{eq:deriv-matrix-power-fourier}. Observe that $\beta_{r}(t)$
is a probability density function because
\begin{equation}
\int_{-\infty}^{\infty}dt\,\beta_{r}(t)=\frac{1}{r}\int_{-\infty}^{\infty}dt\,g_{r}(t)=\frac{1}{r}\lim_{u\to0}\frac{\sinh\!\left(\frac{ru}{2}\right)}{\sinh\!\left(\frac{u}{2}\right)}=1,
\end{equation}
due to \eqref{eq:key-integral-power-deriv}.
\begin{lem}
\label{lem:contour-int-fourier-trans}For all $r\in\left(-1,0\right)\cup\left(0,1\right)$
and $u\in\mathbb{R}$, the following equality holds:
\begin{equation}
\frac{\sinh\!\left(\frac{ru}{2}\right)}{\sinh\!\left(\frac{u}{2}\right)}=\int_{-\infty}^{\infty}dt\,g_{r}(t)e^{-\frac{iut}{2}},
\end{equation}
where $g_{r}(t)$ is defined in \eqref{eq:g-r-func}.
\end{lem}

\begin{proof}
Defining $v=\frac{u}{2}$, the desired equality is equivalent to
\begin{equation}
\frac{\sinh\!\left(rv\right)}{\sinh\!\left(v\right)}=\int_{-\infty}^{\infty}dt\,g_{r}(t)e^{-ivt}.
\end{equation}
By applying the inverse Fourier transform, this is then equivalent
to
\begin{equation}
\frac{1}{2\pi}\int_{-\infty}^{\infty}dv\,\frac{\sinh\!\left(rv\right)}{\sinh\!\left(v\right)}e^{ivt}=g_{r}(t).
\end{equation}
We can evaluate the integral on the left-hand side by means of contour
integration and the Cauchy residue theorem. Suppose that $t>0$, and
consider the positively oriented contour $\gamma_{R}^{+}$ depicted
in \cite[Figure~1]{Wilde2025}, where $R>0$ denotes the radius of
the semicircle depicted there. Consider that
\begin{multline}
\lim_{R\to\infty}\oint_{\gamma_{R}^{+}}dz\ \frac{\sinh\!\left(rz\right)}{\sinh\!\left(z\right)}e^{izt}=\\
\lim_{R\to\infty}\left[\int_{-R}^{R}dv\ \frac{\sinh\!\left(rv\right)}{\sinh\!\left(v\right)}e^{ivt}+\int_{0}^{\pi}d\theta\ \frac{\sinh\!\left(rRe^{i\theta}\right)}{\sinh\!\left(Re^{i\theta}\right)}e^{iRe^{i\theta}t}\right],
\end{multline}
such that the contour integral is broken up into the line integral
along the real axis and the line integral around the arc of the semicircle.
Let us now prove that the line integral around the arc of the semicircle
evaluates to zero in the limit $R\to\infty$. To this end, consider
that
\begin{align}
 & \left|\int_{0}^{\pi}d\theta\ \frac{\sinh\!\left(rRe^{i\theta}\right)}{\sinh\!\left(Re^{i\theta}\right)}e^{iRe^{i\theta}t}\right|\nonumber \\
 & =\left|\int_{0}^{\pi}d\theta\ \frac{\sinh\!\left(rRe^{i\theta}\right)}{\sinh\!\left(Re^{i\theta}\right)}e^{iR\cos(\theta)t}e^{-R\sin(\theta)t}\right|\\
 & \leq\int_{0}^{\pi}d\theta\ \left|\frac{\sinh\!\left(rRe^{i\theta}\right)}{\sinh\!\left(Re^{i\theta}\right)}e^{iR\cos(\theta)t}e^{-R\sin(\theta)t}\right|\\
 & \leq\int_{0}^{\pi}d\theta\ \left|\frac{\sinh\!\left(rRe^{i\theta}\right)}{\sinh\!\left(Re^{i\theta}\right)}\right|e^{-R\sin(\theta)t}\\
 & =\int_{0}^{\pi}d\theta\ \frac{\sqrt{\sinh^{2}\!\left(rR\cos(\theta)\right)+\sin^{2}(rR\sin(\theta))}}{\sqrt{\sinh^{2}\!\left(R\cos(\theta)\right)+\sin^{2}(R\sin(\theta))}}e^{-R\sin(\theta)t}\\
 & \leq\int_{0}^{\pi}d\theta\ \frac{\sqrt{\sinh^{2}\!\left(rR\cos(\theta)\right)+1}}{\sqrt{\sinh^{2}\!\left(R\cos(\theta)\right)}}e^{-R\sin(\theta)t}\\
 & =\int_{0}^{\pi}d\theta\ \frac{\cosh\!\left(rR\cos(\theta)\right)}{\left|\sinh\!\left(R\cos(\theta)\right)\right|}e^{-R\sin(\theta)t}\\
 & \leq\int_{0}^{\pi}d\theta\ e^{\left|r\right|R\left|\cos(\theta)\right|}\frac{1}{2}e^{-R\left|\cos(\theta)\right|}e^{-R\sin(\theta)t}\\
 & =\frac{1}{2}\int_{0}^{\pi}d\theta\ e^{R\left[\left(\left|r\right|-1\right)\left|\cos(\theta)\right|-\sin(\theta)t\right]}\label{eq:up-bnd-arc-semicirc}
\end{align}
where I employed the identity
\begin{equation}
\left|\sinh(x+iy)\right|=\sqrt{\sinh^{2}(x)+\sin^{2}(y)}
\end{equation}
and used the facts that
\begin{align}
\cosh(x) & \leq e^{\left|x\right|},\\
\left|\sinh(x)\right| & \geq\frac{1}{2}e^{\left|x\right|}.
\end{align}
Then the upper bound in \eqref{eq:up-bnd-arc-semicirc} and the dominated
convergence theorem imply that
\begin{align}
 & \lim_{R\to\infty}\left|\int_{0}^{\pi}d\theta\ \frac{\sinh\!\left(rRe^{i\theta}\right)}{\sinh\!\left(Re^{i\theta}\right)}e^{iRe^{i\theta}t}\right|\nonumber \\
 & \leq\lim_{R\to\infty}\frac{1}{2}\int_{0}^{\pi}d\theta\ e^{R\left[\left(\left|r\right|-1\right)\left|\cos(\theta)\right|-\sin(\theta)t\right]}\\
 & =\frac{1}{2}\int_{0}^{\pi}d\theta\ \lim_{R\to\infty}e^{R\left[\left(\left|r\right|-1\right)\left|\cos(\theta)\right|-\sin(\theta)t\right]}\\
 & =0,
\end{align}
where the last equality follows from the assumption that $\left|r\right|<1$.

It thus remains to evaluate the following contour integral:
\begin{equation}
\lim_{R\to\infty}\oint_{\gamma_{R}^{+}}dz\ \frac{\sinh\!\left(rz\right)}{\sinh\!\left(z\right)}e^{izt}.
\end{equation}
To this end, observe that the poles of the function $z\mapsto\frac{\sinh\left(rz\right)}{\sinh\left(z\right)}e^{izt}$
in the upper half of the complex plane are all simple poles and given
by $z_{n}=n\pi i$, where $n\in\left\{ 1,2,3,\ldots\right\} $. The
residues of this function are given by
\begin{align}
\text{Res}_{z=n\pi i}\!\left[\frac{\sinh\!\left(rz\right)}{\sinh\!\left(z\right)}e^{izt}\right] & =\lim_{z\to z_{n}}\left(z-z_{n}\right)\frac{\sinh\!\left(rz\right)}{\sinh\!\left(z\right)}e^{izt}\\
 & =\frac{\sinh\!\left(rz_{n}\right)}{\cosh\!\left(z_{n}\right)}e^{iz_{n}t}\\
 & =\frac{i\sin(rn\pi)}{\left(-1\right)^{n}}e^{-n\pi t}\\
 & =i\left(-1\right)^{n}\sin(rn\pi)e^{-n\pi t}.
\end{align}
Thus, by applying the Cauchy residue theorem, it follows that
\begin{align}
 & \frac{1}{2\pi}\lim_{R\to\infty}\oint_{\gamma_{R}^{+}}du\ \frac{\sinh\!\left(rz\right)}{\sinh\!\left(z\right)}e^{izt}\nonumber \\
 & =i\sum_{n=1}^{\infty}\text{Res}_{z=n\pi i}\!\left[\frac{\sinh\!\left(rz\right)}{\sinh\!\left(z\right)}e^{izt}\right]\\
 & =-\sum_{n=1}^{\infty}\left(-1\right)^{n}\sin(rn\pi)e^{-n\pi t}\\
 & =-\sum_{n=1}^{\infty}\left(-1\right)^{n}\left(\frac{e^{irn\pi}-e^{-irn\pi}}{2i}\right)e^{-n\pi t}\\
 & =-\frac{1}{2i}\sum_{n=1}^{\infty}\left(-e^{ir\pi-\pi t}\right)^{n}-\left(-e^{-ir\pi-\pi t}\right)^{n}\\
 & =-\frac{1}{2i}\left(\frac{-e^{ir\pi-\pi t}}{1+e^{ir\pi-\pi t}}-\frac{-e^{-ir\pi-\pi t}}{1+e^{-ir\pi-\pi t}}\right)\\
 & =-\frac{1}{2i}\left(\frac{-e^{ir\pi-\pi t}\left(1+e^{-ir\pi-\pi t}\right)+e^{-ir\pi-\pi t}\left(1+e^{ir\pi-\pi t}\right)}{\left(1+e^{ir\pi-\pi t}\right)\left(1+e^{-ir\pi-\pi t}\right)}\right)\\
 & =-\frac{1}{2i}\left(\frac{-e^{ir\pi-\pi t}-e^{-2\pi t}+e^{-ir\pi-\pi t}+e^{-2\pi t}}{1+e^{ir\pi-\pi t}+e^{-ir\pi-\pi t}+e^{-2\pi t}}\right)\\
 & =-\frac{1}{2i}\left(\frac{-e^{ir\pi-\pi t}+e^{-ir\pi-\pi t}}{1+e^{ir\pi-\pi t}+e^{-ir\pi-\pi t}+e^{-2\pi t}}\right)\\
 & =\frac{\sin(\pi r)e^{-\pi t}}{1+2\cos(\pi r)e^{-\pi t}+e^{-2\pi t}}\\
 & =\frac{\sin(\pi r)}{e^{\pi t}+2\cos(\pi r)+e^{-\pi t}}\\
 & =\frac{\sin(\pi r)}{2\left(\cosh(\pi t)+\cos(\pi r)\right)},
\end{align}
thus concluding the proof for $t>0$.

The desired integral holds for $t<0$ by a symmetric argument, instead
considering the negatively oriented contour depicted in Figure~2
of \cite[Figure~2]{Wilde2025}.
\end{proof}
\begin{rem}
\label{rem:div-diff-log-proof}The development in \eqref{eq:div-diff-r-power}--\eqref{eq:key-integral-power-deriv}
establishes that the following equality holds for all $r\in\left(-1,0\right)\cup\left(0,1\right)$
and $x,y>0$ such that $x\neq y$:
\begin{equation}
\frac{1}{r}\left(\frac{x^{r}-y^{r}}{x-y}\right)=\left(xy\right)^{\frac{r-1}{2}}\int_{-\infty}^{\infty}dt\,\beta_{r}(t)\left(\frac{x}{y}\right)^{-\frac{it}{2}}.
\end{equation}
By taking the limit $r\to0$ and applying the dominated convergence
theorem, it follows that
\begin{align}
\frac{\ln x-\ln y}{x-y} & =\lim_{r\to0}\frac{1}{r}\left(\frac{x^{r}-y^{r}}{x-y}\right)\\
 & =\lim_{r\to0}\left(xy\right)^{\frac{r-1}{2}}\int_{-\infty}^{\infty}dt\,\beta_{r}(t)\left(\frac{x}{y}\right)^{-\frac{it}{2}}\\
 & =\left(xy\right)^{-\frac{1}{2}}\int_{-\infty}^{\infty}dt\,\lim_{r\to0}\beta_{r}(t)\left(\frac{x}{y}\right)^{-\frac{it}{2}}\\
 & =\left(xy\right)^{-\frac{1}{2}}\int_{-\infty}^{\infty}dt\,\beta(t)\left(\frac{x}{y}\right)^{-\frac{it}{2}},
\end{align}
thus justifying \cite[Eq.~(97)]{Sutter2017}, which was stated therein
without justification.
\end{rem}

\section{Detailed analysis for QSVT subroutines}

\label{sec:Detailed-analysis-QSVT}

\subsection{Block-encoding for modular flow}

\label{subsec:Block-encoding-for-modular-flow}

In this appendix, I determine the cost of block-encodings for the
modular flow evolution $\rho\to\sigma^{-\frac{is}{2}}\rho\sigma^{\frac{is}{2}}$,
where $\rho$ is an input state, $\sigma$ is a density matrix, and
$s\in\mathbb{R}$. Note that this evolution is unitary. A quantum
algorithm for modular flow was recently proposed in \cite{Lim2025},
but the analysis below gives a slight improvement in its performance,
by using an alternative method of block-encoding $\ln\sigma_{v}$,
as given in \cite[Lemma~3]{Qiu2025}.

Here, I follow the blueprint of the approach detailed in \cite[Section~3]{Lim2025}.
Let us begin by observing that $\sigma^{-\frac{is}{2}}=e{}^{-\frac{is}{2}\ln\sigma}$.
Thus, we can start from a block-encoding of $\sigma$, transform it
to a block-encoding of $\ln\sigma$, and then transform this block-encoding
in turn to a block-encoding of $e^{-\frac{is}{2}\ln\sigma}$ by known
QSVT methods for Hamiltonian simulation.

Let us suppose that we have access to a block-encoding $U_{\sigma}$
of $\sigma$. Furthermore, suppose that $\kappa>0$ is such that $\frac{1}{\kappa}I\leq\sigma\leq I$.
By \cite[Lemma~3]{Qiu2025}, we can realize a $\left(2\left(1+\ln\!\left(2\kappa\right)\right),\varepsilon_{a}\right)$
block-encoding of $\ln\sigma$ by making
\begin{equation}
O\!\left(\kappa\ln\!\left(\frac{\ln\kappa}{\varepsilon_{a}}\right)\right)
\end{equation}
queries to $U_{\sigma}$. Now suppose that we have a Hamiltonian $H$,
a time $t$, an error tolerance $\varepsilon_{b}$, $\alpha_{H}\geq\text{\ensuremath{\left\Vert H\right\Vert }}$,
and an $\left(\alpha_{H},\frac{\varepsilon_{b}}{\left|t\right|}\right)$-approximate
block-encoding $U_{H}$ of $H$. We can implement Hamiltonian simulation
$e^{-iHt}$ by QSVT with the following number of queries to the block-encoding
$U_{H}$ \cite[Corollary~60 and Lemma~61]{Gilyen2019}:
\begin{equation}
\Theta\!\left(\alpha_{H}\left|t\right|+\frac{\ln\!\left(1/\varepsilon_{b}\right)}{\ln\!\left(e+\frac{\ln\left(1/\varepsilon_{b}\right)}{\alpha_{H}\left|t\right|}\right)}\right).
\end{equation}
We can then apply this result to the block-encoding of $\ln\sigma$,
while accounting for the scale factor $2\left(1+\ln\!\left(2\kappa\right)\right)=O\!\left(\ln\kappa\right)$,
so that the following number of queries to the block-encoding of $\ln\sigma$
are required for realizing the Hamiltonian evolution $\sigma^{-\frac{is}{2}}$:
\begin{equation}
O\!\left(\left|s\right|\ln\kappa+\frac{\ln\!\left(1/\varepsilon_{b}\right)}{\ln\!\left(e+\frac{\ln\left(1/\varepsilon_{b}\right)}{\left|s\right|\ln\kappa}\right)}\right).
\end{equation}
Thus, in order to have an overall error of $\varepsilon$ in the block-encoding
of $\sigma^{-\frac{is}{2}}$, by composing the two block-encodings,
we require the following number of queries to the block-encoding of
$\sigma$:
\begin{equation}
O\!\left(\left(\kappa\ln\!\left(\frac{\left|s\right|\ln\kappa}{\varepsilon}\right)\right)\left(\left|s\right|\ln\kappa+\frac{\ln\!\left(1/\varepsilon\right)}{\ln\!\left(e+\frac{\ln\left(1/\varepsilon\right)}{\left|s\right|\ln\kappa}\right)}\right)\right).
\end{equation}
The normalization factor realized by this block-encoding of $\sigma^{-\frac{is}{2}}$
is also equal to one. If we suppress logarithmic factors using the
notation $\tilde{O}$, then the total number of queries to the block-encoding
of $\sigma$ is given by
\begin{equation}
\tilde{O}\!\left(\kappa\left|s\right|\ln\!\left(\frac{\left|s\right|}{\varepsilon}\right)\ln\kappa\right),
\end{equation}
which represents a quadratic improvement in the dependence on $\kappa$,
when compared to \cite[Theorem~1]{Lim2025}. Also, note that, in the
application to quantum Boltzmann machines considered here, the time
$s$ is effectively constant, due to Remark~\ref{rem:constant-time-dists}.

\subsection{Block-encoding for inverse square root}

\label{subsec:Block-encoding-for-inverse-sqrt}Starting from a block
encoding of the density operator $\sigma$, it is known how to construct
a block-encoding of $\sigma^{-\frac{1}{2}}$ (see, e.g., \cite{Gilyen2022}
and \cite[Lemmas~27 and 36]{Liu2025}). The formal statement that
we have here is as follows: Let $U_{\sigma}$ be a $(1,0)$-block-encoding
of a density matrix $\sigma$, and suppose that $\kappa>0$ satisfies
$\kappa^{-1}I\leq\sigma$. Then for all $\varepsilon>0$, there exists
a quantum circuit using QSVT that implements an $(\sqrt{\kappa},\varepsilon)$-block-encoding
of $\sigma^{-\frac{1}{2}}$, using the following number of queries
to the block-encoding of $\sigma$:

\begin{equation}
O\!\left(\kappa\ln\!\left(\frac{1}{\varepsilon}\right)\right).
\end{equation}
This is achieved by constructing a polynomial approximation of the
function $f(x)=x^{-\frac{1}{2}}$ on the interval $[\delta,1]$ and
applying the standard QSVT block-encoding transformation.

\subsection{Products of block-encodings}

In this section, I prove Lemma \ref{lem:product-of-BEs}, which shows
how products of block-encodings compose. This lemma corrects a minor
flaw present in \cite[Lemma~30]{Gilyen2019}, which does not account
for an extra error term present in the composition.
\begin{lem}[Product of block-encoded matrices]
\label{lem:product-of-BEs}Let $U$ be an $\left(\alpha,a,\delta\right)$-block-encoding
of $A$, and let $V$ be a $\left(\beta,b,\varepsilon\right)$-block-encoding
of $B$. Then $\left(I_{b}\otimes U\right)\left(I_{a}\otimes V\right)$
is an $\left(\alpha\beta,a+b,\alpha\varepsilon+\beta\delta+\delta\varepsilon\right)$-block-encoding
of $AB$.
\end{lem}

\begin{proof}
By definition, we have that
\begin{align}
\left\Vert A-\tilde{A}\right\Vert  & \leq\delta,\\
\left\Vert B-\tilde{B}\right\Vert  & \leq\varepsilon,
\end{align}
where 
\begin{align}
\tilde{A} & \equiv\alpha\left(\langle0|^{\otimes a}\otimes I\right)U\left(|0\rangle^{\otimes a}\otimes I\right),\\
\tilde{B} & \equiv\beta\left(\langle0|^{\otimes b}\otimes I\right)V\left(|0\rangle^{\otimes b}\otimes I\right).
\end{align}
Given that $U$ and $V$ are unitaries, we conclude that
\begin{align}
\left\Vert \tilde{A}\right\Vert  & \leq\alpha,\\
\left\Vert \tilde{B}\right\Vert  & \leq\beta,\\
\left\Vert A\right\Vert  & =\left\Vert A-\tilde{A}+\tilde{A}\right\Vert \\
 & \leq\left\Vert A-\tilde{A}\right\Vert +\left\Vert \tilde{A}\right\Vert \\
 & \leq\delta+\alpha,\\
\left\Vert B\right\Vert  & =\left\Vert B-\tilde{B}+\tilde{B}\right\Vert \\
 & \leq\left\Vert B-\tilde{B}\right\Vert +\left\Vert \tilde{B}\right\Vert \\
 & \leq\varepsilon+\beta,
\end{align}
from which it follows that
\begin{align}
 & \left\Vert AB-\alpha\beta\left(\langle0|^{\otimes\left(a+b\right)}\otimes I\right)\left(I_{b}\otimes U\right)\left(I_{a}\otimes V\right)\left(|0\rangle^{\otimes\left(a+b\right)}\otimes I\right)\right\Vert \nonumber \\
 & =\left\Vert AB-\tilde{A}\tilde{B}\right\Vert \\
 & =\left\Vert AB-\tilde{A}B+\tilde{A}B-\tilde{A}\tilde{B}\right\Vert \\
 & =\left\Vert \left(A-\tilde{A}\right)B+\tilde{A}\left(B-\tilde{B}\right)\right\Vert \\
 & \leq\left\Vert A-\tilde{A}\right\Vert \left\Vert B\right\Vert +\left\Vert \tilde{A}\right\Vert \left\Vert B-\tilde{B}\right\Vert \\
 & \leq\delta\left(\varepsilon+\beta\right)+\alpha\varepsilon\\
 & =\alpha\varepsilon+\beta\delta+\delta\varepsilon,
\end{align}
thus concluding the proof.
\end{proof}
\begin{rem}
The minor flaw in \cite[Lemma~30]{Gilyen2019} is that therein $\left\Vert B\right\Vert $
is bounded from above by $\beta$, but as shown above, the assumptions
given only guarantee that it is bounded from above by $\varepsilon+\beta$.
\end{rem}

\subsection{Error analysis for Theorem \ref{thm:alg-complexity-claim}}

Let us begin by supposing that
\begin{align}
\left\Vert \alpha_{1}A_{v_{1}}-\sigma_{v}^{-is/2}\right\Vert  & \leq\varepsilon_{1},\\
\left\Vert \alpha_{2}B_{v_{1}}-\sigma_{v}^{-1/2}\right\Vert  & \leq\varepsilon_{2},
\end{align}
where
\begin{align}
A_{v_{1}} & \coloneqq\left(I_{v_{1}}\otimes\langle0|_{a_{1}}\right)U_{\sigma_{v}^{-is/2}}\left(I_{v_{1}}\otimes|0\rangle_{a_{1}}\right),\\
B_{v_{1}} & \coloneqq\left(I_{v_{1}}\otimes\langle0|_{a_{2}}\right)U_{\sigma_{v}^{-1/2}}\left(I_{v_{1}}\otimes|0\rangle_{a_{2}}\right),
\end{align}
so that $U_{\sigma_{v}^{-is/2}}$ and $U_{\sigma_{v}^{-1/2}}$ are
block-encodings of $\sigma_{v}^{-is/2}$ and $\sigma_{v}^{-1/2}$
with normalization factors $\alpha_{1}$ and $\alpha_{2}$, respectively.
Since $\sigma_{v}^{-is/2}$ is unitary, we can choose $\alpha_{1}=1$.
We can also choose $\alpha_{2}=\sqrt{\frac{1}{\lambda_{\min}(\sigma_{v})}}=\sqrt{\kappa}$.
By applying Lemma \ref{lem:product-of-BEs}, it follows that $U_{\sigma_{v}^{-1/2}}U_{\sigma_{v}^{-is/2}}$
is an $\left(\alpha_{1}\alpha_{2},\alpha_{1}\varepsilon_{2}+\alpha_{2}\varepsilon_{1}+\varepsilon_{1}\varepsilon_{2}\right)$-block-encoding
of $\sigma_{v}^{-1/2}\sigma_{v}^{-is/2}$; i.e.,
\begin{align}
\left\Vert \alpha_{1}\alpha_{2}B_{v_{1}}A_{v_{1}}-\sigma_{v}^{-1/2}\sigma_{v}^{-is/2}\right\Vert  & \leq\alpha_{1}\varepsilon_{2}+\alpha_{2}\varepsilon_{1}+\varepsilon_{1}\varepsilon_{2}\\
 & =\varepsilon_{2}+\sqrt{\kappa}\varepsilon_{1}+\varepsilon_{1}\varepsilon_{2}\\
 & =\varepsilon_{2}+\varepsilon_{1}\left(\sqrt{\kappa}+\varepsilon_{2}\right).\label{eq:up-bnd-specific-BE-prod}
\end{align}

Recalling the state $\omega_{v_{1}a_{1}a_{2}v_{2}h}$ defined in \eqref{eq:omega-state-alg}
and the observable $O_{v_{1}a_{1}a_{2}v_{2}h}$ defined in \eqref{eq:observable-all-sys-alg},
then consider that
\begin{align}
 & \Tr\!\left[\left(X_{c}\otimes O_{v_{1}a_{1}a_{2}v_{2}h}\right)\left(\text{c-}F_{v_{1}v_{2}}\right)\left(|+\rangle\!\langle+|_{c}\otimes\omega_{v_{1}a_{1}a_{2}v_{2}h}\right)\left(\text{c-}F_{v_{1}v_{2}}\right)^{\dag}\right]\nonumber \\
 & =\frac{1}{2}\Tr\!\left[\left\{ |0\rangle\!\langle0|_{a_{1}}\otimes|0\rangle\!\langle0|_{a_{2}}\otimes\left(e^{iG(\theta)t}G_{j}e^{-iG(\theta)t}\right)_{v_{2}h},F_{v_{1}v_{2}}\right\} \omega_{v_{1}a_{1}a_{2}v_{2}h}\right]\\
 & =\frac{1}{2}\Tr\!\left[\left\{ \left(e^{iG(\theta)t}G_{j}e^{-iG(\theta)t}\right)_{v_{2}h},F_{v_{1}v_{2}}\right\} \langle0|_{a_{1}}\langle0|_{a_{2}}\omega_{v_{1}a_{1}a_{2}v_{2}h}|0\rangle_{a_{1}}|0\rangle_{a_{2}}\right]\\
 & =\frac{1}{2}\Tr\!\left[\left\{ \left(e^{iG(\theta)t}G_{j}e^{-iG(\theta)t}\right)_{v_{2}h},F_{v_{1}v_{2}}\right\} \left(B_{v_{1}}A_{v_{1}}\rho_{v_{1}}A_{v_{1}}^{\dag}B_{v_{1}}^{\dag}\otimes\sigma_{v_{2}h}\right)\right]\\
 & =\frac{1}{2}\Tr\!\left[F_{v_{1}v_{2}}\left\{ \left(e^{iG(\theta)t}G_{j}e^{-iG(\theta)t}\right)_{v_{2}h},B_{v_{1}}A_{v_{1}}\rho_{v_{1}}A_{v_{1}}^{\dag}B_{v_{1}}^{\dag}\otimes\sigma_{v_{2}h}\right\} \right]\\
 & =\frac{1}{2}\Tr\!\left[F_{v_{1}v_{2}}\left(B_{v_{1}}A_{v_{1}}\rho_{v_{1}}A_{v_{1}}^{\dag}B_{v_{1}}^{\dag}\otimes\left\{ \left(e^{iG(\theta)t}G_{j}e^{-iG(\theta)t}\right)_{v_{2}h},\sigma_{v_{2}h}\right\} \right)\right]\\
 & =\frac{1}{2}\Tr\!\left[B_{v}A_{v}\rho_{v}A_{v}^{\dag}B_{v}^{\dag}\left\{ \left(e^{iG(\theta)t}G_{j}e^{-iG(\theta)t}\right)_{vh},\sigma_{vh}\right\} \right]\\
 & =\frac{1}{2}\Tr\!\left[\left(e^{iG(\theta)t}G_{j}e^{-iG(\theta)t}\right)_{vh}\left\{ B_{v}A_{v}\rho_{v}A_{v}^{\dag}B_{v}^{\dag},\sigma_{vh}\right\} \right].
\end{align}
In the above, I used the swap-trick identity $\Tr[F(X\otimes Y)]=\Tr[XY]$.
The quantity in the last line is what we can estimate on a quantum
computer. We then multiply the estimate by $\left(\alpha_{1}\alpha_{2}\right)^{2}=\kappa$.
By Hoeffding's inequality, to have $\varepsilon$ error in the estimate
of the scaled quantity with $\delta$ failure probability, we require
the following number of repetitions of the quantum circuit in Figure
\ref{fig:q-circuit-gradient-est}:
\begin{equation}
O\!\left(\frac{\left\Vert G_{j}\right\Vert ^{2}\left(\alpha_{1}\alpha_{2}\right)^{4}}{\varepsilon^{2}}\ln\!\left(\frac{1}{\delta}\right)\right)=O\!\left(\left(\frac{\kappa\left\Vert G_{j}\right\Vert }{\varepsilon}\right)^{2}\ln\!\left(\frac{1}{\delta}\right)\right).\label{eq:hoeffding-samples}
\end{equation}
In what follows, I make use of the H\"older inequality:
\begin{equation}
\left\Vert XYZ\right\Vert _{1}\leq\left\Vert X\right\Vert \left\Vert Y\right\Vert _{1}\left\Vert Z\right\Vert .
\end{equation}
Now observe that
\begin{align}
 & \left\Vert \left(\alpha_{1}\alpha_{2}\right)^{2}B_{v}A_{v}\rho_{v}A_{v}^{\dag}B_{v}^{\dag}-\sigma_{v}^{-1/2}\sigma_{v}^{-is/2}\rho_{v}\sigma_{v}^{is/2}\sigma_{v}^{-1/2}\right\Vert _{1}\nonumber \\
 & =\left\Vert \begin{array}{c}
\left(\alpha_{1}\alpha_{2}B_{v}A_{v}\right)\rho_{v}\left(\alpha_{1}\alpha_{2}B_{v}A_{v}\right)^{\dag}-\left(\alpha_{1}\alpha_{2}B_{v}A_{v}\right)\rho_{v}\left(\sigma_{v}^{-1/2}\sigma_{v}^{-is/2}\right)^{\dag}\\
+\left(\alpha_{1}\alpha_{2}B_{v}A_{v}\right)\rho_{v}\left(\sigma_{v}^{-1/2}\sigma_{v}^{-is/2}\right)^{\dag}-\sigma_{v}^{-1/2}\sigma_{v}^{-is/2}\rho_{v}\left(\sigma_{v}^{-1/2}\sigma_{v}^{-is/2}\right)^{\dag}
\end{array}\right\Vert _{1}\\
 & =\left\Vert \begin{array}{c}
\left(\alpha_{1}\alpha_{2}B_{v}A_{v}\right)\rho_{v}\left[\left(\alpha_{1}\alpha_{2}B_{v}A_{v}\right)^{\dag}-\left(\sigma_{v}^{-1/2}\sigma_{v}^{-is/2}\right)^{\dag}\right]\\
+\left[\left(\alpha_{1}\alpha_{2}B_{v}A_{v}\right)-\sigma_{v}^{-1/2}\sigma_{v}^{-is/2}\right]\rho_{v}\left(\sigma_{v}^{-1/2}\sigma_{v}^{-is/2}\right)^{\dag}
\end{array}\right\Vert _{1}\\
 & \leq\left\Vert \left(\alpha_{1}\alpha_{2}B_{v}A_{v}\right)\rho_{v}\left[\left(\alpha_{1}\alpha_{2}B_{v}A_{v}\right)^{\dag}-\left(\sigma_{v}^{-1/2}\sigma_{v}^{-is/2}\right)^{\dag}\right]\right\Vert _{1}\\
 & \qquad+\left\Vert \left[\left(\alpha_{1}\alpha_{2}B_{v}A_{v}\right)-\sigma_{v}^{-1/2}\sigma_{v}^{-is/2}\right]\rho_{v}\left(\sigma_{v}^{-1/2}\sigma_{v}^{-is/2}\right)^{\dag}\right\Vert _{1}\\
 & \leq\left\Vert \alpha_{1}\alpha_{2}B_{v}A_{v}\right\Vert \left\Vert \rho_{v}\right\Vert _{1}\left\Vert \left(\alpha_{1}\alpha_{2}B_{v}A_{v}\right)^{\dag}-\left(\sigma_{v}^{-1/2}\sigma_{v}^{-is/2}\right)^{\dag}\right\Vert \\
 & \qquad+\left\Vert \left(\alpha_{1}\alpha_{2}B_{v}A_{v}\right)-\sigma_{v}^{-1/2}\sigma_{v}^{-is/2}\right\Vert \left\Vert \rho_{v}\right\Vert _{1}\left\Vert \left(\sigma_{v}^{-1/2}\sigma_{v}^{-is/2}\right)^{\dag}\right\Vert \\
 & \leq\alpha_{1}\alpha_{2}\left(\varepsilon_{2}+\varepsilon_{1}\left(\sqrt{\kappa}+\varepsilon_{2}\right)\right)+\left(\varepsilon_{2}+\varepsilon_{1}\left(\sqrt{\kappa}+\varepsilon_{2}\right)\right)\sqrt{\kappa}\\
 & =2\sqrt{\kappa}\left(\varepsilon_{2}+\varepsilon_{1}\left(\sqrt{\kappa}+\varepsilon_{2}\right)\right)\\
 & =2\sqrt{\kappa}\varepsilon_{2}+2\varepsilon_{1}\left(\kappa+\sqrt{\kappa}\varepsilon_{2}\right).\\
 & \equiv\varepsilon',\nonumber 
\end{align}
where I used that $\left\Vert \rho_{v}\right\Vert _{1}=1$, $\left\Vert \sigma_{v}^{-is/2}\right\Vert =1$,
and \eqref{eq:up-bnd-specific-BE-prod}. Then it follows that
\begin{multline}
\frac{\left(\alpha_{1}\alpha_{2}\right)^{2}}{2}\Tr\!\left[\left(e^{iG(\theta)t}G_{j}e^{-iG(\theta)t}\right)_{vh}\left\{ B_{v}A_{v}\rho_{v}A_{v}^{\dag}B_{v}^{\dag},\sigma_{vh}\right\} \right]\\
=\Tr\!\left[\frac{1}{2}\left\{ \left(e^{iG(\theta)t}G_{j}e^{-iG(\theta)t}\right)_{vh},\sigma_{vh}\right\} \left(\alpha_{1}\alpha_{2}\right)^{2}B_{v}A_{v}\rho_{v}A_{v}^{\dag}B_{v}^{\dag}\right]
\end{multline}
and
\begin{align}
 & \left|\begin{array}{c}
\Tr\!\left[\frac{1}{2}\left\{ \left(e^{iG(\theta)t}G_{j}e^{-iG(\theta)t}\right)_{vh},\sigma_{vh}\right\} \left(\alpha_{1}\alpha_{2}\right)^{2}B_{v}A_{v}\rho_{v}A_{v}^{\dag}B_{v}^{\dag}\right]\\
-\Tr\!\left[\frac{1}{2}\left\{ \left(e^{iG(\theta)t}G_{j}e^{-iG(\theta)t}\right)_{vh},\sigma_{vh}\right\} \sigma_{v}^{-1/2}\sigma_{v}^{-is/2}\rho_{v}\sigma_{v}^{is/2}\sigma_{v}^{-1/2}\right]
\end{array}\right|\nonumber \\
 & \leq\left\Vert \frac{1}{2}\left\{ \left(e^{iG(\theta)t}G_{j}e^{-iG(\theta)t}\right)_{vh},\sigma_{vh}\right\} \right\Vert \times\nonumber \\
 & \qquad\left\Vert \left(\alpha_{1}\alpha_{2}\right)^{2}B_{v}A_{v}\rho_{v}A_{v}^{\dag}B_{v}^{\dag}-\sigma_{v}^{-1/2}\sigma_{v}^{-is/2}\rho_{v}\sigma_{v}^{is/2}\sigma_{v}^{-1/2}\right\Vert _{1}\\
 & \leq\left\Vert G_{j}\right\Vert \varepsilon'.
\end{align}
Thus, from these calculations, it follows that the error in estimating
the desired quantity
\begin{equation}
\Tr\!\left[\frac{1}{2}\left\{ \left(e^{iG(\theta)t}G_{j}e^{-iG(\theta)t}\right)_{vh},\sigma_{vh}\right\} \sigma_{v}^{-1/2}\sigma_{v}^{-is/2}\rho_{v}\sigma_{v}^{is/2}\sigma_{v}^{-1/2}\right]
\end{equation}
is given by
\begin{align}
\varepsilon+\left\Vert G_{j}\right\Vert \varepsilon' & =\varepsilon+\left\Vert G_{j}\right\Vert \left[2\sqrt{\kappa}\varepsilon_{2}+2\varepsilon_{1}\left(\kappa+\sqrt{\kappa}\varepsilon_{2}\right)\right]\\
 & =\varepsilon+\left\Vert G_{j}\right\Vert O\left(\varepsilon_{1}\kappa+\varepsilon_{2}\sqrt{\kappa}\right).
\end{align}
So, to have $\varepsilon$ error in the estimate, it is necessary
to choose
\begin{equation}
\varepsilon_{1}=O\!\left(\frac{\varepsilon}{\kappa\left\Vert G_{j}\right\Vert }\right),\qquad\varepsilon_{2}=O\!\left(\frac{\varepsilon}{\sqrt{\kappa}\left\Vert G_{j}\right\Vert }\right).
\end{equation}
Thus, by combining with the statements in Appendices \ref{subsec:Block-encoding-for-modular-flow}
and \ref{subsec:Block-encoding-for-inverse-sqrt}, the number of queries
required to the block-encoding of $\sigma$ to realize $\sigma^{-\frac{is}{2}}$
is
\begin{equation}
\tilde{O}\!\left(\kappa\left|s\right|\ln\!\left(\frac{\kappa\left|s\right|\left\Vert G_{j}\right\Vert }{\varepsilon}\right)\right),
\end{equation}
while the number of queries to the block-encoding of $\sigma$ to
realize $\sigma^{-\frac{1}{2}}$ is
\begin{equation}
O\!\left(\kappa\ln\!\left(\frac{\sqrt{\kappa}\left\Vert G_{j}\right\Vert }{\varepsilon}\right)\right).
\end{equation}
Given that $\left|s\right|$ is effectively a constant, as mentioned
in Remark \ref{rem:constant-time-dists}, it follows that the query
complexity to the block-encoding of $\sigma$ is 
\begin{equation}
\tilde{O}\!\left(\kappa\ln\!\left(\frac{\kappa\left\Vert G_{j}\right\Vert }{\varepsilon}\right)\right),
\end{equation}
and multiplying this by the number of trials needed, as given in \eqref{eq:hoeffding-samples},
leads to the claim in \eqref{eq:total-num-queries-alg}, completing
the proof of Theorem \ref{thm:alg-complexity-claim}.
\end{document}